\documentclass{LMCS}

\def\dOi{9(4:20)2013}
\lmcsheading%
{\dOi}
{1--51}
{}
{}
{Feb.~26, 2013}
{Dec.~11, 2013}
{}

\ACMCCS{[{\bf Theory of computation}]: Semantics and reasoning---Program
semantics---Algebraic semantics}

\usepackage{amsfonts,amsthm,amssymb,amsmath,color}
\usepackage{hyperref}
\usepackage{comment,verbatim}

  \usepackage[curve,matrix,arrow,cmtip]{xy} 
  \UseComputerModernTips  \everyentry={\vphantom(}
  \newdir{>>>}{!/2pt/\dir{>>}*!/6.7pt/\dir{-}*!/6.7pt/\dir{>}}
  \newdir{<<<}{!/2pt/\dir{<<}*!/6.7pt/\dir{-}*!/6.7pt/\dir{<}}

\def \bchpaula {\begin{color}{blue}} 
\def \echpaula {\end{color}}

\definecolor{akcolor}{rgb}{1,0,1}
\def \bchak {\begin{color}{akcolor}} 
\def \echak {\end{color}}

\theoremstyle{definition}
\newtheorem{notation}[thm]{Notation}
\newdimen\proofrulebreadth \proofrulebreadth=.05em
\newdimen\proofdotseparation \proofdotseparation=1.25ex
\newdimen\proofrulebaseline \proofrulebaseline=2ex
\newcount\proofdotnumber \proofdotnumber=3
\let\then\relax
\def\hfi{\hskip0pt plus.0001fil}
\mathchardef\squigto="3A3B
\newif\ifinsideprooftree\insideprooftreefalse
\newif\ifonleftofproofrule\onleftofproofrulefalse
\newif\ifproofdots\proofdotsfalse
\newif\ifdoubleproof\doubleprooffalse
\let\wereinproofbit\relax
\newdimen\shortenproofleft
\newdimen\shortenproofright
\newdimen\proofbelowshift
\newbox\proofabove
\newbox\proofbelow
\newbox\proofrulename
\def\shiftproofbelow{\let\next\relax\afterassignment\setshiftproofbelow\dimen0 }
\def\shiftproofbelowneg{\def\next{\multiply\dimen0 by-1 }%
\afterassignment\setshiftproofbelow\dimen0 }
\def\setshiftproofbelow{\next\proofbelowshift=\dimen0 }
\def\setproofrulebreadth{\proofrulebreadth}

\def\prooftree{
\ifnum  \lastpenalty=1
\then   \unpenalty
\else   \onleftofproofrulefalse
\fi
\ifonleftofproofrule
\else   \ifinsideprooftree
        \then   \hskip.5em plus1fil
        \fi
\fi
\bgroup
\setbox\proofbelow=\hbox{}\setbox\proofrulename=\hbox{}%
\let\justifies\proofover\let\leadsto\proofoverdots\let\Justifies\proofoverdbl
\let\using\proofusing\let\[\prooftree
\ifinsideprooftree\let\]\endprooftree\fi
\proofdotsfalse\doubleprooffalse
\let\thickness\setproofrulebreadth
\let\shiftright\shiftproofbelow \let\shift\shiftproofbelow
\let\shiftleft\shiftproofbelowneg
\let\ifwasinsideprooftree\ifinsideprooftree
\insideprooftreetrue
\setbox\proofabove=\hbox\bgroup$\displaystyle 
\let\wereinproofbit\prooftree
\shortenproofleft=0pt \shortenproofright=0pt \proofbelowshift=0pt
\onleftofproofruletrue\penalty1
}

\def\eproofbit{
\ifx    \wereinproofbit\prooftree
\then   \ifcase \lastpenalty
        \then   \shortenproofright=0pt  
        \or     \unpenalty\hfil         
        \or     \unpenalty\unskip       
        \else   \shortenproofright=0pt  
        \fi
\fi
\global\dimen0=\shortenproofleft
\global\dimen1=\shortenproofright
\global\dimen2=\proofrulebreadth
\global\dimen3=\proofbelowshift
\global\dimen4=\proofdotseparation
\global\count255=\proofdotnumber
$\egroup  
\shortenproofleft=\dimen0
\shortenproofright=\dimen1
\proofrulebreadth=\dimen2
\proofbelowshift=\dimen3
\proofdotseparation=\dimen4
\proofdotnumber=\count255
}

\def\proofover{
\eproofbit 
\setbox\proofbelow=\hbox\bgroup 
\let\wereinproofbit\proofover
$\displaystyle
}%
\def\proofoverdbl{
\eproofbit 
\doubleprooftrue
\setbox\proofbelow=\hbox\bgroup 
\let\wereinproofbit\proofoverdbl
$\displaystyle
}%
\def\proofoverdots{
\eproofbit 
\proofdotstrue
\setbox\proofbelow=\hbox\bgroup 
\let\wereinproofbit\proofoverdots
$\displaystyle
}%
\def\proofusing{
\eproofbit 
\setbox\proofrulename=\hbox\bgroup 
\let\wereinproofbit\proofusing
\kern0.3em$
}

\def\endprooftree{
\eproofbit 
  \dimen5 =0pt
\dimen0=\wd\proofabove \advance\dimen0-\shortenproofleft
\advance\dimen0-\shortenproofright
\dimen1=.5\dimen0 \advance\dimen1-.5\wd\proofbelow
\dimen4=\dimen1
\advance\dimen1\proofbelowshift \advance\dimen4-\proofbelowshift
\ifdim  \dimen1<0pt
\then   \advance\shortenproofleft\dimen1
        \advance\dimen0-\dimen1
        \dimen1=0pt
        \ifdim  \shortenproofleft<0pt
        \then   \setbox\proofabove=\hbox{%
                        \kern-\shortenproofleft\unhbox\proofabove}%
                \shortenproofleft=0pt
        \fi
\fi
\ifdim  \dimen4<0pt
\then   \advance\shortenproofright\dimen4
        \advance\dimen0-\dimen4
        \dimen4=0pt
\fi
\ifdim  \shortenproofright<\wd\proofrulename
\then   \shortenproofright=\wd\proofrulename
\fi
\dimen2=\shortenproofleft \advance\dimen2 by\dimen1
\dimen3=\shortenproofright\advance\dimen3 by\dimen4
\ifproofdots
\then
        \dimen6=\shortenproofleft \advance\dimen6 .5\dimen0
        \setbox1=\vbox to\proofdotseparation{\vss\hbox{$\cdot$}\vss}%
        \setbox0=\hbox{%
                \advance\dimen6-.5\wd1
                \kern\dimen6
                $\vcenter to\proofdotnumber\proofdotseparation
                        {\leaders\box1\vfill}$%
                \unhbox\proofrulename}%
\else   \dimen6=\fontdimen22\the\textfont2 
        \dimen7=\dimen6
        \advance\dimen6by.5\proofrulebreadth
        \advance\dimen7by-.5\proofrulebreadth
        \setbox0=\hbox{%
                \kern\shortenproofleft
                \ifdoubleproof
                \then   \hbox to\dimen0{%
                        $\mathsurround0pt\mathord=\mkern-6mu%
                        \cleaders\hbox{$\mkern-2mu=\mkern-2mu$}\hfill
                        \mkern-6mu\mathord=$}%
                \else   \vrule height\dimen6 depth-\dimen7 width\dimen0
                \fi
                \unhbox\proofrulename}%
        \ht0=\dimen6 \dp0=-\dimen7
\fi
\let\doll\relax
\ifwasinsideprooftree
\then   \let\VBOX\vbox
\else   \ifmmode\else$\let\doll=$\fi
        \let\VBOX\vcenter
\fi
\VBOX   {\baselineskip\proofrulebaseline \lineskip.2ex
        \expandafter\lineskiplimit\ifproofdots0ex\else-0.6ex\fi
        \hbox   spread\dimen5   {\hfi\unhbox\proofabove\hfi}%
        \hbox{\box0}%
        \hbox   {\kern\dimen2 \box\proofbelow}}\doll%
\global\dimen2=\dimen2
\global\dimen3=\dimen3
\egroup 
\ifonleftofproofrule
\then   \shortenproofleft=\dimen2
\fi
\shortenproofright=\dimen3
\onleftofproofrulefalse
\ifinsideprooftree
\then   \hskip.5em plus 1fil \penalty2
\fi
}

\newcommand{\subs}{{\sf subs}}
\newcommand{\subsa}{\subs_{\alpha}}

\newcommand{\mapsubsext}{{\sf h}_{\subsa}}
\newcommand{\mapsubsl}{\mapsubsext}

\newcommand{\inl}{{\sf inl}}
\newcommand{\inr}{{\sf inr}}

\newcommand{\inbota}{{\sf inbot}}
\newcommand{\invara}{{\sf invar}}
\newcommand{\inabsa}{{\sf inabs}}
\newcommand{\inappa}{{\sf inapp}}

 \newcommand{\gbertreea}{{\sf g}_{\bertreea}}
 \newcommand{\glltreea}{{\sf g}_{\lltreea}}
 \newcommand{\gbohmtreea}{{\sf g}_{\bohmtreea}}

\newcommand{\finalcarrier}{D}
\newcommand{\finalmap}{\delta}
\newcommand{\carrierone}{\setone}
\newcommand{\mapone}{g}
\newcommand{\carriertwo}{\settwo}
\newcommand{\maptwo}{h}
\newcommand{\mapthree}{f}

\newcommand{\nat}{\Bbb N}

\newcommand{\alphaconv}{=_{\alpha}}
\newcommand{\alphaconvg}{\sim_{\alpha}}
\newcommand{\fresh}{\#}

\newcommand{\eqclass}[1]{{\sf #1}}
\newcommand{\eqclassa}[1]{[ #1]_{\alpha}}
\newcommand{\echiv}[1]{[{#1}]_\alpha}

 \newcommand{\abs}[2]{\langle #1 \rangle #2}
 \newcommand{\conc}{@}
 
\newcommand{\fv}{{\sf fv}}
\newcommand{\bv}{{\sf bv}}
\newcommand{\var}{{\sf var}}

\newcommand{\setvars}{\mathcal{ V}}

\newcommand{\fnom}[1]{{\sf {#1}}}
\newcommand{\FLambda}{{\sf L}_{\alpha}}
\newcommand{\TLambda}{{L}}
\newcommand{\LNom}{\FLambda}

\newcommand{\TLambdabot}{{Lb}}
\newcommand{\LNombot}{{\sf Lb}_{\alpha}}

\newcommand{\unfold}{{\sf unfold}}
\newcommand{\unfolda}{{\sf unfold}}

\newcommand{\opsf}{{\sf op}}
\newcommand{\opsfc}[3]{{\sf op}({\langle\overline{#1_1}\rangle}.{#3_1},\ldots,\langle\overline{{#1_#2}}\rangle.{#3_#2})}

\newcommand{\Abs}{\sf Abs}
\newcommand{\App}{\sf App}

\newcommand{\Set}{\mathsf{Set}}
\newcommand{\Nom}{{\sf Nom}}
\newcommand{\supp}{\mathsf{supp}}
\newcommand{\grp}{\frak{S}(\mathcal{V})}

\newcommand{\setone}{X}
\newcommand{\settwo}{Y}

\newcommand{\setthree}{W}
\newcommand{\setfour}{Z}
\newcommand{\setS}{S}
\newcommand{\elone}{u}
\newcommand{\eltwo}{v}
\newcommand{\elthree}{w}
\newcommand{\elfour}{z}
\newcommand{\f}[1]{{#1}_{\sf fs}}

\newcommand{\one}{\mathrel{\rightarrow}}

\newcommand{\onebeta}{\one_{\beta}}

\newcommand{\onebot}{\mathrel{\rightarrow_\bot}}

\newcommand{\betaa}{\beta_{\alpha}}
\newcommand{\onebetaa}{\one_{\betaa}}

\newcommand{\betaheada}{\beta_{1}}
\newcommand{\betaweakheada}{\beta_{2}}
\newcommand{\betatopa}{\beta_{3}}

\newcommand{\onebetaheada}{\one_{\betaheada}}
\newcommand{\onebetaweakheada}{\one_{\betaweakheada}}

\newcommand{\onebetatopa}{\one_{\betatopa}}
\newcommand{\finbetatopa}{\fin_{\betatopa}}

\newcommand{\fin}{\mathrel{\rightarrow\!\!\!\!\!\rightarrow}}
\newcommand{\finbeta}{\fin_{\beta}}

\newcommand{\finbetaheada}{\fin_{\betaheada}}
\newcommand{\finbetaweakheada}{\fin_{\betaweakheada}}

\newcommand{\many}{\mathrel{\rightarrow\!\!\!\!\!\rightarrow\!\!\!\!\!\rightarrow}}
\newcommand{\manybeta}{\many_{\beta}}

\renewcommand{\L}{\Lambda}
\newcommand{\Lbi}{{\Lambda}^\infty}
\newcommand{\Lbif}{\Lambda^\infty_{\rm f{}fv}}
\newcommand{\Lbifs}{\Lambda^\infty_{\sf fs}}
\newcommand{\Const}{{\mathcal C}}
\newcommand{\Lc}{\Lambda(\Const)}
\newcommand{\Lbic}{{\Lambda}(\Const)^\infty}

\newcommand{\Lboti}{\Lambda^\infty_\bot}
\newcommand{\Lbotif}{\Lambda^\infty_{\bot{\rm f{}fv}}}

\newcommand{\fcoalg}{{\sf T}}
\newcommand{\ialg}{{\sf I}}
\newcommand{\ialga}{{\sf I_\alpha}}
\newcommand{\tffv}{\Tiffv}
\newcommand{\tffvs}{P}

\newcommand{\Lbia}{{\Lambda}^\infty/{=_\alpha}}
\newcommand{\Lbaone}{\Lambda/\hspace{-0.4em} \alphaconv}

\newcommand{\dab}{d_\alpha^\infty}
\newcommand{\Lbifa}{\mathsf{\Lambda}^{\infty}_\alpha} 
 \newcommand{\Lbifaone}{\Lbif/  \!\! \alphaconv}

\newcommand{\Lbaonec}{\Lambda(\Const)/\hspace{-0.4em} \alphaconv}
\newcommand{\Lbifaonec}{\Lambda(\Const)^\infty_{\rm f{}fv}/  \!\! \alphaconv}
\newcommand{\Lcai}{(\Lambda(\Const)/  \!\! \alphaconv)^\infty}

\newcommand{\TSig}{{T_\Sigma}}
\newcommand{\TSiga}{{T_\Sigma}/\hspace{-0.4em} \alphaconv}
\newcommand{\TSigai}{({T_\Sigma}/\hspace{-0.4em} \alphaconv)^\infty}
\newcommand{\TSigi}{T_\Sigma^\infty}
\newcommand{\Tiffv}{(\TSigi)_{\rm f{}fv}}
\newcommand{\TSigia}{\tffv/\hspace{-0.4em} \alphaconv}
\newcommand{\Tfinal}{(\TSigi)_{\sf fs}}
\newcommand{\Tfinala}{\TSigai_{\sf fs}}

\newcommand{\clba}{(\Lbaone)^{\infty}}
\newcommand{\cflba}{(\Lbaone)^{\infty}_{\sf fs}}

\newcommand{\Lib}{\Lbi}
\newcommand{\Lb}{\Lambda}

\newcommand{\setU}{{\mathcal U}}

\newcommand{\ogre}{{\sf ogre}}
\newcommand{\ibv}{{\sf Pinfbv}}
\newcommand{\ibvnf}{{\sf infbv}}
\newcommand{\Vterm}{{\sf allfv}}
\newcommand{\Iterm}{{\sf  I}} 

\newcommand{\Omegaterm}{{\sf \Omega}}
\newcommand{\Yterm}{{\sf fix}}

\newcommand{\bohmtree}{{ BT}}
\newcommand{\lltree}{{LLT} }
\newcommand{\bertree}{{ BeT} }

\newcommand{\bohmtreea}{{\sf BT}_{\alpha} }
\newcommand{\lltreea}{{\sf LLT}_{\alpha} }
\newcommand{\bertreea}{{\sf BeT}_{\alpha} }

\newcommand{\setBT}{\mathcal{BT}}
\newcommand{\setLLT}{\mathcal {LLT}}
\newcommand{\setBerT}{ \mathcal{B}\it{e}\mathcal{T}}

\newcommand{\setBTa}{\setBT_{\alpha}}
\newcommand{\setLLTa}{\setLLT_{\alpha}}

\newcommand{\setBTffv}{(\setBT)_{\rm f{}fv}}
\newcommand{\setLLTffv}{(\setLLT)_{\rm f{}fv} }

\newcommand{\funsetBTa}{{\sf S}_{\sf BT}}
\newcommand{\funsetLLTa}{{\sf S}_{\sf LLT}}
\newcommand{\funsetBerTa}{{\sf S}_{\sf BerT}}

\newcommand{\hbis}{\sim_{\mathit{hnf}}}
\newcommand{\whbis}{\sim_{\mathit{whnf}}}

\newcommand{\ProofRule}[3]{
\prooftree
   {#1}  
\justifies 
   {#2}
 \using (#3)
\endprooftree
}

\newcommand{\rules}[2]{\mbox{$\frac%
    {\mbox{\normalsize \rule[-5pt]{0pt}{14pt} $#1$}}
    {\mbox{\normalsize \rule[0pt]{0pt}{10pt}$#2$}}$}}

\newcommand{\fgt}{U}
\newcommand{\fgtu}{U}

\newcommand{\pullbackcorner}[1][dr]{\save*!/#1-1.5pc/#1:(-1.2,1.2)@^{|-}\restore}

\newcommand{\SWpullbackcorner}[1][ur]{\save*!/#1+1.2pc/#1:(1,-1)@^{|-}\restore}

\newcommand{\amap}[1]{[-]^{(#1)}_\alpha}

\begin{document}

\title{Nominal Coalgebraic Data Types with Applications to Lambda Calculus}
\author[A.~Kurz]{Alexander Kurz}
\address{Department of Computer Science, University of Leicester\\
  Leicester, UK}
\email{\{ak155, ps65, fdv1\}@mcs.le.ac.uk, daniela.petrisan@gmail.com}

\author[D.~Petri\c{s}an]{Daniela Petri\c{s}an}
\address{\vspace{-18 pt}}

\author[P.\ Severi]{Paula Severi}
\address{\vspace{-18 pt}}

\author[F.-J.!de Vries]{Fer-Jan de Vries}
\address{\vspace{-18 pt}}

\begin{abstract}
  We investigate final coalgebras in nominal sets. This allows us to
  define types of infinite data with binding for which all
  constructions automatically respect alpha equivalence. We give
  applications to the infinitary lambda calculus.
\end{abstract}
\maketitle

\vspace{-12 pt}\tableofcontents\vspace{-12 pt}
\enlargethispage{2\baselineskip}

\section{Introduction}
We investigate types of infinite data with binding. A leading example
is the infinitary $\lambda$-calculus. To construct explicitly a
domain of infinitary $\lambda$-terms, one usually starts from finite
$\lambda$-terms and then applies two constructions: Metric completion
to obtain infinite terms and quotienting up to
$\alpha$-equivalence. Although each of these constructions appear to
be routine, we show that their combination is more subtle than one
might think at first. For example, one either needs to assume
uncountably many variables or find a solution to the problem that, in
the case of countably many variables, metric completion does not
commute with quotienting by $\alpha$-equivalence.
On the other hand, general principles suggest to construct infinitary
$\lambda$-terms more abstractly as final coalgebras in the category of
nominal sets. It allows us to treat infinitary data types with binding
in general and provides a principle of definition and proof by
coinduction.  We also show that the syntactic approach of
completing and quotienting agrees with the semantic approach of final
coalgebras in nominal sets.

To summarise, the paper contributes to coalgebra, to nominal sets, and
to the infinitary $\lambda$-calculus. To coalgebra, by showing that
as for coalgebras over sets also over nominal sets we obtain from
finality a definition and proof principle of coinduction. To nominal
sets, by investigating limits and introducing notions of safe maps and
bound variables. To the infinitary $\lambda$-calculus, by clarifying
the fundamental constructions in the case of countably many variables
and by showing that the informal reasoning with $\alpha$-equivalence
classes of infinitary $\lambda$-terms is indeed mathematically precise.

In the remainder of the introduction we outline the contents of the paper.

\bigskip\noindent\textbf{Nominal sets} were introduced in
\cite{gabb-pitt:lics99}, but see also \cite{Hofmann99,fiore:lics99} for
related proposals. Roughly speaking, a nominal set is a set $X$
equipped with an action of a set of permutations on some countably
infinite set $$\setvars$$ of `names' or `variables'. One can then
define the support of some $x\in\setone$ as the smallest set of
variables on which $x$ depends (we review the precise definitions in
Section~\ref{sec:nominal-sets}), thus giving an abstract account of
the `free variables' in $x$. It is characteristic of nominal sets that
all elements have finite support.
In other words, modelling syntax in nominal sets requires us to only
consider terms with finitely many free variables. But, and this is one
of the themes of this paper, it is possible to have terms with
infinitely many bound variables. (The question what may constitute an
abstract account of bound variables in nominal sets will be discussed
in Section~\ref{sec:limits-in-Nom}.)

\medskip\noindent\textbf{Variable binding} in nominal sets can be
described by a type constructor $\setone\mapsto [\setvars]\setone$
which can be understood as a quotient
$$\setvars\times\setone\to [\setvars]\setone$$
identifying elements $(v,x)$ up to $\alpha$-equivalence, that is, up
to renaming of the `bound' variable $v$. For example, whereas
$\lambda$-terms are given by the initial algebra of the functor
\begin{equation}\label{eq:TLambda}
  \TLambda \ \setone =\mathcal{V}+\mathcal{V}\times \setone+\setone \times \setone
\end{equation}
it was shown in \cite{gabb-pitt:lics99} that $\lambda$-terms up to
$\alpha$-equivalence are given by the initial algebra of the functor
\begin{equation}\label{eq:FLambda}
\FLambda \  \setone = \setvars+[\setvars]\setone+ \setone\times \setone\end{equation} 

\medskip\noindent\textbf{Alpha-structural recursion}
\cite{Pitts:asr,PittsAM:alpsri,PittsMGS2011notes} is the induction
principle that ensues from syntax as an initial algebra in the
category $\Nom$ of nominal sets. For example, the classic definition
of substitution in the $\lambda$-calculus \cite{bare:1984}
\renewcommand{\arraystretch}{1.5}
\begin{equation} \label{definition:informal:substitution}
  \begin{array}{rcll}
    y [x:=N] &  \ = \ & \left \{ \begin{array}{ll}
        N & \mbox{ if $y=x$}, \\
        y & \mbox{otherwise}, 
      \end{array} \right .
    \\
    (P Q)    [x:=N] &  = &       (P[x:=N] Q[x:=N]),    \\ 
    (\lambda y. P)   [x:=N]    &  = & \lambda y. (P [x:=N]) & 
    \mbox{ if $y \not \in \fv(N) \cup \{x\}$}.
  \end{array}
\end{equation}
is not an inductive definition in the usual sense. 
  Because of the side condition, substitution is only a partial function on raw terms.
But, as explained in detail in \cite{Pitts:asr},
\eqref{definition:informal:substitution} is an inductive definition
according to \eqref{eq:FLambda}. Moreover, \cite{Pitts:asr} 
establishes a general induction principle for inductively defined data
types with variable binding, explaining when partially defined
functions in $\Set$ give rise to totally defined functions in $\Nom$.

\medskip\noindent\textbf{Alpha-structural corecursion} introduced
in this paper is the analogue of $\alpha$-structural recursion for
coinductive datatypes. For example, in the study of the infinitary
$\lambda$-calculus \cite{kenn:klop:slee:vrie:1995b,KKSV97,KV03}, which
we review in Section~\ref{sec:inf-lambda}, one is interested in the
final coalgebra of $\FLambda$. We describe the corecursion principle
ensuing from final coalgebras in nominal sets and show that
\eqref{definition:informal:substitution} is indeed a coinductive
definition of substitution for infinitary $\lambda$-terms.

\medskip\noindent\textbf{Infinitely many free variables in a term}
(Section~\ref{sec:ilcinnom}), which do appear if we take the final
coalgebra of $\TLambda$ in sets, pose a problem. To see this, note
that \eqref{definition:informal:substitution} becomes an inductive
definition by choosing a suitable representative $\lambda y.P$ such
that $y \not \in \fv(N)\cup \{x\}$. This approach is not immediately
viable for the infinitary $\lambda$-calculus, because we may have terms that exhaust all the available
variables, so that we cannot find a fresh $y$.
For example, consider the infinite $\lambda$-term $\Vterm = x_0 (x_1 (x_2 (\ldots)))$ which contains all variables from $\setvars$.  In the following
$\beta$-step
\begin{equation}
  \label{eq:wrong-beta}
  (\lambda x_0 x_1. x_0 x_1) \Vterm \onebeta (\lambda x_1. x_0 x_1)
  [x_0 := \Vterm ]  
\end{equation}
we have that $x_1 \in \fv(\Vterm)$ and, therefore, the $x_1$ in $\lambda x_1. x_0 x_1$ should be
replaced by some fresh variable, which is impossible because $\Vterm$
contains all of them~\cite{Salibra01nonmodularityresults}.

\medskip\noindent\textbf{Restricting to finitely many free variables},
as opposed to allowing $\setvars$ to be uncountable, is the solution
adopted in this paper (but we will come back to infinitely many free variable in Section~\ref{section:infinitelymanyfreevariables}). That is, in our example, we will consider the
set
\begin{equation} \label{def:Lbif} \Lbif = \{ M \in \Lbi \mid \fv(M)
  \mbox{ is finite }\}
\end{equation}
of $\lambda$-terms with \emph{finitely many free variables}, avoiding
terms such as $\Vterm$.  On the one hand, finitely many free variables
are sufficient in order to capture the infinite normal forms of terms
representing programs, since the limit of an infinite
$\beta$-reduction sequence starting from a finite term has always a
finite number of free variables. On the other hand, restricting to
finitely many free variables has the advantage of allowing us to work
with nominal sets.

\medskip\noindent\textbf{Infinitely many bound variables} must be
allowed, since additional fresh variables may be needed at each
$\beta$-reduction step to avoid capture. For example, consider the
finite term $\ibv \equiv \Yterm (\lambda f x y. x y (f (x y)))$ which
has the following reduction sequence:
\[
\begin{array}{ll}
\ibv   
& \fin_{\beta} \lambda x  y. x y (\ibv (x  y )) \\
& \alphaconv  \lambda  x_0 x_1. x_0  x_1  (\ibv (x_0  x_1)) \\
& \fin_{\beta} \lambda  x_0 x_1. x_0  x_1  (\lambda  y. x_0  x_1    y  
                     (\ibv (x_0  x_1  y)))     \\ 
     & \alphaconv  \lambda x_0 x_1. x_0  x_1  
     (\lambda  x_2. x_0  x_1  x_2  (\ibv (x_0  x_1  x_2))) \\
  & \fin_{\beta}   
  \lambda x_0 x_1. x_0  x_1  
     (\lambda  x_2. x_0  x_1  x_2  (\lambda y. x_0 x_1 x_2 y (\ibv (x_0 x_1 x_2 y)))) \\
      & \alphaconv  \lambda x_0 x_1. x_0  x_1  
     (\lambda  x_2. x_0  x_1  x_2 
     (\lambda x_3. x_0 x_1 x_2 x_3 (\ibv (x_0 x_1 x_2 x_3))))\\
   &   \vdots 
\end{array}
\]
The limit of the above sequence is the infinite term:
\[
\ibvnf \equiv  \lambda x_0.  \lambda  x_1. x_0  x_1  
(\lambda  x_2. x_0  x_1  x_2 (\lambda x_3. x_0 x_1 x_2 x_3 (\ldots))).
\]
The term $\ibvnf$  has an infinite number of bound variables.
All the terms in its $\alpha$-equivalence class 
have an infinite number of bound variables.

\medskip\noindent\textbf{The different classes of $\lambda$-terms}
arising from the discussion above are summarised in the following
picture, which is one of the contributions of our work. Previous work
on infinitary $\lambda$-calculus either assumed uncountably many
variables or did not make the careful distinctions discussed below.
\begin{equation}
  \label{eq:classeslambda}
\begin{gathered}
  \xymatrix@R=15pt@C=30pt
{
\Lambda\  \ar@{>->}[r]\ar@{->>}^{}[dd]& \Lbi\ar@{->>}^{}[d]  &   \ \Lbif\ar@{->>}^{}[d]\ar@{_(->}[l]\\
& \Lbia\ar@{>->}[d] & \Lbifaone\ar@{->}^{\cong}[d]\ar@{}[l]|{(*)}\\
\Lbaone\ \ar@{>->}[r] & \clba &   \ \cflba  \ar@{_(->}[l] 
}
\end{gathered}
\end{equation}
In the diagram, $\Lambda$ denotes the set of finite
$\lambda$-terms. Vertical arrows $\xymatrix{\ar@{->>}[r]&}$ denote
quotienting by $\alpha$-equivalence. Infinitary $\lambda$-terms are
constructed by metric completions $\xymatrix{\ar@{>->}[r]&}$.  The
rightmost column arises from maps $\xymatrix{&\ar@{_(->}[l]}$ that
restrict to terms with finitely many free variables.  Going
first right and then down in the diagram means to first complete to
infinitary terms and then to quotient by $\alpha$-equivalence, whereas
going first down and then right, means to first quotient and then to
complete. If both ways of constructing infinitary terms up-to
$\alpha$-equivalence coincide, then we say that \emph{metric
  completion commutes with quotienting by $\alpha$-equivalence.} The
two main results here are the following.
\begin{itemize}
\item The vertical map in the middle column $\Lbi\to\clba$ is not onto
  (Example~\ref{ex:alpha-inf-lam}), hence metric completion and
  quotienting by $\alpha$-equivalence do not commute for terms with
  countably many free variables (as opposed to the case of uncountably many variables, see Theorem~\ref{thm:uncountable-vars}).
\item The vertical map $\Lbif\to\cflba$ in the right-hand column is
  onto \cite[Theorem 22]{kurzetal:cmcs2012}, in other words,
  restricted to terms with finitely many free variables, the two
  operations of metric completion and quotienting by
  $\alpha$-equivalence do commute.
\end{itemize}

\medskip\noindent\textbf{Nominal coalgebraic datatypes for a binding
  signature} (Section~\ref{sec:nominalcodatatypes}) generalise 
\eqref{eq:classeslambda} to the diagram below (where we omitted the
middle row obtained from epi-mono factorisations).
\begin{equation}
  \label{eq:classesbindsig}
 \begin{gathered}
  \xymatrix@C=30pt
{
\TSig\  \ar@{>->}[r]\ar@{->>}^{}[d]& \TSigi\ar@{->}^{}[d]  &   \ \tffv\ar@{->>}^{}[d]\ar@{_(->}[l]\ar@{}[dl]|{(*)}\\
\TSiga\ \ar@{>->}[r] & (\TSiga)^\infty\ &   \ \Tfinala
\ar@{_(->}[l] \\
}
\end{gathered}
\end{equation}
$\Sigma$ is a so-called binding signature \cite{fiore:lics99} and $\TSig$
and $\TSiga$ are initial algebras. The middle column of metric
completions arises via unique arrows $\xymatrix{\ar@{>->}[r]&}$ from
an $\omega$-colimit into an $\omega^\textit{op}$-limit as in
\cite[Proposition~3.1]{barr99}. In the right-hand column,
$\Tfinala$ is the final coalgebra in $\Nom$ and $\tffv$ can be
defined as making the right-hand square into a pullback. The theorem
that metric completion commutes with quotienting by
$\alpha$-equivalence then follows from one of our main technical
contributions, namely that pulling back the (not necessarily
surjective) middle vertical arrow $\TSigi\to(\TSiga)^\infty$ along $\xymatrix{(\TSiga)^\infty\ &   \ \Tfinala
\ar@{_(->}[l]}$ yields a
surjection $\tffv\to\Tfinala$.

\medskip\noindent\textbf{Representing limits in nominal sets}
(Section~\ref{sec:limits-in-Nom})  provides the setting which
enables us to give a semantic proof of the result discussed in the
previous paragraph. To see the connection, denote by
$\fgt:\Nom\to\Set$ the forgetful functor and use the (well-known)
result that the middle column of 
\eqref{eq:classesbindsig} arises as limits of $\omega^\textit{op}$-chains as
depicted in
\begin{equation}
  \label{eq:}
 \begin{gathered}
  \xymatrix
{
  \fgtu\setone_0\ar@{<-}[r]\ar@{->>}[d] & \fgtu\setone_1 \ar@{<-}[r]\ar@{->>}[d] & \fgtu\setone_2 \ar@{<-}[r]\ar@{->>}[d] & \cdots  \ar@{<-}[r] & \lim\fgtu\setone_n\ar@{->}[d]\ar@{<-^)}[r]&P \ar@{}[dl]|{(*)}\ar@{->>}[d] \\
\fgtu\settwo_0\ar@{<-}[r] & \fgtu\settwo_1 \ar@{<-}[r] & \fgtu\settwo_1 \ar@{<-}[r] & \cdots  \ar@{<-}[r] & \lim\fgtu\settwo_n\ar@{<-^)}[r]&\fgtu\lim\settwo_n
}
\end{gathered}
\end{equation}
Similarly, the bottom right-hand corner of 
\eqref{eq:classesbindsig} is given by the limit in $\Nom$ of the
lower sequence.  The question whether metric completion commutes with
$\alpha$-equivalence now becomes an instance of a more general
question. Given an $\omega^\textit{op}$-sequence of surjections $X_n\to
Y_n$ in $\Nom$, can the limit $\lim\settwo_n$ in $\Nom$ be represented
by a surjection $P\to \fgtu\lim\settwo_n$, where $P$ is defined to be
the pullback $(*)$?  A careful analysis of this situation is carried
out in Section~\ref{sec:limits-in-Nom}. In particular, the notions of
safe squares, safe maps and of the bound variables relative to a map
are introduced and it is shown that $P\to \fgtu\lim\settwo_n$ is onto
if all vertical maps and squares in the chain are safe. We also explore the relationship
between safe maps and maps with orbit-finite fibres. 

\medskip\noindent\textbf{As  applications}, we 
    give 
a general definition of substitution on     
the final coalgebra  coming from a binding signature.
We also give 
    corecursive definitions of various notions
    of infinite normal form (B\"ohm, L\'evy-Longo and Berarducci
    trees) on $\alpha$-equivalence classes of terms.
We also show a solution, suggested to us by Pitts, of how to 
treat  infinitely many free variables in nominal sets.

\medskip\noindent\textbf{Related Work.} 
  This paper generalises \cite[Theorem 22]{kurzetal:cmcs2012}
  from the particular functor describing $\lambda$-calculus to
  arbitrary binding signatures. Along with this generalisation we
  replaced the syntactic proof (depending on a concrete presentation
  of the functor) of \cite[Theorem 22]{kurzetal:cmcs2012} with a
  semantic argument for the generalised
  Theorem~\ref{thm:comp-alpha-comm} of this paper. In particular, the
  new material in Section~\ref{sec:limits-in-Nom} allows us to show
  that all elements of the final coalgebra are presented by infinite
  terms with finitely many free variables, using only semantic (that is,
  category theoretic) properties of the $\Nom$-endofunctors of \eqref{eq:gramF}.

\section{Preliminaries on Algebra and Coalgebra}
\label{sec:alg-coalg}

\newcommand{\ccal}{\mathcal{C}}

\textbf{Finite data types, or \emph{algebraic data types}}, can be
studied as initial algebras for functors on different
categories. Consider an endofunctor $F$ on a category $\ccal$ and an
object $X$ of $\ccal$. An $F$-algebra $(X,\alpha)$ with carrier $X$ is
a $\ccal$-morphism $\alpha:FX\to X$. Given two $F$-algebras
$(X,\alpha)$ and $(Y,\beta)$ an $F$-algebra morphism is an arrow
$f:X\to Y$ such that $f\circ \alpha=\beta\circ Ff$. The $F$-algebras
thus form a category. One can prove the existence of an initial object in this
category under certain assumptions on the cocompleteness 
of $\ccal$ and on the ``size" of the
functor $F$. Assume that $\ccal$ has colimits. Then we can consider the
following sequence, starting from the initial $\ccal$-object $0$.
\begin{equation}
  \label{eq:in-chain-prel}
  \xymatrix{
0\ar[r]^{!}&F0\ar[r]^{F!}&F^20\ar[r]^{F^2!}&\ldots\ar[r] &F^\omega0\ar[r] & \cdots
}
\end{equation}
Above $!:0\to F0$ is the unique map from the initial object $0$ to $F0$. When we reach a limit ordinal $\alpha$ we define $F^\alpha0$ as $\mathrm{colim}_{\beta<\alpha} F^\beta0$. The colimit of this sequence, when it exists, is the carrier of the initial $F$-coalgebra. Notice that for each ordinal $\iota$ we have a canonical map from $F^\iota0$ into the initial algebra. For example, if $F$ preserves colimits of $\omega$-chains, 
then the initial $F$-algebra is $F^\omega 0$.

Colimits of $\omega$-chains are an example of a well studied class of colimits, namely \emph{filtered colimits}~\cite{AdamekR:lpac}. Recall that a filtered colimit is a colimit of a diagram $J:\mathsf D\to\mathsf C$ where $\mathsf D$ is a category such that  any finite diagram in $\mathsf D$ has a cocone. Functors that preserve filtered colimits are called \emph{finitary}. Finitary functors have an initial algebra and the computation of the colimit of the initial chain stops after $\omega$ steps.

\medskip\noindent\textbf{Infinitary data types, or \emph{coalgebraic
    data types}}, are understood as final coalgebras for suitable
functors. An $F$-coalgebra $(X,\xi)$ is defined as an arrow $\xi:X\to
FX$. A coalgebra morphism between $(X,\xi)$ and $(Y,\zeta)$ is a
$\ccal$-morphism $f:X\to Y$ such that $Ff\circ\xi=\zeta\circ
f$. Similarly to the initial algebra situation, we can consider the
final sequence
\begin{equation}
  \label{eq:fin-chain-prel}
  \xymatrix{
1&F1\ar[l]_{!}&F^21\ar[l]_{F!}&\ldots\ar[l] &F^\omega1\ar[l] & \cdots\ar[l]
}
\end{equation}
where $1$ is the final $\ccal$-object and $!$ is the unique arrow from
$F1$ to $1$. Assume that $\ccal$ has limits. For limit ordinals we compute the limit of the diagram
constructed previously. The limit of the final sequence, if it exists,
is the carrier of the final coalgebra. Therefore, functors that
preserve limits of $\omega^\mathit{op}$-chains, sometimes called
\emph{continuous}, have a final coalgebra whose carrier is
$F^\omega1$. We also have canonical maps from the final coalgebra to each
$F^\iota1$ for all ordinals $\iota$.

\medskip\noindent\textbf{The metric completion} of $F^\omega0$ is
given by $F^\omega1$, see Barr \cite[Proposition
3.1]{barr99}. In more detail, if $F$ is an endofunctor on $\Set$ and
$F0$ is nonempty, then one can equip the set $F^\omega 1$ with a metric and
prove that it is the metric completion of $F^\omega 0$. The metric on $F^\omega 1$ is obtained using the
projections $p_n:F^\omega 1\to F^n1$. Explicitly, for $t,t'\in
F^\omega 1$ we put $d(t,t')=2^{-\max\{n\ |\ p_nt=p_nt'\}}$.

\section{Preliminaries on Infinitary Lambda Calculus}
\label{sec:inf-lambda}

\subsection{Infinitary Terms as a Final Coalgebra}

We assume familiarity with basic notions and notations of the finite
$\lambda$-calculus~\cite{bare:1984}. 
The set $\Lb$ of finite $\lambda$-terms is defined by
induction from the grammar:
\begin{equation} \label{definition:coinductive:infiniteterms} M ::=
   x \mid (\lambda x. M) \mid (MM) 
\end{equation}
where $x$ ranges over a given set $\setvars$ of variables.

 First we explain how the set
$\Lbi$ of finite and infinite $\lambda$-terms can be constructed as
the metric completion of the set $\Lb$ of finite $\lambda$-terms. Then
we will briefly recall some notions and facts of infinitary
$\lambda$-calculus~\cite{KKSV97,KV03}.  The notion of 
$\alpha$-equivalence will be given in Definitions~\ref{def:alphafinitelambdaterms}
and~\ref{def:alpha-inf}.

The idea of putting a metric on a set of terms goes at least back to
Arnold and Nivat~\cite{arnoniva80}. To do so we define truncations.

 \begin{defi}[Truncation] \label{def:truncation}
The truncation
   of a term $M \in \Lb$ at depth $n \in \nat$ is defined by induction  on $n$:
   \begin{equation}
     \label{eq:truncation}
     \begin{array}{ll}
       M^0 & = *, \\\\
       M^{n+1}&  = \left \{ \begin{array}{ll}
           x & \textit{ if $M=x\in \mathcal{V}$},\\
           \lambda x.N^n & \textit{ if $M=\lambda x.N$}, \\
           N^nP^n & \textit{ if $M=NP$}.\\ 
         \end{array} \right . 
     \end{array}
   \end{equation}
   where $*$ is any constant not appearing in the syntax of the $\lambda$-calculus, for example $\emptyset$.
 \end{defi}

\begin{defi}[Metric] \label{def:metric} We define  
  a metric $d: \Lb \times \Lb \to [0,1]$ by
  \begin{equation}
    \label{eq:metric}
    d(M, N) = 2^{-m},
  \end{equation}
  where $m = \sup \{ n \in \nat \mid M^n = N^n \} $ and we use the
  convention $2^{-\infty}=0$.
\end{defi}

In fact,  $(\Lb,d)$ is an ultrametric space, since for
all $M,N,P\in\Lb$ we have
$d(M,N)\le\max\{d(M,P),d(P,N)\}$, as one can easily check.

The set $\Lbi$ of finite and infinite $\lambda$-terms 
is now defined as the metric completion
of the set $\Lb$ of finite terms with respect to the metric $d$.
Alternatively, $\Lbi$ can be defined by interpreting
  (\ref{definition:coinductive:infiniteterms}) as a coinductive
  definition. The fact that both definitions coincide is a consequence
of Barr's theorem on final coalgebras for bicontinuous $\Set$
endofunctors. 

Indeed, interpreting (\ref{definition:coinductive:infiniteterms})
  coinductively amounts to taking as $\lambda$-terms the elements of the
  final coalgebra for the $\Set$-endofunctor
\begin{equation}
  \label{eq:set-funct-lambda}
  \TLambda \ \setone =\mathcal{V}+\mathcal{V}\times \setone+\setone \times \setone.
\end{equation}
Notice that the set $\Lb$ of finite $\lambda$-terms constitutes the
initial algebra for $\TLambda$. A closer look at the proof of
Barr~\cite[Theorem~3.2 and Proposition~3.1]{barr99} shows now that the
metric $d$ on $\Lb$ of Definition~\ref{def:metric} coincides with the
metric induced by the final coalgebra. Hence, by
\cite[Proposition~3.1]{barr99}, the completion of the initial
$L$-algebra $\Lb$ in the metric $d$ is the final $L$-coalgebra.

To summarise, the final $\TLambda$-coalgebra
$(\Lbi,\unfold:\Lbi\to\TLambda(\Lbi))$ is the Cauchy completion of
$\Lb$ and we have a dense inclusion map $\iota:\Lb\to\Lbi$.  It is
well-known that the structure map of the final coalgebra
$\unfold:\Lbi\to \TLambda(\Lbi)$ is an isomorphism, hence the set
$\TLambda(\Lbi)$ can be equipped with a complete metric. 
The map $\unfold:\Lbi\to \TLambda(\Lbi)$ is the unique uniformly continuous map
 from $\Lbi$ to $\TLambda(\Lbi)$ making diagram~(\ref{eq:str-map-fin-coalg}) commutative:
\begin{equation}
  \label{eq:str-map-fin-coalg}
\vcenter{
  \xymatrix@R=10pt{
    \Lb\ar[d]_-{\iota}\ar[r]^-{\simeq} &
    \TLambda(\Lb)\ar[d]^-{\TLambda(\iota)} \\
    \Lbi \ar[r]_-{\unfold} &  \TLambda(\Lbi)
  }
}
\end{equation}

Having defined the set $\Lib$ of finite and infinite $\lambda$-terms we
now extend the usual syntactic conventions for finite $\lambda$-calculus to
infinitary $\lambda$-calculus.  Terms and variables will respectively be
written with (super- and subscripted) letters $M,N$ and $x,y,z$.
Terms of the form $(M_1 M_2)$ and $(\lambda x. M)$ will respectively be
called applications and abstractions.

The truncation of an infinite term $M\in\Lbi$ at depth $n$ is defined
just as in Definition \ref{def:truncation} by induction on $n$.
Observe that $(M^n)_{n\in\mathbb{N}}$ is a Cauchy sequence in
$(\Lbi,d)$ that converges to $M$.

The set of free and bound variables of a finite term $M$ is defined as
usual and denoted by $\fv(M)$ and $\bv(M)$ respectively.  We
extend $\fv(M), \bv(M)$ to infinitary terms $M \in \Lib$ using
truncations by
\[ \begin{array}{lll} \fv(M) = \bigcup_{n \in \nat} \fv(M^n)& \ \ \ &
  \bv(M) = \bigcup_{n \in \nat} \bv(M^n).
\end{array}\]
Also, $\var(M)=\fv(M) \cup \bv(M)$.

We define $\beta$-reduction on $\Lbi$ and denote it as $\onebeta$ in
the usual way: the smallest relation that contains $(\lambda x. P)Q
\onebeta P[x:=Q]$ and is closed under contexts.  The reflexive and
transitive closure of $\onebeta$ is denoted by $\finbeta$.  For the
definition of $\manybeta$ that assumes a sequence of reduction steps
of any ordinal length, see for
instance~\cite{kenn:klop:slee:vrie:1995b}.  Terms of the form
$(\lambda x. P)Q$ are called redexes. Normal forms are terms without
redexes and hence cannot be changed by further computation.

The definition of infinitary $\lambda$-calculus is completed by  enriching
the syntax~\eqref{definition:coinductive:infiniteterms} with a fresh constant
$\bot$ and then  adding
$\bot$-reduction, denoted by $\onebot$, defined as
 the smallest relation closed
under contexts and containing $M \onebot \bot$ for $M$ belonging to some
fixed set $\setU$ of \emph{meaningless terms}. If and only if the set
$\setU$ satisfies certain properties, the resulting infinitary
calculus is confluent and normalising, in which case each term has a
unique normal form~\cite{KKSV97,KV03,severidevries:rta2011}.

\subsection{Computing the Infinite Normal Forms using Corecursion}

The normal form of a $\lambda$-term can be thought to represent its
meaning, the maximal amount of information embodied in the term,
stable in the sense that it cannot be changed by further
computation. Note that this concept of meaning depends on the chosen
set $\setU$ of meaningless terms for which there is ample, uncountable
choice~\cite{severidevries:rta2011}.

For concrete sets of meaningless terms an alternative, ``informal''
corecursive definition of the normal form of a term in the
corresponding infinitary $\lambda$-calculus can sometimes be given. Three
of them are well known and they are recalled in 
\eqref{definition:informal:bohmtree}, \eqref{definition:informal:levylongotree}
and \eqref{definition:informal:berarduccitree}. 

In his book~\cite{bare:1984}, Barendregt argued that the terms without
head normal forms should be considered as meaningless terms.  Any
finite $\lambda$-term is either a head normal form (hnf), that is, a
term of the form $\lambda x_1 \ldots \lambda x_n .xN_1 \ldots N_m$, or
it is a term of the form $\lambda x_1 \ldots \lambda x_n .((\lambda
y.P)Q)N_1 \ldots N_m$ where the redex $(\lambda y.P)Q$ is called the
head redex. Starting with a term $M$ that is not in hnf one can
repeatedly contract the head redex. Either this will go on forever or
terminate with a hnf, which represents part of the information
embodied in a term. In the latter case one can repeat this process on
the subterms $N_i$ to try to compute more information.  This idea led
Barendregt to his elegant ``informal'' definition of the B\"ohm tree
$\bohmtree (M)$ of a term $M$, that we now recognise as a corecursive
definition.
\begin{equation} \label{definition:informal:bohmtree}
{\small
  \bohmtree  (M )   = \left\{
  \begin{array}{l}
  \lambda x_1 \ldots \lambda x_{n}. y  
  \bohmtree(M_1)  \ldots \bohmtree (M_{m}), \\
  \quad\quad\quad\quad\quad\quad\quad\quad\quad\quad { \textit{if }
  M \finbeta \lambda x_1 \ldots \lambda x_{n}. y M_1  \ldots M_m}, \\ 
  \bot \quad\quad\quad\quad\quad\quad\quad\quad\quad\textit{ otherwise}.
  \end{array}\right.
}
\end{equation}

The image of $\bohmtree$ is denoted as $\setBT$ and 
can be explicitly defined as follows.

\begin{defi}[Set of B\"ohm trees]
\label{informal:setbohmtrees}
The set  $\setBT$ of {\em B\"ohm trees} is defined as the maximal 
set such that  for all
 $M \in \setBT$, either $M = \bot$ or $M = \lambda x_1 \ldots \lambda x_n. y M_1 \ldots M_m$
where $M_1, \ldots, M_m \in \setBT$ for some $m,n \geq 0$.

\end{defi}

Clearly, any B\"ohm tree in $\setBT$ is an infinitary lambda term over the
syntax~\eqref{definition:coinductive:infiniteterms} enriched with $\bot$.

Taking for $\setU$ the set of terms without hnf, one can show using
the confluence property that the normal forms of the corresponding
infinitary $\lambda$-calculus satisfy the equations
in~(\ref{definition:informal:bohmtree}). That is, the B\"ohm tree of a
term $M$ is the normal form of $M$ in the infinitary
$\lambda$-calculus that equates all terms without head normal form
with $\bot$~\cite{bare:1984,KKSV97}.

Alternatively, as Abramsky has forcefully argued
in~\cite{Abramsky90thelazy}, one can take the set of terms without
weak head normal form (whnf) as set of meaningless terms. Any finite
$\lambda$-term is either a weak head normal form, that is, a term of
either of the two forms $xM_1 \ldots M_m$, or $\lambda x. N$, or it is
a term of the form $((\lambda y.P)Q)M_1 \ldots M_m$ where the redex
$(\lambda y.P)Q$ is called the weak head redex. In perfect analogy
with before, starting with a term $M$ that is not in whnf one can
repeatedly contract the weak head redex. Either this will go on forever or terminate with a whnf. In the latter case one 
can repeat this process on the subterms $M_i$ of
the tail of the whnf or on the subterm $N$ of its body to try to
compute more information. This describes a lazy computation strategy,
that postpones reduction under abstractions as much as possible.

The normal forms of the corresponding infinitary $\lambda$-calculus that equates all terms without weak
head normal form with $\bot$ satisfy the equations~(\ref{definition:informal:levylongotree})
that define the L\'evy-Longo tree $\lltree (M)$ of a term $M$
corecursively~\cite{long83,levy76,abraong93,KKSV97}.
\begin{equation} \label{definition:informal:levylongotree}
{\small
\lltree  (M )  = \left\{
\begin{array}{ll}
      y  \lltree(M_1)  \ldots \lltree (M_{m}) &
  \textit{ if } M \finbeta  y M_1  \ldots M_m, \\
     \lambda x. \lltree (N) &
  \textit{ if }M \finbeta \lambda x.N,\\
   \bot &\textit{ otherwise}.
\end{array}
\right.
}
\end{equation}

The image of $\lltree$ is denoted as $\setLLT$ and can be explicitly defined
as follows.

\begin{defi}[Set of   L\'evy-Longo trees]
\label{informal:setoflltrees}
The set  $\setLLT$ of L\'evy-Longo trees can be defined as the maximal set
that satisfies that whenever $M \in \setLLT$ then $M$ has one of the
following shapes:
\begin{enumerate}
\item 
either $M = \bot$, or 
\item $M = \lambda x. N$ for some $N \in \setLLT$, or 
\item  $M = y M_1 \ldots M_n $ for some $M_1, \ldots, M_n \in \setLLT$.
\end{enumerate}
\end{defi}

The least set of meaningless terms that gives rise to a confluent and
normalising infinitary $\lambda$-calculus is the set of terms without
a top normal form.  Here a term $M$ is a top normal form (tnf) if it
is either a variable, an abstraction or an application of the form
$M_1 M_2$ in which $M_1$ is a {\em zero term}, i.e.\ a term that cannot 
reduce to an abstraction.  The well-known term $\Omegaterm = (\lambda
x.xx)(\lambda x.xx)$ has no tnf. The normal forms of this calculus can
alternatively be characterised by the corecursive definition of the
Berarducci tree~\cite{bera96,KKSV97} $\bertree (M)$ of a term $M$:
\begin{equation} \label{definition:informal:berarduccitree}
{\small
\bertree  (M )  = 
\left\{
\begin{array}{ll}
  x &
  \textit{ if } M \finbeta x, \\
  \lambda x. \bertree (N) &
  \textit{ if } M \finbeta \lambda x.N, \\
 \bertree(N)  \bertree(P) &
  \textit{ if } M \finbeta N P 
  \textit{ and $N$ is a zero term}, \\
  \bot& \textit{ otherwise},
\end{array}
\right.
}
\end{equation}

The image of $\bertree$ is denoted as $\setBerT$ and can be
explicitly defined as follows.

\begin{defi}[Set of Berarducci  trees]
\label{informal:setofbertrees}
The set  $\setBerT$ of {\em Berarducci  trees} can be defined as the maximal
set that satisfies that whenever $M \in \setBerT$ then $M$ has one
of the following shapes:
\begin{enumerate}
\item 
$M = \lambda x. N$ for some $N \in \setBerT$, or
\item $M = y M_1 \ldots M_n$ for some $M_1, \ldots, M_n \in \setBerT$, or
\item $M = \bot M_1 \ldots M_n$ for some $M_1, \ldots, M_n \in \setBerT$, or
\item $M =  (( \ldots ) M_2 ) M_1$ for some $(M_i)_{i\geq 1}$ such that
$M_i \in \setBerT$ for all $i \geq 1$.

\end{enumerate}
\end{defi}

Some examples of  trees are shown in Figure \ref{figure:examplesoftrees}.

\begin{figure}
\[
\begin{array}{c|c|c|c}
\mbox{{\bf Term}} &
\mbox{{\bf B\"ohm tree}}& 
\mbox{{\bf L\'evy-Longo tree}}& 
\mbox{{\bf Berarducci tree}} \\  
M & \bohmtree(M)  & \lltree(M) & \bertree(M) \\
\hline  & & & \\
\Yterm \ x &
x (x (x ( \ldots))) &
x (x (x ( \ldots))) &
x (x (x ( \ldots))) \\
\Yterm (\lambda y x. x y) &
\lambda x_1. x_1  (\lambda x_2. x_2 ( \ldots)) &
\lambda x_1. x_1  (\lambda x_2. x_2 ( \ldots))  &
 \lambda x_1. x_1  (\lambda x_2. x_2 ( \ldots)) \\
\lambda x. \Omega  
&  
\bot
& 
\lambda x. \bot
&
 \lambda x. \bot \\
 \Yterm \ (\lambda x y. x) &
 \bot &
 \lambda x_1 \lambda x_2 \lambda x_3 \ldots &
 \lambda x_1 \lambda x_2 \lambda x_3 \ldots \\
 \Omega \Omega & 
 \bot &
 \bot &
 \bot \bot \\
 \Yterm (\lambda y. y x) &
 \bot &
 \bot &
(( (\ldots) x) x) x)
\end{array}
\]
\caption{Examples of B\"ohm, L\'evy-Longo and Berarducci trees}
\label{figure:examplesoftrees}
\end{figure}

It is possible to formalise
(\ref{definition:informal:bohmtree})-(\ref{definition:informal:berarduccitree})
using corecursion via the final $\TLambda$-coalgebra, provided we give
concrete reduction strategies to compute the various forms used in the
definitions. 
 However, in order to take into account
$\alpha$-equivalence, we will prove an $\alpha$-corecursion principle
based on nominal sets.

\section{Preliminaries on Nominal Sets}
\label{sec:nominal-sets}

We recall basic facts on nominal sets~\cite{gabb-pitt:lics99-j,Pitts-book}.

Consider a countably infinite set $\setvars$ of `variables' (or `atoms' or
`names') and the group $\grp$ of permutations on $\setvars$ generated
by transpositions, which are permutations of the form $(x\; y)$ that
swap $x$ and $y$. Consider a set $\setone$ equipped with an action of
the group $\grp$, denoted by $\cdot:\grp\times\setone \to \setone $.
We say that $\elone \in \setone$ is supported by a set $\setS
\subseteq \setvars$ when for all $\pi\in\grp$ such that $\pi(x)=x$ for
all $x \in \setS$ we have $\pi \cdot \elone =\elone$.  We say that
$\elone \in \setone$ is finitely supported if there exists a finite
$\setS \subseteq\setvars$ which supports $\elone$.

\begin{defi}[Nominal set]
  A nominal set $(\setone,\cdot)$ is set $\setone$ equipped with a
  $\grp$-action such that all elements of $\setone$ are finitely
  supported.  Given nominal sets $(\setone,\cdot)$ and
  $(\settwo,\cdot)$, a map $f:\setone \to \settwo$ is called
  \emph{equivariant} when $f(\pi\cdot \elone)=\pi\cdot f(\elone)$ for all
  $\pi\in\grp$ and $\elone \in\setone$. The category of nominal sets
  and equivariant maps is denoted by $\Nom$.
\end{defi}

A crucial property of nominal sets is that each element of a nominal
set has a least finite support, see~\cite{gabb-pitt:lics99-j}. Indeed, if two finite sets $\setS_1$
and $\setS_2$ support $\elone$, then their intersection also supports
$\elone$.  
\begin{notation}[Support and freshness]
The smallest finite support of $\elone$ is denoted by
$\supp(\elone)$.  If $x \in\mathcal{V}\setminus\supp(\elone)$ we say
that $x$ is \emph{fresh} for $\elone$, and write $x \# \elone$. More
generally, given two nominal sets $(\setone,\cdot)$ and
$(\settwo,\cdot)$, $\elone\in\setone$ and $\eltwo\in\settwo$, we write
$\elone\#\eltwo$ for $\supp(\elone)\cap\supp(\eltwo)=\emptyset$. Given $S,T\subseteq\setvars$, we write $S\#\elone$ for $\supp(\elone)\cap S=\emptyset$.
We also write $S \# T $ for $S \cap T = \emptyset$.
\end{notation}

\begin{rem}\label{rem:supp-shrink}
An important property of $\supp$ is that for every equivariant
$f:\setone \to \settwo$ and $\elone \in \setone$, we have
$\supp(f(\elone))\subseteq\supp(\elone)$.
\end{rem}

\begin{exa}
  The set of names $\mathcal{V}$ equipped with the evaluation action
  given by $\pi\cdot x=\pi(x)$ is a nominal set.
\end{exa}

\begin{exa}
  The finite subsets of atoms $\mathcal{P}_{\mathit{fin}}(\setvars)$
  form a nominal set with the pointwise action $\pi\cdot \setone
  =\{\pi(u)\ |\ u\in \setone \}$ for all $\setone
  \in\mathcal{P}_{\mathit{fin}}(\setvars)$.
\end{exa}

\begin{rem}\label{rem:supp-equiv}
  Notice that taking the support of elements of a nominal set $\setone$ gives an equivariant map 
$\supp:\setone\to\mathcal{P}_{\mathit{fin}}(\setvars)$. Indeed, one can show that for any $\elone\in\setone$ and $\pi\in\grp$ we have $\supp(\pi\cdot \elone)=\pi\cdot\supp(\elone)$. As a consequence, $\supp(\pi\cdot\elone)$ and $\supp(\elone)$ have the same cardinality for any permutation $\pi$.
\end{rem}

\begin{exa}\label{ex:lambda-nom}
  The set $\Lb$ of finite $\lambda$-terms with the action
  $\cdot:\grp\times\Lb\to\Lb$ inductively defined by
  \begin{equation}
    \label{eq:lam-act}    
    \begin{array}{rcl}
      \pi\cdot x&  = & \pi (x) \\
      \pi\cdot (\lambda x. M ) & = & \lambda\pi(x).(\pi\cdot M)\\
      \pi\cdot (MN)&  = &((\pi\cdot M)( \pi\cdot N))
    \end{array}
  \end{equation}
  is a nominal set. In this example we do not take into account
  $\alpha$-equivalence, so the support of a $\lambda$-term $M$ is the
  set of all variables occurring either bound or free in $M$.
\end{exa}

Given a $\grp$-action $\cdot$ on a set $\setone$, let $\f{\setone}$
denote the set
  \begin{equation}
    \label{eq:fs}
\f{\setone}=\{ \elone \in \setone \mid \elone \textit{ is finitely supported} \}.    
  \end{equation}
  Then $\cdot$ restricts to a
  $\grp$-action on $\f{\setone}$ and $(\f{\setone},\cdot)$ is a
  nominal set.
  
\begin{exa}\label{ex:lambda-inf-act}
  The set $\Lbi$ of finite and infinite $\lambda$-terms can be
  equipped with the action $\cdot:\grp\times\Lbi\to\Lbi$ defined
  \emph{coinductively} by~(\ref{eq:lam-act}).  Alternatively,
  $\pi\cdot(-)$ can be defined using the universal property of the
  metric completion, as the unique map that extends
  $\Lb\stackrel{\pi\cdot
    (-)}{\longrightarrow}\Lb\stackrel{\iota}{\longrightarrow}\Lbi$.  Observe that $(\pi\cdot M)^n=\pi\cdot M^n$ for all
  $M\in\Lbi$ and $n\in\nat$.  Notice that $(\Lbi,\cdot)$ is not a
  nominal set since the set of variables in a term, and hence its
  support, can be infinite.  But $(\f{(\Lbi)},\cdot)$ is a nominal set
  and $\supp(M)=\var(M)$ for all $M\in \f{(\Lbi)}$.
\end{exa}

\begin{defi}[Abstraction]\label{def:abstraction} Let
  $(\setone,\cdot)$ be a nominal set. One defines $\alphaconvg$ on
  $\setvars \times \setone$ by
  \begin{equation}
    \label{eq:equiv-abstr}
    (x_1, \elone_1) \alphaconvg (x_2, \elone_2) \Leftrightarrow
    (\exists \ z \fresh \{x_1, \elone_1, x_2, \elone_2\}) (x_1 \ z)
    \cdot \elone_1 = (x_2 \ z) \cdot \elone_2
  \end{equation}
  The $\alphaconvg$-equivalence class of $(x,\elone)$ is denoted by
  $\abs{x}{\elone}$.  The abstraction $[\setvars]\setone$ of the
  nominal set $\setone$ is the quotient $(\setvars \times
  \setone)/\!\!\alphaconvg$. The $\grp$-action on $[\setvars]\setone$ is
  defined by
    \begin{equation}
      \label{eq:act-abstraction}
\pi \cdot \abs{x}{\elone} =
    \abs{\pi \cdot x}{\pi \cdot \elone}.      
    \end{equation}
    Given equivariant $f:(\setone,\cdot)\to(\settwo,\cdot)$, we define
    $[\setvars]f:[\setvars]\setone\to[\setvars]\settwo$ by
\begin{equation}
  \label{eq:abs-morph}
  \abs{x}\elone\mapsto\abs{x}f(\elone).
\end{equation}
\end{defi}

\begin{defi}\label{def:concretion}[Concretion]
  Let
  $(\setone,\cdot)$ be a nominal set. Concretion is the partial function
  $\conc:[\setvars]\setone\times\setvars\to\setone$ with
  $\abs{y}\elone\conc z$, the `concretion of
  $\abs{y}\elone$ at $z$', defined as
  $\abs{y}\elone\conc z = (z\ y) \cdot \elone$ if
  $z\in\setvars\setminus\supp(\abs{y}\elone)$.
\end{defi}

\noindent Notice that $y\fresh\abs{y}\elone$ and $(\abs{y}\elone)\conc
  y=\elone$. Moreover, observe that $[\setvars]\setone\times\setvars$ is a nominal set with the coordinatewise action of $\grp$. One can show that concretion is equivariant. Indeed, if $z\#\abs{y}\elone$ then $\pi\cdot z\#\abs{\pi\cdot y}\pi\cdot\elone$ and  $\pi\cdot (\abs{y}\elone\conc z)=(\abs{\pi\cdot y}\pi\cdot\elone)\conc \pi\cdot z$.
  \begin{defi}\label{def:int-hom}[Internal hom]
    Given two nominal sets $(\setone,\cdot)$ and $(\settwo,\cdot)$, we define the internal hom $[\setone,\settwo]$ as the nominal set of all functions $f:\setone\to\settwo$ that are finitely supported with respect to the action
$$(\pi\cdot f)(\elone)=\pi\cdot f(\pi^{-1}\cdot \elone).$$
  \end{defi}
  \begin{rem}\label{rem:fs-functions}
    A function $f:\setone\to\settwo$ is finitely supported if and only if there exists a finite set  $S$ of names, such that for all permutations $\pi\in\grp$ that fix the names in $S$ and for all $x\in X$ we have $\pi\cdot f(x)=f(\pi\cdot x)$.
  \end{rem}
\noindent\textbf{Limits and colimits in $\Nom$.}
Further, we recall some general results form~\cite{PittsAM:nomlfo-jv} that will be necessary in the rest of the paper.
The category $\Nom$ is complete and cocomplete. The forgetful functor
to $\Set$ creates finite products and all colimits. For example, the
product of two nominal sets $(\setone,\cdot)$ and $(\settwo,\cdot)$ is
$(\setone \times\settwo,\cdot)$ where
\[ \pi\cdot(\elone,\eltwo)=(\pi\cdot \elone,\pi\cdot \eltwo).\]
Arbitrary products in $\Nom$ are computed differently than in
$\Set$. Given a family of nominal sets $(\setone_i,\cdot_i)_{i\in
  I}$, we can equip the set of all tuples $\{(\elone_i)_{i\in I}\ |\
\elone_i\in \setone_i\}$ with the pointwise action given by
 \begin{equation}\label{equ:arb-nom-prod}
   \pi\cdot (\elone_i)_{i\in I}=(\pi\cdot_i \elone_i)_{i\in I}.    
 \end{equation}
 This is a $\frak{S}(\setvars)$-action, but some tuples may not be
 finitely supported.  The product of $(\setone_i,\cdot_i)_{i\in I}$ in
 $\Nom$ is the nominal set $(\prod\limits_{i\in
   I}\f{(\setone_i,\cdot_i))}$ of tuples of the form $(\elone_i)_{i\in
   I}$ that are finitely supported with respect to the action
 of~(\ref{equ:arb-nom-prod}).

The limit in $\Nom$ of an $\omega^\mathit{op}$-chain 
\[
  \xymatrix{
\setone_1&\setone_2\ar[l]_{f_1}&\ldots\ar[l]_{f_2} 
}
\]
is the nominal set of \emph{finitely supported} tuples $(\elone_1,\elone_2,\ldots)$ such that for all $i\in\omega$ we have $\elone_i\in\setone_i$ and $f_{i+1}(\elone_{i+1})=\elone_i$.

The initial object of $\Nom$ is the empty nominal set with the trivial action. In $\Nom$, all monomorphisms are strong and are precisely the injective equivariant maps.

\medskip\noindent\textbf{Categorical properties of $\Nom$.} The category $\Nom$
is \emph{locally finitely presentable}, see \cite{AdamekR:lpac}. 
An  object $X$ in a category $\mathsf C$ is called finitely presentable when the hom functor $\mathsf C(X,-)$ preserves filtered colimits. For example the finitely presentable sets are the finite ones. A locally small category $\mathsf C$ is called locally finitely presentable when it is cocomplete and it has a small set $A$ of finitely presentable objects, such that any object in $\mathsf C$ is a filtered colimit of objects in $A$.

We describe next the finitely presentable objects in $\Nom$. First we need to define the notion of \emph{orbit}.

\begin{defi}[Orbit]
   Consider a nominal set $(\setone,\cdot)$ and $x, y\in \setone$. We say that $x$ and $y$ are orbit equivalent when there exists $\pi\in\grp$ such that $\pi\cdot x=y$. An orbit  $\mathcal{O}\subseteq\setone$ is an equivalence class with respect to this equivalence relation.
 \end{defi}
 
 \begin{rem}
 Let $\mathcal{O}$ denote an orbit in a nominal set and  consider $x\in\mathcal{O}$. Then $\mathcal{O}= \{\pi\cdot x \mid \pi\in\grp\}$.
 \end{rem}

 A proof of the following proposition is in~\cite[Proposition 2.3.7]{petrisan:phd}.

 \begin{prop}
  A nominal set $(\setone,\cdot)$ is finitely presentable in $\Nom$ if and only if it has finitely many orbits.  
 \end{prop}

Observe that a nominal set is the disjoint union of its orbits, and hence the directed union of all its nominal subsets with finitely many orbits.

\medskip\noindent\textbf{Properties of the abstraction functor.}  The abstraction functor $[\setvars](-):\Nom\to\Nom$ preserves all limits and colimits, see~\cite{Pitts-book}. The remainder of this section is only needed in the proof of Lemma~\ref{lem:off-maps-closure-prop}.

By~\cite[Theorem~4.13]{Pitts-book} 
we know that $[\setvars](-)$ has a right adjoint $R:\Nom\to\Nom$ defined by 
$$R\setone=\{ g\in[\setvars,\setone]\ |\ (\forall x\in\setvars) x\# g(x)\}.$$
Above, $[\setvars,\setone]$ is the nominal set of finitely supported maps from $\setvars$ to $\setone$, as in~Definition~\ref{def:int-hom}. The argument below~\cite[(4.27)]{Pitts-book} also implies that $R$ does not preserve colimits, the counterexample being the coproduct of two nominal sets. However we can show that $R$ preserves filtered colimits. But first let us give a simpler description of the nominal set $R\setone$. 

\begin{lem}\label{rem:obs-RX}
Consider a nominal set $X$ and a $\Set$-function $g:\setvars\to\setone$. We have that $g\in R\setone$ if and only if the following conditions are satisfied:
\begin{enumerate}
\item For all $x\in\setvars$ we have $x\#g(x)$.
\item There exists a finite set of names $S$ such that for all $x\in\setvars$ we have $\supp(g(x))\subseteq S$ and $g$ is constant on $\setvars\setminus S$.
\end{enumerate}
\end{lem}
 \proof For the direct implication, notice that (1) is clearly satisfied. For (2), let $S$ denote the support of $g$. Observe that, by Remark~\ref{rem:fs-functions} we have that $\supp(g(x))\subseteq S\cup\{x\}$. Since $x\#g(x)$ we have that $\supp(g(x))\subseteq S$ for all $x\in\setvars$.  Now consider $x,x'\not\in S$. Then, by Remark~\ref{rem:fs-functions}, we have $(x\ x')\cdot g(x)=g(x')$. On the other hand, $(x\ x')\cdot g(x)=g(x)$ holds since $x,x'\#g(x)$. Hence $g(x)=g(x')$, thus $g$ is constant outside $S$.

Conversely, we only have to show that $g$ is finitely supported. Let $\pi\in\grp$ be a permutation that fixes the set $S$. By Remark~\ref{rem:fs-functions}, it is enough to show that $\pi\cdot g(x)=g(\pi\cdot x)$ for all $x\in\setvars$. Since $\supp(g(x))\subseteq S$ we have that $\pi\cdot g(x)=g(x)$. We can also prove that $g(\pi\cdot x)=g(x)$. This is clear for $x\in S$. For the case when $x\not\in S$, observe that $\pi\cdot x=\pi(x)\not\in S$, thus we can use the fact that $g$ is constant outside $S$.
\qed
\begin{prop}
  The functor $[\setvars](-):\Nom\to\Nom$ has a finitary right adjoint.
\end{prop}

\newcommand{\colim}{\mathrm{colim}}

\proof 
 Consider a filtered diagram $\mathsf{D}$ and a functor $J:\mathsf{D}\to\Nom$. In order to prove that $R\colim_{d\in\mathsf{D}}{Jd}\simeq\colim_{d\in\mathsf{D}}RJ(d)$ it is enough to show that each finitely supported map $g:\setvars\to\colim_{d\in\mathsf{D}}Jd$ such that for all $x\in\setvars$ we have $x\# g(x)$ factors through a $Jd$ for some $d\in\mathsf{D}$
\begin{equation}
  \label{eq:2}
  \xymatrix{
\setvars\ar@{-->}[d]_-{g_d}\ar[rd]^-{g}\\
Jd\ar[r]^-{\iota}&\colim_{d\in\mathsf{D}}Jd
}
\end{equation}
and the  map $g_d$ is finitely supported and for all $x\in\setvars$ we have $x\#g_d(x)$.
Let $S$ denote the finite support of $g$ and let $x_0\not\in S$. There exists $d_0\in\mathsf{D}$ such that $g(x_0)=[y_0]$ for some $y_0\in Jd_0$. We used  square brackets here to denote the equivalence classes needed in the computation of filtered colimits. By Lemma~\ref{rem:obs-RX} we have that $\supp([y_0])\subseteq S$ and for all $x\not\in S$ we have that $g(x)=[y_0]$.  Notice that $S=\{x_1,\ldots, x_n\}$ is a finite set and for each $x_i$ there exists $d_i\in\mathsf{D}$ such that $g(x_i)=[y_i]$ for some $y_i\in Jd_i$. 
  We can assume without loss of generality that for all $0\le i\le n$ we have $\supp(y_i)=\supp([y_i])$, see~\cite[Proposition~2.3.7]{petrisan:phd}. 

By Lemma~\ref{rem:obs-RX} for all $0\le i\le n$ we have $x_i\# g(x_i)=[y_i]$ and $\supp (g(x_i))=\supp([y_i])\subseteq S$. Hence, for all $0\le i\le n$ we have $x_i\#y_i$ and 
$\supp(y_i)\subseteq S$.
Since $\mathsf{D}$ is filtered, there exists $d\in\mathsf{D}$ and arrows $u_i:d_i\to d$ in $\mathsf{D}$ for all $0\le i\le n$. 
We define $g_d:\setvars\to Jd$ by
\begin{equation}
  \label{eq:g-factor}
  g_d(x)=\left\{
  \begin{array}{ll}
    Ju_0(y_0) &\textrm{ if }x\not\in S,\\
    Ju_i(y_i)&\textrm{ if }x=x_i\textrm{ for some }1\le i\le n.
  \end{array}
\right.
\end{equation}
Notice that $g_d$ satisfies the properties (1) and (2) of Lemma~\ref{rem:obs-RX}. Indeed, we  have that $x\# g_d(x)$ and $\supp(g_d(x))\subseteq S$ for all $x\in\setvars$ (because the maps $Ju_i$ are equivariant and thus can only shrink the support of the elements, recall Remark~\ref{rem:supp-shrink}).
Thus $g_d$ is in  $R(Jd)$ by Lemma~\ref{rem:obs-RX} and it makes diagram~\eqref{eq:2} commutative. Therefore $R$ preserves filtered colimits.
\qed

\begin{lem}
\label{lem:abs-pres-fp}
The functor $[\setvars](-):\Nom\to\Nom$ preserves finitely presentable objects.  
\end{lem}
\proof This is immediate by the previous lemma. Let $\setone$ be a finitely presentable nominal set and let $\colim\ \settwo_i$ be a filtered colimit. Then
  \begin{equation}
    \label{eq:1}
    \begin{array}{lcl}
  \Nom([\setvars]\setone,\colim \ \settwo_i)& \simeq & \Nom(\setone,R(\colim\ \settwo_i))  \\
  & \simeq &\Nom(\setone,\colim\ R(\settwo_i))\\  
&\simeq &\colim \ \Nom(\setone,R(\settwo_i))\\
&\simeq &\colim\ \Nom([\setvars]\setone,\settwo_i).
\end{array}
  \end{equation} 
 This shows that $[\setvars]\setone$ is finitely presentable.
\qed

\section{Alpha Corecursion Principle for Nominal Coalgebraic Data Types}
\label{sec:coalgebraicdatatypes}

In this section we introduce nominal coalgebraic data types as a means
of studying infinitary data up to $\alpha$-equivalence.  In
Section~\ref{sec:fin-coalg} we describe final coalgebras of certain
$\Nom$-endofunctors. We apply this result for functors arising from
signatures with binding. We prove that the set of equivalence classes
of infinitary terms with finitely many free variables is the final
coalgebra of a $\Nom$-functor.  Our running example is the infinitary
$\lambda$-calculus, and indeed the results of~\cite{kurzetal:cmcs2012}
are  particular instances of the main theorems in this section.

\subsection{Final Coalgebras of $\Nom$ Functors}
\label{sec:fin-coalg}
In this section we describe the final coalgebras for certain
endofunctors on $\Nom$. It is well known that an endofunctor that
preserves limits of $\omega^\mathit{op}$-chains has a final coalgebra
which is computed as the limit of the final chain. Similarly,
endofunctors that preserve colimits of $\omega$-chains have an initial
algebra obtained as the colimit of the initial sequence. We will generalise Barr's theorem~\cite{barr99} relating final coalgebras and initial algebras to functors on nominal sets. To this end we need to introduce nominal (complete) metric spaces.

\begin{defi}[Nominal metric space]
  A \emph{nominal metric space} is a tuple $(\setone,\cdot,d)$ such
  that $(\setone,\cdot)$ is a nominal set and
  $d:\setone\times\setone\to [0,1]$ is an equivariant metric when the
  interval $[0,1]$ is equipped with the trivial action. That is, $d(x,y)=d(\pi\cdot x,\pi\cdot y)$ for all $x,y\in\setone $ and $\pi\in\grp$.
\end{defi}

\begin{defi}[Finitely supported Cauchy sequence]
  A \emph{finitely supported Cauchy sequence} in a nominal metric
  space is a Cauchy sequence $(x_n)_{n\ge 1}$ such that there exists a
  finite set of variables $S\subseteq\setvars$ that supports all
  elements $x_n$. A nominal metric space is \emph{complete} when every
  finitely supported Cauchy sequence converges.
\end{defi}
\begin{rem}[Nominal Completion]
Given a nominal metric space $(\setone,\cdot, d)$ one can construct its
nominal completion $(\overline{\setone},\cdot,\overline{d})$ by adding the
limits of all the finitely supported Cauchy sequences. This
construction has the following universal property. For any complete
metric space $(\settwo,\cdot, e)$ and any equivariant uniformly continuous
function $f:\setone\to\settwo$ there exists a unique equivariant
uniformly continuous map $\overline{f}:\overline{\setone}\to\settwo$
extending $f$. The proofs are straightforward.  
\end{rem}

In what follows, Theorem~\ref{thm:fin} and Theorem~\ref{thm:fin-v2} show that the final coalgebra of certain
$\Nom$-functors can be regarded as a nominal metric space and is the
nominal completion of the initial algebra. Theorem~\ref{thm:fin} is an instance of Ad\'
amek's generalisation of Barr's theorem from~\cite{Adamek03}. The fact
that $T$ is the \emph{nominal} Cauchy completion of $I$ is equivalent
to $\hom(B,\fcoalg)$ being the Cauchy completion of $\hom(B,\ialg)$
for any finitely presentable objects $B\in\Nom$. Nevertheless, we
opted to sketch a direct proof below, not only for the sake of
completeness, but also because we can work directly with metrics on
the initial algebra and final coalgebra, paying attention to some
extra conditions regarding equivariance and finite support. Moreover, a careful inspection of the proof allows us to prove a small variation of this result, see Theorem~\ref{thm:fin-v2}. 
One crucial hypothesis in~\cite{Adamek03} is that the functor at issue,
say $\fnom{F}$, has the property that $\fnom{F}0$ has an element, where
$0$ is the initial element in the category. By this, it is meant that
a morphism $s:1\to \fnom{F}0$ exists. However, by insights that go back to work
by Fraenkel and
Mostowski~\cite{jech,Pitts-book},  the axiom of choice doesn't hold in the topos of nominal
sets. So the fact that
$\fnom{F}0\neq 0$ is not equivalent to the existence of a
$\Nom$-morphism $s:1\to \fnom{F}0$. As an example consider the functor $\FLambda$  of~\ref{exa:funct-lam-calc}, where we can apply Theorem~\ref{thm:fin-v2}, but not Theorem~\ref{thm:fin}.

We first formulate Theorem~\ref{thm:fin} following~\cite{Adamek03} in
assuming that a morphism $s:1\to
\fnom{F}0$ exists and $\fnom{F}$ preserves limits of
$\omega^\mathit{op}$-chains and monomorphisms. Note that we do not assume that $\fnom{F}$
preserves colimits of $\omega$-chains. The existence of the final
coalgebra gives a size constraint on the functor $\fnom{F}$ that
ensures also the existence of the initial algebra.

However, in Theorem~\ref{thm:fin-v2} 
we will see that a variation of
Theorem~\ref{thm:fin} in which we only assume that $\fnom{F}0\neq 0$
holds provided that $\fnom{F}$ can be extended to finitely-supported
maps.

\begin{thm}\label{thm:fin}
  Let $\fnom{F}:\Nom\to\Nom$ be a functor that preserves limits of
  $\omega^\mathit{op}$-chains and monomorphisms and such that a
  morphism $s:1\to \fnom{F}0$ exists. Then $\fnom{F}$ has a final
  coalgebra $\fcoalg$ and an initial algebra $\ialg$, both of which be
  equipped with equivariant metrics. Moreover the final coalgebra is
  the nominal completion of the initial algebra.
\end{thm}
\proof 
We split the proof in four parts.
\begin{enumerate}
\item[(1)]\emph{Existence of final $\fnom{F}$-coalgebra and initial
    $\fnom{F}$-algebra.}  Since it preserves limits of
  $\omega^\mathit{op}$-chains, the functor $\fnom{F}$ has a final
  coalgebra which can be computed as the limit of the
  $\omega^\mathit{op}$-chain
  \begin{equation}
    \label{eq:chain}
    \xymatrix{
      1\ar@{<-}[r]^{!} & \fnom{F}1 \ar@{<-}[r]^{\fnom{F}!} & \fnom{F}^21 \ar@{<-}[r]^{\fnom{F}^2!} & \cdots  \ar@{<-}[r] & \fcoalg
    }
  \end{equation}

  Since $\Nom$ is locally finitely presentable and $F$ preserves
  monomorphisms we have that the initial $F$-algebra exists and is a
  subobject of $\fcoalg$. For full details and the general proof
  see~\cite[Proposition~3.4]{Adamek03}.  The idea is to prove by
  transfinite induction that for every ordinal $\iota$ we have a
  monomorphism $u_\iota:\fnom{F}^\iota0\to\fnom{F}^\iota1$. For example for finite
  ordinals we put $u_i=\fnom{F}^i(u)$ where $u:0\to 1$ is the unique
  morphism into the final nominal set. For all $\iota\ge\omega$ we 
  have that $\fnom{F}^\iota0$ is a subobject of $\fnom{F}^\iota1\cong\fnom{F}^\omega1$. Using the fact that $\Nom$ is well powered, we know that $\fnom{F}^\omega1$ only has a set of
  subobjects.  Hence the initial sequence converges to
  the initial $\fnom{F}$-algebra, denoted by $\ialg$.  Moreover since
  for all $\iota\ge\omega$ we have
  $\fnom{F}^\iota1\cong\fnom{F}^\omega 1\cong\fcoalg$ there exists a
  monomorphism $v:\ialg\to\fcoalg$.

  \begin{equation}
    \label{eq:thm-fin}
\begin{gathered}
    \xymatrix
    {
      0\ar[r]\ar[d]^{u} &  \fnom{F}0 \ar[r]\ar@{^{(}->}[d]^{u_1} &  \fnom{F}^20 \ar[r]\ar@{^{(}->}[d]^{u_2} & \cdots  \ar[r] & \fnom{F}^\omega0\ar@{^{(}->}[d]^{u_\omega}\ar[r]
      &\cdots\ar[r] & \fnom{F}^\iota 0\ar@{^{(}->}[d]^{u_\iota}&\cong &\ialg\ar@{^{(}-->}[d]^v \\
      1\ar@{<-}[r] & \fnom{F}1 \ar@{<-}[r] & \fnom{F}^21 \ar@{<-}[r] & \cdots  \ar@{<-}[r] & \fnom{F}^\omega1\ar@{<-}[r]^{\cong}&\cdots& \fnom{F}^\iota1\ar[l]&\cong&\fcoalg
    }
  \end{gathered}
\end{equation}

\item[(2)]\emph{Equivariant metrics on $\fcoalg$ and $\ialg$.}  First
  we define a metric on $\fcoalg$. Let $p_n:\fcoalg\to\fnom{F}^n1$
  denote the projections of the limit in~(\ref{eq:chain}). Put
  \begin{equation}
    \label{eq:metric-gen}
    d(x,y)=2^{-max\{n\ \mid \ p_nx=p_ny\}}
  \end{equation}
  with the convention that $2^{-\omega}=0$.  Considering the interval
  $[0,1]$ as a nominal set with the trivial action the map
  $d:\fcoalg\times\fcoalg\to[0,1]$ becomes equivariant. This follows
  easily since each $p_n$ is equivariant.

\item[(3)]\emph{$\fcoalg$ is a nominal complete metric space.}  Next
  we show that $(\fcoalg, d)$ is a nominal complete metric
  space. Consider a finitely supported Cauchy sequence
  $(x_n)_n\subseteq\fcoalg$. This implies that there exists a finite
  set $S\subseteq\setvars$ that supports all $x_n$. Without loss of
  generality we may assume that for every $n$ we have
  $d(x_n,x_{n+1})\le 2^{-n}$. Therefore, we have that
  $p_nx_n=p_nx_{n+1}$ for all $n$, or equivalently,
  $p_nx_n=\fnom{F}^n!(p_{n+1}x_{n+1})$. Since for all $n$ we have
  $\supp(x_n)\subseteq S$ and $p_n$ is equivariant, it follows that
  $\supp(p_nx_n)\subseteq S$ for all $n$. The existence of such a
  common support is essential, recall how limits of
  $\omega^\mathit{op}$-chains are computed in $\Nom$. It follows that
  there exists an element $x\in\fcoalg$ corresponding to the tuple
  $(p_nx_n)_n$. We thus have that $p_nx=p_nx_n$ for all $n$ and this
  proves that $x$ is a limit of $(x_n)_n$ with respect to the metric
  $d$.
\item[(4)]\emph{$\fcoalg$ is the nominal completion of $\ialg$.} Since
  $\ialg$ contains $\fnom{F}^\omega0$, see~\eqref{eq:thm-fin}, it is
  enough to prove that any element $x\in\fcoalg$ can be written as the
  limit of a Cauchy sequence of elements in $\fnom{F}^\omega0$.

  Next we use the existence of a morphism $s:1\to\fnom{F}0$. Notice
  that there exists a unique (finitely supported) map into $1$,
  therefore we have that $!u_1s=\mathit{id}_1$.  We define
  $x_n\in\fnom{F}^\omega0$ as the image under the inclusion
  $\fnom{F}^{n+1}0\to\fnom{F}^\omega0$ of $\fnom{F}^{n}s(p_nx)$.
  Notice that the support of $x_n$ is included in the support of $x$
  for all natural numbers $n$. It is easy to check that $p_nu_\omega
  x_n=p_nx$, thus $d(u_\omega x_n,x)\le 2^{-n}$. Therefore the
  sequence $(x_n)_n$ is a finitely supported Cauchy sequence whose
  limit is $x$.\qed
\end{enumerate}

\noindent In order to relax the assumption that the map $s:1\to\fnom{F}0$ is equivariant we need to require the functor $\fnom{F}$ to be $\Nom$-enriched. This means that for any two nominal sets $\setone$ and $\settwo$ we have an equivariant map $\fnom{F}_{\setone,\settwo}:[\setone,\settwo]\to[\fnom{F}\setone,\fnom{F}\settwo]$ which behaves well with respect to composition~\cite{KellyGM:bascec}, where $[\setone, \settwo]$ denotes the internal hom in the cartesian closed category $\Nom$, and consists of the finitely supported maps from $\setone$ to $\settwo$, see Definition~\ref{def:int-hom}. 

\begin{thm}
  \label{thm:fin-v2}
  Let $\fnom{F}:\Nom\to\Nom$ be a functor that preserves limits of
  $\omega^\mathit{op}$-chains and monomorphisms and such that $\fnom{F}0\neq 0$. Assume further that the functor $\fnom{F}$ is $\Nom$-enriched. Then $\fnom{F}$ has a final
  coalgebra $\fcoalg$ and an initial algebra $\ialg$, both of which can be
  equipped with equivariant metrics. Moreover the final coalgebra is
  the nominal completion of the initial algebra.
\end{thm}

\proof
Notice that in $\Nom$ the fact that $\fnom{F}0\neq 0$ is equivalent to the existence of a \emph{finitely supported} map $s:1\to \fnom{F}0$. Thus, if $\fnom{F}$ can be extended to finitely supported maps, the proof follows the same lines as Theorem~\ref{thm:fin} with a small difference. The maps $\fnom{F}^ns$ are finitely-supported rather than equivariant. Since $\fnom{F}$ is $\Nom$-enriched the support of $\fnom{F}^ns$ is included in the support of $s$ for any finitely supported map $s$. This follows by Remark~\ref{rem:supp-shrink} since the maps $\fnom{F}_{\setone,\settwo}:[\setone,\settwo]\to[\fnom{F}\setone,\fnom{F}\settwo]$ are equivariant.
 Therefore the support of the Cauchy sequence $(x_n)_n$ is included in the support of $x$ union the support of the map $s$.
\qed

As a side observation, notice that $\fnom{F}$ being an enriched $\Nom$-functor with respect to the cartesian symmetric monoidal structure of $\Nom$ is equivalent to the existence of a strength, a notion that we will also use in Section~\ref{sec:substitution}.

One can easily show that  functors obtained from the grammar~\eqref{eq:gramF} below have the properties required in Theorems~\ref{thm:fin} and~\ref{thm:fin-v2}.

\begin{prop} \label{prop:generalizedgrammar} Endofunctors on $\Nom$
  obtained from the grammar

  \begin{equation}
    \label{eq:gramF}
    F::= \setvars\  |\  K\  |\  \mathsf{Id}\  |\  \coprod F\  |\  F\times F\  |\  [\setvars] F
  \end{equation}
  preserve monomorphisms, epimorphims
and limits of $\omega^\mathit{op}$-chains and are $\Nom$-enriched. Above $K$ denotes a constant functor, $\coprod F$ denotes
  at most countable coproducts, while $[\setvars] F$ denotes
  precomposition with the abstraction functor introduced in
  Definition~\ref{def:abstraction}.
\end{prop}

\proof It is immediate that the constant functors $\setvars$ and $K$,
the identity functor $\mathsf{Id}$ preserve all limits and colimits.  Countable
coproducts commute in $\Nom$ with limits of
$\omega^\mathit{op}$-chains and preserve monomorphisms   and epimorphisms.
Binary
products commute in $\Nom$ with all limits and preserve epimorphisms. Finally, the abstraction
functor $[\setvars](-):\Nom\to\Nom$ is both a right adjoint,
see~\cite[Theorem~4.12]{Pitts-book}, and a left adjoint, see~\cite[Theorem~4.13]{Pitts-book}. Therefore $[\setvars](-)$
preserves all limits and colimits, in particular limits of
$\omega^\mathit{op}$-chains,  monos and epis. 
The fact that $[\setvars](-)$ is $\Nom$-enriched follows from~\cite[Lemma~4.10]{PittsMGS2011notes}. The functors obtained from  products and coproducts can be easily proved to be $\Nom$-enriched, see~\cite{KellyGM:bascec}.
 \qed

As a corollary, we can apply Theorem~\ref{thm:fin} to any endofunctor
of the form~(\ref{eq:gramF}) with the additional property that a $\Nom$-morphism 
$s:1\to\fnom{F}0$ exists. Similarly we can apply Theorem~\ref{thm:fin-v2} to any functor of the form~\eqref{eq:gramF} with the additional property that $\fnom{F}0\neq 0$. 

As an aside, note that the class of  functors with the property  that $\fnom{F}0\neq 0$ is not  closed under countable products:

\begin{exa}
  Consider the functors $\fnom{F}_n:\Nom\to\Nom$ defined by
  $\fnom{F}_n(X)=P_n(\setvars)+X$ where $P_n(\setvars)$ is the nominal
  set of subsets of variables of cardinality $n$ with the pointwise
  action. Observe that for every $n$ we have $\fnom{F}_n0\neq 0$ , but
  $\prod\limits_{n<\omega} \fnom{F}_n0=0$.
\end{exa}

\subsection{Nominal Algebraic Data Types for Binding Signatures}
\label{sec:nom-colag-dt}

In this section we will introduce binding signatures 
 \cite{fiore:lics99} and we will see
how the set of finite raw terms, respectively the set of
$\alpha$-equivalence classes of terms for a binding signature can be
obtained as initial algebras for $\Nom$-functors. The results in this
section are based
on~\cite{gabb-pitt:lics99-j,PittsAM:alpsri,PittsMGS2011notes}.

\begin{defi}[Terms coming from a binding signature]
 \label{definition:bindingsignature} A \emph{binding
    signature} is a pair $(\Sigma,\mathsf{ar})$ where $\Sigma$ is a
  set of operations and $\mathsf{ar}:\Sigma\to\mathbb{N}^*$ specifies
  the binding arity of the operations.  The set of finite \emph{raw
    terms} $T_\Sigma$ for a binding signature $(\Sigma,\mathsf{ar})$
  is defined by the inductive rules

  \begin{equation}
    \label{eq:raw-terms}
    \begin{array}{ccc}
      \ProofRule{}{x\in T_\Sigma}{x\in\setvars} & \qquad & 
       \ProofRule{t_1\in T_\Sigma,\ldots, t_k\in T_\Sigma}{\opsfc{x}{k}{t}\in T_\Sigma}{\mathsf{ar}({\opsf})=n_1,\ldots, n_k}
    \end{array}
  \end{equation}
where in the second rule of~\eqref{eq:raw-terms} 
  $\overline{x_i}$  denotes a list $x_i^1,\ldots, x_i^{n_i}$ of
  of names of length $n_i$ where $n_i \geq 0$ and $1 \leq i \leq k$.
  If $k=0$, we write ${\sf op} ()$ for representing constants. 
  Note that~\eqref{eq:raw-terms} defines raw
  terms, not quotiented by $\alpha$-equivalence, but the intention
  here is to say that $\opsf$ binds the names in the list
  $\overline{x_i}$ in the scope given by the term $t_i$.
\end{defi}

\begin{exa}[$\lambda$-terms]
  \label{ex:bind-sig-lam-cal}
  The binding signature for $\lambda$-calculus consists of two
  operations $\Abs$ (abstraction) and $\App$ (application) with
  respective arities
  \begin{equation}
    \label{eq:ex-bind-sig-lam}
    \begin{array}{c}
      \mathsf{ar}(\Abs)=1,\\
      \mathsf{ar}(\App)=0,0.
    \end{array}
  \end{equation}
  The raw terms in this signature are  the finite $\lambda$-terms
   {\em not}  quotiented by $\alpha$-equivalence. 
  For example, $\lambda x.x$ is written as $\Abs (\langle x \rangle x)$
  and $( x \ y)$ as $\App (\langle \rangle x, \langle \rangle y)$. 
  Recall
  from Example~\ref{ex:lambda-nom} that the set of $\lambda$-terms
   can be equipped with an
  action of the group $\grp$ and thus turned into a nominal set.
\end{exa}

We can define inductively a $\grp$-action on raw terms for an
arbitrary binding signature
\begin{equation}
  \label{eq:act-raw-terms}
  \begin{array}[ ]{ll}
    \pi\cdot x=\pi(x) & \textrm{ if } x\in\setvars,\\
    \pi\cdot\opsfc{x}{k}{t}={\sf op}({\langle\pi\cdot\overline{x_1}\rangle}.{t_1},\ldots,\langle\pi\cdot\overline{{x_k}}\rangle.{t_k}),
  \end{array}
\end{equation}
where $\pi$ acts pointwise on the lists $x_i$, that is, if
$\overline{x_i}=x_i^1,\ldots,x_i^{n_i}$ then $\pi\cdot\overline{x_i}$
is the list $\pi\cdot x_i^1,\ldots,\pi\cdot x_i^{n_i}$. Then we have
that $(T_\Sigma,\cdot)$ is a nominal set. Notice that the support of a
term $t\in T_\Sigma$ is then the set of all variables occurring in
$t$.

\begin{defi}[Free and bound variables]
  We can inductively define the set of free variables in a term $t\in
  T_\Sigma$:
  \begin{equation}
    \label{eq:fv-in-terms}
    \begin{array}[ ]{ll}
      \fv(x)={x} & \textrm{ if } x\in\setvars,\\
      \fv(\opsfc{x}{k}{t})=\bigcup_{i=1}^k\fv(t_i)\setminus \bigcup_{i=1}^k\overline{x_i}.
    \end{array}
  \end{equation}
  The set of bound variables is defined similarly:
  \begin{equation}
    \label{eq:bv-in-terms}
    \begin{array}[ ]{ll}
      \bv(x)=\emptyset & \textrm{ if } x\in\setvars,\\
      \bv(\opsfc{x}{k}{t})=\bigcup_{i=1}^k\bv(t_i)\cup \bigcup_{i=1}^k\overline{x_i}.
    \end{array}
  \end{equation} 
\end{defi}
For the binding signature in Example~\ref{ex:bind-sig-lam-cal} we
obtain the usual definition of free and bound variable in a
$\lambda$-term.

We will now proceed to define $\alpha$-equivalence on finite raw terms
for a binding signature. We will use the fact that $T_\Sigma$ is a
nominal set. First let us recall the case of $\lambda$-calculus. On
finite $\lambda$-terms, $\alpha$-equivalence can be defined inductively
using the permutation action $\cdot:\grp\times\Lb\to\Lb$ of
Example~\ref{ex:lambda-nom}, see~\cite{gabb-pitt:lics99-j}. 

\begin{defi}[Alpha equivalence on  $\lambda$-terms]
\label{def:alphafinitelambdaterms}
Let $M, N,
M', N' \in \Lb$. The relation $=_\alpha$ is the least equivalence
relation closed under the rules:
\[
\label{eq:alpha-rules}
\begin{array}{c}
  \ProofRule{}{x =_\alpha x }{var} \qquad   \ProofRule{M=_\alpha N\quad M'=_\alpha N'  }{MM'=_\alpha NN' }{app} \\[1.5em]
  \ProofRule{(x\; z)\cdot M=_\alpha(y\; z)\cdot  N \ \  z\fresh  (x,y,  M,N) } 
  {\lambda x.M=_\alpha \lambda y.N }{abs}          
\end{array}  
\]
\end{defi}
The relation $=_\alpha$ is equivariant, that is, $M=_\alpha N$ implies
$\pi\cdot M=_\alpha\pi\cdot N$ for all $\pi\in\grp$. Thus we obtain a
nominal set $(\Lbaone,\cdot)$ where $\supp(M) = \fv(M)$.

\begin{defi}[Alpha equivalence on terms coming from a binding signature]
  \label{def:aequiv-fin-terms}
  For an arbitrary binding signature $(\Sigma,\mathsf{ar})$, we can
  define $\alpha$-equivalence on finite  terms inductively:

  \[
  \label{eq:alpha-rules-gen-sig}
  \begin{array}{c}
    \ProofRule{}{x =_\alpha x }{var}  \\[1.5em]
    \ProofRule{(\overline{x_i}\;\overline{z_i})\cdot t_i =_\alpha (\overline{y_i}\;\overline{z_i})\cdot s_i \quad 1\le i\le k \quad  \textrm{ for fresh }\overline{z_i}}{\opsfc{x}{k}{t}=_\alpha\opsfc{y}{k}{s}}{\opsf}\\  
  \end{array}  
  \vspace{0.5em}
  \]
  where each $\overline{z_i}$ is a list of distinct elements of length
  $n_i$ which are fresh for all the terms involved. Also, by
  $(\overline{y_i}\;\overline{z_i})$ we mean the composition of
  transposition $(x_i^1\;z_i^1)\ldots (x_i^{n_i}\;z_i^{n_i})$.
\end{defi}

The relation $=_\alpha$ is equivariant, that is, $t=_\alpha s$ implies
$\pi\cdot t=_\alpha\pi\cdot s$ for all $\pi\in\grp$. Thus we obtain a
nominal set $(\TSiga,\cdot)$. The equivalence class of a term $t\in
T_\Sigma$ is denoted by $[t]_\alpha$ and we have that
$\supp([t]_\alpha) = \fv(t)$.

We can express the $\alpha$-equivalence classes of terms in a binding
signature as an initial algebra for a $\Nom$-functor. The next result
allowed Gabbay and Pitts to formulate $\alpha$-structural induction
and recursion principles, see~\cite{PittsMGS2011notes,
  gabb-pitt:lics99-j,PittsAM:alpsri}.

\begin{thm}[Nominal algebraic data types]
  The nominal set $(\TSiga,\cdot)$ of finite $\alpha$-equivalence
  classes of terms is the initial algebra for the functor
  $\fnom{F}_\alpha:\Nom\to\Nom$ given by:
  \begin{equation}
    \label{eq:F-alpha}
\fnom{F}_\alpha\ \setone=\setvars+\coprod\limits_{\substack{\mathsf{op}\in\Sigma \\ \mathsf{ar}(\mathsf{op})=n_1,\ldots, n_k }}[\setvars]^{n_1}\setone\times\ldots\times[\setvars]^{n_k}\setone.    
  \end{equation}
\end{thm}\medskip

\noindent On the other hand, the nominal set $(T_\Sigma,\cdot)$ of finite raw
terms is the initial algebra for the functor $\fnom{F}:\Nom\to\Nom$
given by:
\begin{equation}
  \label{eq:eq:F-simple}
\fnom{F}\ \setone=\setvars+\!\!\!\!\!\!\!\!\!\!\!\coprod\limits_{\substack{\mathsf{op}\in\Sigma \\ \mathsf{ar}(\mathsf{op})=n_1,\ldots, n_k }}\!\!\!\!\!\!\!\!\!\!\!(\setvars^{n_1}\times\setone)\times\ldots\times(\setvars^{n_k}\times\setone).
\end{equation}
We can also obtain $T_\Sigma$ as the initial algebra for a $\Set$-functor. Indeed, we define a functor
\begin{equation}
  \label{eq:F-Set}
F:\Set\to\Set  
\end{equation}
by the same formula as~\eqref{eq:eq:F-simple}. The initial $F$-algebra
is the set $T_\Sigma$. Notice that $\fnom{F}$ is a lifting of $F$
to nominal sets in the sense that the next diagram commutes.

\begin{equation}
  \label{eq:F-extension}
 \begin{gathered}
  \xymatrix
{
\Nom\ar[r]^{\fnom{F}}\ar[d]_{U}& \Nom\ar[d]^{U}\\
\Set\ar[r]^{F}&\Set
}
\end{gathered}
\end{equation}

\begin{rem}
  Notice that both functors $\fnom{F}_\alpha$ and $\fnom{F}$ are
  obtained from grammar~\eqref{eq:gramF}.
\end{rem}

\begin{exa}
\label{exa:funct-lam-calc}
  Consider the binding signature for $\lambda$-calculus of
  Example~\ref{ex:bind-sig-lam-cal}. Then the nominal set
  $(\Lbaone,\cdot)$ is the initial algebra of the functor
  $\FLambda:\Nom\to\Nom$ already mentioned 
   in the introduction~\eqref{eq:FLambda}:
  \[
  \FLambda \ \setone = \setvars+[\setvars]\setone+ \setone\times
  \setone
  \]
  while the raw $\lambda$-terms for a nominal set which is the initial
  $\fnom{L}$-algebra for $\fnom{L}:\Nom\to\Nom$ given by
  \[
  \fnom{L} \ \setone = \setvars+\setvars\times\setone+ \setone\times
  \setone.
  \]

\end{exa}

\subsection{An abstract account of $\alpha$-equivalence.} In the above
example $\alpha$-equivalence is the kernel of the map
$\Lambda\twoheadrightarrow\Lbaone$. Using
Definition~\ref{def:abstraction} we see that this map is induced at
the level of $\Nom$-endofunctors by the natural transformation
\begin{equation}
  \label{eq:q-abs}
  \theta_\setone:\setvars\times \setone\twoheadrightarrow [\setvars]\setone
\end{equation}
defined by $(x,\elone)\mapsto\langle x\rangle\elone$.

Consider an arbitrary $\fnom{F}_\alpha:\Nom\to\Nom$ obtained from
grammar~(\ref{eq:gramF}) or, equivalently, from a binding signature $\Sigma$. To $\fnom{F}_\alpha$  we can associate a functor $\fnom{F}:\Nom\to\Nom$ in which all occurrences of the abstraction functor $[\setvars](-)$ are replaced by $\setvars\times -$. Therefore $\fnom{F}$ is a polynomial functor and the natural transformation~(\ref{eq:q-abs}) induces a
natural transformation $q:\fnom{F}\to \fnom{F}_\alpha$ given by the following inductive rules:
\begin{equation}
  \label{eq:8}
  \begin{array}{c}
    \rules{\fnom{F}_\alpha=\mathsf{Id}}{\fnom{F}=\mathsf{Id}\quad q=\mathit{id}}
\qquad 
   \rules{\fnom{F}_\alpha=\setvars}{\fnom{F}=\setvars\quad q=\mathit{id}}
\qquad
   \rules{\fnom{F}_\alpha=\mathsf{K}}{\fnom{F}=\mathsf{K}\quad q=\mathit{id}} \\[1.5em]
    \rules{\fnom{F}_\alpha=[\setvars]\fnom{F'_\alpha}\qquad q':\fnom{F'}\to\fnom{F'_\alpha}}{\fnom{F}=\setvars\times\fnom{F'}\qquad q=\theta\circ(\setvars\times q')}\\[1.5em] 
 \rules{\fnom{F}_\alpha=\fnom{F'_\alpha}\times\fnom{F''_\alpha}\qquad q':\fnom{F'}\to\fnom{F'_\alpha}\qquad q'':\fnom{F''}\to\fnom{F''_\alpha}}{\fnom{F}=\fnom{F'}\times\fnom{F''}\qquad q=q'\times q''} \\[1.5em] 
   \rules{\fnom{F}_\alpha=\coprod\fnom{(F_i)_\alpha}\qquad q_i:\fnom{F_i}\to\fnom{(F_i)_\alpha}}{\fnom{F}=\coprod\fnom{F_i}\qquad q=\coprod q_i}
  \end{array}
\end{equation}
Using the initial chain, this in turn gives a surjective map from the
initial $\fnom{F}$-algebra $\ialg=(T_\Sigma,\cdot)$ to the initial
$\fnom{F}_\alpha$-algebra $\ialga=(\TSiga,\cdot)$.  Indeed, we have the
following diagram
\begin{equation}
  \label{eq:init-chain-surj}
  \begin{gathered}
    \xymatrix {
      0\ar[r]\ar[d]^{\amap{0}} & \fnom{F}0 \ar[r]\ar@{->>}[d]^{\amap{1}} & \fnom{F}^20 \ar[r]\ar@{->>}[d]^{\amap{2}} & \cdots  \ar[r] & \ialg\ar@{-->>}[d]^{[-]_\alpha} \\
      0\ar[r] & \fnom{F}_\alpha0 \ar[r] & \fnom{F}_\alpha^20 \ar[r] &
      \cdots \ar[r] & \ialga }
  \end{gathered}
\end{equation}
where $\amap{0}= id_{0}$  and 
$\amap{n+1}=q_{\fnom{F}_\alpha^n0}\circ\fnom{F}\amap{n}.$  The horizontal arrows are defined as in~\eqref{eq:in-chain-prel}.

Since the functors
$\fnom{F}$ obtained from grammar~\eqref{eq:gramF} preserve
surjections, we have that all the maps $\amap{n}$ are
surjective. The compositions
$\fnom{F}^n0\twoheadrightarrow\fnom{F}_\alpha^n0\to \ialga$ form a
cocone, thus using the universal property of the colimit $\ialg$ we
obtain a unique map $[-]_\alpha:\ialg\to\ialga$ such that
diagram~\eqref{eq:init-chain-surj} commutes. To show that $[-]_\alpha$
is surjective assume $v,w:\ialga\to B$ are two equivariant maps such
that $v[-]_\alpha=w[-]_\alpha$. It follows 
that $v$ and $w$ equalise
$\fnom{F}^n0\twoheadrightarrow\fnom{F}_\alpha^n0\to \ialga$ for all $n$.
 Using the universal property of $\ialga$, it follows that $v=w$. Hence
$[-]_\alpha$ is surjective.

\begin{rem}
  We had to explain why the maps $\amap{n}$ are surjections. In
  $\Set$ this follows easily, since all $\Set$-functors preserve
  surjections (we just have to consider a right inverse). In $\Nom$ we
  have equivariant surjective maps that do not have right
  inverses. E.g.\ the unique map from $\setvars$ to the final nominal
  set $\{*\}$ has no equivariant right inverse, because the support of
  $*$ is empty whereas all elements of $\setvars$ have non-empty
  support. 
\end{rem}

  Given a binding signature $(\Sigma,\mathsf{ar})$ we described $\alpha$-equivalence on finite terms in two  equivalent ways:
  \begin{itemize}
  \item \emph{syntactically},  as in Definition~\ref{def:aequiv-fin-terms};
  \item \emph{semantically},  via initial algebras, as in~\eqref{eq:init-chain-surj}.
  \end{itemize}

While the latter presentation of $\alpha$-equivalence  may seem rather pedantic, the abstract perspective sheds some light on the problems one has with defining  $\alpha$-equivalence and finding representatives for $\alpha$-equivalence classes in the infinitary case.

\subsection{Problems with Alpha Equivalence in the Infinitary Case} 
\label{sec:ilcinnom}

As an illustration, let us first look at possible definitions of $\alpha$-equivalence for the infinitary $\lambda$-calculus. We see that there are two possible ways of defining $\alpha$-equivalence classes and, contrary to all expectations, they are not equivalent in case of a countable set of variables $\setvars$.

Recall that the set of raw infinitary $\lambda$-terms $\Lib$ is the final coalgebra of the $\Set$-functor in~\eqref{eq:set-funct-lambda}. We define $\alpha$-equivalence on the set $\Lbi$
using truncations.  This definition is slightly
different from, though equivalent, to those used
in~\cite{kenn:klop:slee:vrie:1995b,KKSV97,KV03}. 

\begin{defi}[Alpha equivalence on infinitary $\lambda$-terms]\label{def:alpha-inf}
  We extend the notion of $\alpha$-conversion to the
  set   $\Lib$ via
  \begin{center}
    $ M \alphaconv N $ iff $M^n \alphaconv N^n$ for all $n \in \nat$.
  \end{center}
We thus obtain the quotient $$\Lbia.$$
The notion of truncation can be extended to $\Lbia$ via $\eqclassa{M}^n = \eqclassa{M^n}$.
\end{defi}

A second approach to define the set of $\alpha$-equivalence classes of
infinitary terms is to consider the metric completion of the quotient
$\Lbaone$.

\begin{defi}[Metric on $\alpha$-equivalence classes]
  We define $d_\alpha: \Lb \times \Lb \to [0,1]$ via
  \begin{equation}
    \label{eq:dalpha}
    d_\alpha(M, N) = \inf\{ 2^{-n} \mid M^n =_\alpha N^n,\ n \in \nat\}.
  \end{equation}
\end{defi}
 
We have that $d_\alpha$ is a pseudometric on $\Lb$ and
$d_\alpha(M,N)=0 $ if and only if $M=_\alpha N$. Thus $d_\alpha$ gives
rise to a metric on $\Lbaone$ denoted by abuse of notation also by
$d_\alpha$. 

We consider the metric completion 
of $\Lbaone$ with respect to $d_\alpha$, denoted by
$$\clba.$$
Observe that $d_{\alpha}$ extends to a pseudometric
$\dab:\Lbi\times\Lbi\to [0,1]$ given by the same formula as in~\eqref{eq:dalpha}.
Then $M=_\alpha N$ in the sense of
Definition~\ref{def:alpha-inf} if and only if $\dab(M,N)=0$. Hence
we obtain a metric  on $\Lbia$, also denoted by $\dab$.
\begin{thm}\label{thm:uncountable-vars}
\label{theorem:metricompletionofLba}
 Let  $\setvars$ be uncountable. Then we have that 
  $(\Lbia,\dab)$ is  isomorphic to   $\clba$. 
\end{thm}

We do not include the proof of this theorem, since in this paper we are
only interested in the case where $\setvars$ is countable. The idea of
the proof is to show that $(\Lbia,\dab)$ is a complete metric space
and then to use the universality property of the metric
completion. This argument fails when the set of variables is at most
countable. Indeed, we can show that for countable $\mathcal{V}$ the
space $(\Lbia,\dab)$ is not complete.

\begin{exa}\label{ex:alpha-inf-lam}
  Assume that $\mathcal{V}$ is countable, say
  $\mathcal{V}=\{x_0,x_1,\ldots\}$ and consider the sequence
  $([\lambda x_n.x_n(x_0(x_1(\ldots x_{n-1})))]_\alpha)_{n\ge 1}$ in
  $\Lbia$. This is a Cauchy sequence with respect to $\dab$, but has
  no limit in $\Lbia$. Indeed, assume towards  a contradiction that the
  limit $\ell$ exists.  On  one hand we can prove that
  $\fv(\ell)=\mathcal{V}$, on the other hand $\ell$ should be of the form
$[\lambda u.u(x_0(x_1\ldots))]_{\alpha}$ 
  for some variable $u$. But this contradicts
  the fact that $u$ is free in $\ell$.
\end{exa}

The example shows that with a countable $\setvars$ the two
possible definitions of infinitary $\lambda$-terms up to
$\alpha$-equivalence do not coincide. In other words, metric
completion and quotienting by $\alpha$ do not commute. We find the
following formulation of this phenomenon useful as well.

\begin{rem}
  \label{ex:lam-inf-not-surj}
The canonical map $$[-]_\alpha:\Lbi\rightarrow\clba$$ taking the $\alpha$-equivalence class of an infinitary  $\lambda$-term is not  surjective. 
\end{rem}

These problems of $\alpha$-equivalence in the presence of countably
many variables disappear if we consider the set $\Lbif$ of infinitary terms with
finitely many free variables.  

\begin{notation}[Restriction to finitely many free variables]
\label{notation:restrictiontoffv}
Let $\Lbif$ denote the set $\{ M \in \Lbi \mid
\fv(M) \mbox{ is finite } \}$.
\end{notation}

  \begin{rem} Note that $\Lbif$ is different from the set
    $(\Lbi)_\mathsf{fs}=\Lbifs$, defined in (\ref{eq:fs}), of
    $\lambda$-terms with finitely many variables, bound or free. We do have that $\Lbifs\subseteq\Lbif$, but the inclusion is strict. Indeed, the following two terms belong to $\Lbif$  but not to $\Lbifs$:
    \begin{itemize}
    \item $\ogre \equiv \lambda x_1.\lambda x_2.\lambda x_3\dots$
    \item $\ibvnf \equiv  \lambda x_0.  \lambda  x_1. x_0  x_1  
(\lambda  x_2. x_0  x_1  x_2 (\lambda x_3. x_0 x_1 x_2 x_3 (\ldots)) )
$.
   \end{itemize}
   For $\ogre$ we can find $N \in \Lbifs$ such that $N \alphaconv
   \ogre$, e.g.\ $N \equiv \lambda x_1.\lambda x_1.\lambda x_1\dots$.
    For $\ibvnf$ this is not possible.
\end{rem}
   
 \begin{rem}
  The equivalence relation $=_\alpha$ of
  Definition~\ref{def:alpha-inf} restricts to $\Lbif$, but not to
  $\Lbifs$, as shown by $\ogre =_\alpha\lambda x_1.\lambda x_1.\lambda x_1\dots$ where only the latter term is in $\Lbifs$.
\end{rem}

For all $M,N\in\Lbif$ we have that $M=_\alpha N$ implies
  $\pi\cdot M=_\alpha \pi\cdot N$. Hence we can equip $\Lbifaone$ with a $\grp$-action given by 
$\pi \cdot \echiv{M}=\echiv{ \pi \cdot M}$ and we can easily check that 
$(\Lbifaone,\ \cdot)$ is a nominal set. 
Indeed, $\echiv{M}\in\Lbifaone$ 
is supported by the finite set $\fv(M)$. 

The permutation action on $\Lbaone$ can be extended to $\clba$ as
follows. For each $\pi\in\grp$ we have that $\pi\cdot
(-):\Lbaone\to\Lbaone$ is a uniformly continuous function with respect to 
$d_\alpha$. Using the universal property of the metric completion, this function can be extended to a uniformly continuous map on $\clba$:

  \[
  \label{eq:pi-on-Lbi}
  \xymatrix@R=15pt{
    \Lbaone \ar[dr]_{\pi\cdot(-)}\ar[rr] &   & \clba  \ar[dd]^{\pi\cdot(-)}\\
    & \Lbaone \ar[dr] & \\
    & & \clba}
  \]
  Thus we have a nominal set
  $(\cflba,\cdot)$. In~\cite{kurzetal:cmcs2012} we showed that this
  nominal set is isomorphic to $(\Lbifaone,\ \cdot)$ and to the carrier of the
  final coalgebra of the $\Nom$-functor $\FLambda$. Hence, for each
  $\alpha$-equivalence class of infinitary terms with finitely many
  variables we can find a representative. This means that we have a
  surjective map
\begin{equation}
  \label{eq:alpha-qoutient-ilb}
  [-]_\alpha:\Lbif\to\cflba
\end{equation}
whose kernel is the $\alpha$-equivalence relation on $\Lbif$.

\begin{rem}
  \label{rem:lbifs-to-cflba-not-surj}
On the other hand, the restriction $\Lbifs\to\cflba$ to $\Lbifs$ \ of \eqref{eq:alpha-qoutient-ilb} is not surjective. For example the equivalence class of 
$\ibvnf\in\Lbif$ is  obtained as the limit of a 
finitely supported (actually emptily supported) sequence:

\[
\begin{array}{l}
  \eqclassa{\lambda x_0. * }\\
  \eqclassa{\lambda x_0 x_1. *  }\\
  \eqclassa{\lambda x_0 x_1. x_0 * }\\
  \ldots
\end{array}
\]
and thus belongs to  $\cflba$. 
However, for the term $\ibvnf$ there is no
$N \in \Lbifs$ such that $N \alphaconv \ibvnf$.
\end{rem}

To summarise, we have seen the following classes of $\lambda$-terms
\begin{equation}
  \label{eq:classeslambda-2}
\begin{gathered}
  \xymatrix@R=25pt@C=30pt
{
\Lambda\  \ar@{>->}[r]\ar@{->>}^{}[dd]
& \Lbi\ar@{->>}^{}[d]  
&   \ \Lbif\ar@{->>}^{}[d]\ar@{_(->}[l]
&\Lbifs\ar@{_(->}[l]\ar[ddl] \\
& \Lbia\ar@{>->}[d] 
& \Lbifaone\ar@{>->>}_{\cong}[d]\ar@{}[l]|{}\\
\Lbaone\ \ar@{>->}[r] 
& \clba 
& \ \cflba\ar@{_(->}[l] \ \ 
}
\end{gathered}
\end{equation}
with inclusions, injections and surjections as indicated by the arrows. $\Lambda$
and $\Lbaone$ are initial algebras for 
 \fnom{L} 
and $\FLambda$,
respectively. Similarly $\Lbifs$ and $\cflba$ are final coalgebras
for the same functors. But the coinductive situation is complicated by
the fact that the canonical map $\Lbifs\to\cflba$ is not onto, that
is, infinitary $\lambda$-terms upto $\alpha$-equivalence do not arise
by quotienting, in $\Nom$, the raw infinitary $\lambda$-terms (which
do allow only finitely many bound variables). Instead of $\Lbifs$ we
need to work with $\Lbif$, which is not a final coalgebra. That $\Lbif$ can
be given a semantic characterisation is one of the topics of the
next subsection.

\subsection{Nominal Coalgebraic Data Types for Binding Signatures}
\label{sec:nominalcodatatypes}

The aim of this section is to introduce nominal coalgebraic data types
in their generality.  We will generalise the previous subsection to
arbitrary binding signatures and give semantic characterisations of
all the vertices of  \eqref{eq:classeslambda-2}.

In particular, at the end of the section, we will have explained the
following diagram, which generalises \eqref{eq:classeslambda-2}
(eliding the middle row of \eqref{eq:classeslambda-2} obtained by
epi-mono factorisations).
\begin{equation}
  \label{eq:classesbindsig-2}
 \begin{gathered}
  \xymatrix@C=30pt
{
\TSig\  \ar@{>->}[r]\ar@{->>}^{}[d]
& \TSigi \ar[d]
&   \ \Tiffv\ar@{->>}^{}[d]\ar@{_(->}[l]\ar@{}[dl]|{}
& \Tfinal\ar@{_(->}[l]\ar[dl]
\\
\TSiga\ \ar@{>->}[r] 
&\TSigai\ 
&   \ \Tfinala\ar@{_(->}[l] \\
}
\end{gathered}
\end{equation}
\medskip 

\noindent Recall from Section~\ref{sec:nom-colag-dt} the definition of
$\TSig$ and $\TSiga$. Also recall that $\TSig$ and $\TSiga$ are 
initial algebras of the $\Nom$-endofunctors $\fnom{F}$ and $\fnom{F}_\alpha$, (see respectively 
\eqref{eq:eq:F-simple} and~\eqref{eq:F-alpha}).
The transformation $\xymatrix{\TSig\ar@{->>}[r]&\TSiga}$  induced by a natural
transformation $\fnom{F}\to\fnom{F}_\alpha$ is quotienting by
$\alpha$-equivalence. Next, we define $\TSigi$.

\begin{defi}[Infinitary terms coming from a binding signature]
\label{def:TSigi}
 Consider a binding signature $(\Sigma,\mathsf{ar})$. 
 \begin{enumerate}
 \item The set of infinitary raw terms $\TSigi$ is defined coinductively by
  \begin{equation}
    \label{eq:inf-raw-terms}
    \begin{array}{ccc}
      \ProofRule{}{x\in \TSigi}{x\in\setvars} & \qquad & \ProofRule{t_1\in \TSigi,\ldots, t_k\in \TSigi}{\opsfc{x}{k}{t}\in \TSigi}{\mathsf{ar}({\opsf})=n_1,\ldots, n_k}
    \end{array}
  \end{equation}
\item Truncation of raw terms at depth $n$ is defined by induction on $n$:
 \begin{equation}
     \label{eq:truncation-gen-term}
     \begin{array}{ll}
       t^0 & = * \\\\
       t^{n+1}&  = \left \{ \begin{array}{ll}
           x & \textit{ if $t=x\in \mathcal{V}$}\\
           \opsfc{x}{k}{t^n} & \textit{ if $t=\opsfc{x}{k}{t}$} \\
         \end{array} \right . 
     \end{array}
   \end{equation}
where $\{*\}$ is a terminal object in $\Nom$.
\item To define $\alpha$-equivalence, let $t$ and $s$ be two infinitary raw terms in $\TSigi$. We say that $t=_\alpha s$ when the truncations at all depths are $\alpha$-equivalent in the sense of Definition~\ref{def:aequiv-fin-terms}, that is,  for all $n$ we have $t^n=_\alpha s^n$. 
\item The sets  $\fv(t)$ and $\bv(t)$ of free and bound
  variables of an infinitary raw term $t\in\TSigi$ are defined 
 as follows.
  \[ \begin{array}{lll} \fv(t) = \bigcup_{n \in \nat} \fv(t^n)& \ \ \ &
  \bv(t) = \bigcup_{n \in \nat} \bv(t^n).
\end{array}\]
\end{enumerate}
\end{defi}

\begin{rem}
$\TSigi$ is the final coalgebra for the $\Set$-functor defined in~\eqref{eq:F-Set}.
\end{rem}

\begin{defi}
We denote by $(\TSiga)^\infty$ the metric completion of  $\TSiga$ with respect to the metric
  $d_\alpha$ given by
$$d_\alpha([t]_\alpha,[s]_\alpha)=\inf \{ 2^{-n} \mid t^n =_\alpha s^n,\ n \in \nat\}.$$
\end{defi}

Notice that $(\TSiga)^\infty$ is equipped with a canonical permutation
action, but it is not a nominal set, since not all elements are
finitely supported (namely those terms with infinitely many free
variables).

\begin{rem}
  Going back to Section~\ref{sec:alg-coalg}, and in the notation of
  \eqref{eq:init-chain-surj}, we have that
  $\fgt\ialg\cong\fgt\fnom{F}^\omega0\cong\TSig$ and
  $\fgt\ialga\cong\fgt\fnom{F}_\alpha^\omega0\cong\TSiga$. The completions
  $\TSigi$ and $(\TSiga)^\infty$ then can be obtained as limits of
  $\omega^\mathrm{op}$-sequences:
\begin{equation}
  \label{eq:fin-chains-sets}
 \begin{gathered}
  \xymatrix
{
1\ar@{<-}[r]\ar@{>->>}[d] & U\fnom{F}1 \ar@{<-}[r]\ar@{->>}[d]^{\fgt\amap{1}} & U\fnom{F}^21 \ar@{<-}[r]\ar@{->>}[d]^{\fgt\amap{2}} & \cdots  \ar@{<-}[r] & \lim U\fnom{F}^n1\ar[d]^{[-]_\alpha} \cong \TSigi\\ 
1\ar@{<-}[r] & U\fnom{F}_\alpha1 \ar@{<-}[r] & U\fnom{F}_\alpha^21 \ar@{<-}[r] & \cdots  \ar@{<-}[r] & \lim U\fnom{F}_\alpha^n1 \cong (\TSiga)^\infty
}
\end{gathered}
\end{equation}
Moreover, the horizontal arrows in the diagram are precisely the
truncations (with $1=\{*\}$) and the kernels of the vertical arrows
capture the $\alpha$-equivalence of Definition~\ref{def:TSigi}.
Since all maps $\amap{n}:\fnom{F}^n1\to\fnom{F}_\alpha^n1$ are
surjective, each element of $\fnom{F}_\alpha^n1$ can be expressed as
the $\alpha$-equivalence class of a raw term $t\in\fnom{F}^n1$. If
$[t]_\alpha\in\fnom{F}_\alpha^n1$ we have that
$\supp([t]_\alpha)=\fv(t)$.
\end{rem}

The set  $\lim\fgt\fnom{F}_\alpha^n1\cong (\TSiga)^\infty $ appears to
be a natural domain for infinitary terms up to $\alpha$-equivalence
and it will reappear in
Section~\ref{section:infinitelymanyfreevariables}. But it fails to
have desirable properties.  Indeed, it is not a final coalgebra of
a $\Set$-functor in any obvious way, nor is it a nominal set.
Moreover the map
$[-]_\alpha:\TSigi\to\TSigai$
is not surjective in general, as shown
by Example~\ref{ex:alpha-inf-lam}.
In the case of the infinitary $\lambda$-calculus we solved this issue
by restricting our attention to terms with finitely many free
variables \cite{kurzetal:cmcs2012}.  We do the same in the case of a
general binding signature.

\medskip Recalling from \eqref{eq:fs} the notation $(-)_{\sf fs}$,
we now obtain from the rightmost edge of \eqref{eq:fin-chains-sets}
\begin{equation}
\Tfinal \quad \textrm{ and } \quad \Tfinala.
\end{equation}
In the first case we restrict to finitely many variables and in the
second case to finitely many free variables. More precisely,
$\Tfinala$ consists of limits of Cauchy sequences $([t_n]_\alpha)_n$ of
$\alpha$-equivalence classes of finite terms which altogether have
only finitely many free variables, that is,
$\bigcup\limits_{n}\fv(t_n)$ is finite. Similarly $\Tfinal$ consists of 
limits of Cauchy sequences $(t_n)_n$ of finite terms which altogether
have only finitely many variables.

\begin{rem}
  According to Theorem~\ref{thm:fin-v2}, we have that $\Tfinal$, respectively
  $\Tfinala$, can be taken to be the final $\fnom{F}$-coalgebra
  $\fcoalg$, respectively the final $\fnom{F}_\alpha$-coalgebra
  $\fcoalg_\alpha$. 
\end{rem}

  Just as in the case of the initial chains
  (see~\eqref{eq:init-chain-surj}), the natural transformation
  $q:\fnom{F}\to\fnom{F}_\alpha$ defined by~\eqref{eq:8} induces a
  unique map $\fcoalg\to\fcoalg_\alpha$ from the final
  $\fnom{F}$-coalgebra to the final $\fnom{F}_\alpha$-coalgebra.
  However, unlike in the initial algebra
  situation~\eqref{eq:init-chain-surj},  the induced map $[-]_\alpha$
  is not surjective in general, see
  Remark~\ref{rem:lbifs-to-cflba-not-surj}. 
  The aim of the remainder of this section is to prove that, nevertheless,
  $\fcoalg_\alpha$ is the quotient by $\alpha$-equivalence of the
  infinitary terms with finitely many free variables for which we
  introduce the following notation.

\begin{defi}
  Denote by 
  \begin{equation}
    \tffv
  \end{equation}
 the set of elements of $\TSigi$ having only finitely many free variables.
\end{defi}

Notice that $\tffv$ consists of limits of Cauchy sequences $(t_n)_n$ of finite terms which altogether have only finitely many free variables. The definition above relies on the syntactic notion of free variable. Proposition~\ref{prop:tffv-pullback} shows that a semantic definition is
possible. To this end, we first give a semantic definition of the inclusion $\xymatrix{\Tfinala\ar@{^(->}[r]&\TSigai}$ as the map $\iota_\alpha$ arising in~\eqref{eq:fin-chains-sets-fgt}.

\begin{rem} 
  As a final $\fnom{F}$-coalgebra, $\fcoalg$ induces a cone over the
  sequence $(\fnom{F}^n1)_{n<\omega}$, dualising~\eqref{eq:init-chain-surj}.

This induces a unique map $\iota:\fgt\fcoalg\to\lim\fgt\fnom{F}^n1$.
In the same way, by finality of $\fcoalg_\alpha$, one obtains
$\iota_\alpha:\fgt\fcoalg_\alpha\to\lim\fgt\fnom{F}_\alpha^n1$.%
\begin{equation}
  \label{eq:fin-chains-sets-fgt}
 \begin{gathered}
  \xymatrix
{
1\ar@{<-}[r]\ar[d] & \fgt \fnom{F}1 \ar@{<-}[r]\ar@{->>}[d] & \fgt \fnom{F}^21 \ar@{<-}[r]\ar@{->>}[d] & \cdots  \ar@{<-}[r] & \lim \fgt \fnom{F}^n1\ar@{<-}[r]^-\iota\ar[d]_{[-]_\alpha} & \fgt \fcoalg\ar^{\fgt [-]_\alpha}[d]\ar@{}[dl]|{(*)} \\
1\ar@{<-}[r] & \fgt \fnom{F}_\alpha1 \ar@{<-}[r] & \fgt \fnom{F}_\alpha^21 \ar@{<-}[r] & \cdots  \ar@{<-}[r] & \lim \fgt \fnom{F}_\alpha^n1\ar@{<-}[r]_-{\iota_\alpha}& \fgt \fcoalg_\alpha
}
\end{gathered}
\end{equation}
\end{rem}

\begin{prop}
\label{prop:tffv-pullback}
  The set $\tffv$ of infinitary raw terms with finitely many free variables is the pullback of the maps $\iota_\alpha:\fgt\fcoalg_\alpha\to\lim \fgt \fnom{F}_\alpha^n1$ and $[-]_\alpha:\lim \fgt \fnom{F}^n1\to\lim \fgt \fnom{F}_\alpha^n1$.
\begin{equation}
  \label{eq:tffv-as-pullback}
 \begin{gathered}
\xymatrix
{
 \lim U\fnom{F}^n1\ar[d]_-{[-]_\alpha} &\tffv \ar[d]^{[-]_\alpha}\ar@{_{(}->}[l]\SWpullbackcorner  \\
\lim U\fnom{F}_\alpha^n1&U\fcoalg_\alpha\ar[l]_-{\iota_\alpha}
}
\end{gathered}
\end{equation}
\end{prop}

\proof First we have to define the map $[-]_\alpha:\tffv\to
U\fcoalg_\alpha$.  Given $t\in\tffv$, notice that we can construct a
finitely supported sequence in the sets $\fnom{F}_\alpha^n1$ by
applying the map $\amap{n}$ to the truncations at depth $n$ of
$t$. Indeed, each $[t^n]_\alpha^{(n)}$ is supported by the finite set
of free variables in $t$. We define $[t]_\alpha\in\fcoalg_\alpha$ to
be the unique element whose projection in $\fnom{F}_\alpha^n1$ is
exactly $[t^n]_\alpha^{(n)}$ for all $n$. We can easily check that the
square~\eqref{eq:tffv-as-pullback} commutes.

Consider a pair $(t,s)$ with $t\in\lim U\fnom{F}^n1\simeq \TSigi$ and
$s\in U\fcoalg_\alpha$ such that $\iota_\alpha(s)=[t]_\alpha$. This
implies that $[t^n]_\alpha^{(n)}$ is equal to the projection of $s$
into $\fnom{F}_\alpha^n1$, and thus the free variables of each
truncation $t^n$ are contained in the finite set that supports
$s$. Therefore $t$ has finitely many free variables, that is, 
$t\in\tffv$ and $s=[t]_\alpha$. So $\tffv$ is indeed a
pullback.  \qed

To summarise, we are now ready to give a semantic version of
diagram~\eqref{eq:classesbindsig-2} -- which we set about to prove at the beginning of the section:
\begin{equation}
  \label{eq:classes-sem}
 \begin{gathered}
  \xymatrix@C=30pt
{
\Iterm \  \ar@{>->}[r]\ar@{->>}^{}[d]
& \lim U\fnom{F}^n1 \ar[d]
&   \ P\ar@{->>}^{}[d]\ar@{_(->}[l]\SWpullbackcorner 
&\ \fcoalg\ar@{_(->}[l]\ar[dl]
\\
\Iterm_\alpha\ \ar@{>->}[r] 
&\lim U\fnom{F}_\alpha^n1\ 
&   \ \fcoalg_\alpha\ar@{_(->}[l] \\
}
\end{gathered}
\end{equation}
with $\Iterm$ and $\Iterm_\alpha$ the initial algebras as well as
$\fcoalg$ and $\fcoalg_\alpha$ the final (or terminal) coalgebras of
$\fnom{F}$ and $\fnom{F}_\alpha$, and $P$ being a
pullback in $\Set$. To improve readability we omitted writing the forgetful functor $U:\Nom\to\Set$ in~\eqref{eq:classes-sem}.

The map $\xymatrix{\fcoalg\ar@{^(->}[r]&P}$ above is obtained using the universal property of $P$ and the commutativity of the square $(*)$ in~\eqref{eq:fin-chains-sets-fgt} and corresponds to the inclusion $\xymatrix{\Tfinal\ar@{^(->}[r]&\Tiffv}$, (see Proposition~\ref{prop:tffv-pullback}).
\begin{thm}\label{thm:comp-alpha-comm}
  Completing $\TSig$ by Cauchy sequences $(t_n)_n$ such that
  $\bigcup\limits_n\fv(t_n)$ is finite and quotienting by $\alpha$-equivalence
  commute. This means that 
  \begin{itemize}
  \item the two equivalent diagrams below commute
\begin{equation}
\label{eq:metr-compl-alpha-comm}
\begin{gathered}
  \begin{array}{ccc}
   \xymatrix{
U\ialg\ar@{->>}[d]_-{[-]_\alpha}\ar[r] & \tffvs \ar@{->>}[d]^-{[-]_\alpha} \\
U\ialg_\alpha\ar[r] & U\fcoalg_\alpha
}  
&
\qquad
&
   \xymatrix{
\TSig\ar@{->>}[d]_-{[-]_\alpha}\ar[r] &\tffv \ar@{->>}[d]^--{[-]_\alpha}\\
\TSiga\ar[r] & \Tfinala
}
\\
\textrm{semantic version}
&
&
\textrm{syntactic version;}
\end{array}
\end{gathered}
\end{equation}
\item the map $[-]_\alpha:\tffvs\to\fgt\fcoalg_\alpha$ is
surjective, or equivalently each element in $\Tfinala$ can be represented
as an equivalence class of an infinitary raw term with finitely many
free variables.
\end{itemize}
\end{thm}

\proof The first bullet is easier to prove. We show that the semantic version of~\eqref{eq:metr-compl-alpha-comm} commutes, using the commutativity of
\begin{equation}
  \label{eq:6}
  \begin{gathered}
     \xymatrix{
\ialg\ar@{->>}[d]_-{[-]_\alpha}\ar[r] & \fcoalg \ar[d]^-{[-]_\alpha} \\
\ialg_\alpha\ar[r] & \fcoalg_\alpha
}
  \end{gathered}
\end{equation}
The argument uses the finality of $\fcoalg_\alpha$ and the fact that all the four arrows in~\eqref{eq:6} are $\fnom{F_\alpha}$-coalgebra morphisms.

By pasting the right-hand triangle of~\eqref{eq:classes-sem}, we obtain the commutativity of the desired diagram:
\begin{equation}
  \begin{gathered}
     \xymatrix{
U\ialg\ar@{->>}[d]_-{[-]_\alpha}\ar[r] & U\fcoalg \ar[d]_-{[-]_\alpha}\ar@{^(->}[r] & P\ar[dl]\\
U\ialg_\alpha\ar[r] & U\fcoalg_\alpha
}
  \end{gathered}
\end{equation}
The second part of the theorem, stating that the map
$P\to\fgt\fcoalg_\alpha$ in
\begin{equation}
  \label{eq:PTonto}
 \begin{gathered}
  \xymatrix
{
1\ar@{<-}[r]\ar[d] 
& \fgt \fnom{F}1 \ar@{<-}[r]\ar@{->>}[d] 
& \fgt \fnom{F}^21 \ar@{<-}[r]\ar@{->>}[d] 
& \cdots  \ar@{<-}[r] 
& \lim \fgt \fnom{F}^n1\ar@{<-}[r]\ar[d]_{} 
& P\ar^{}[d]\ar@{}[dl]|{} \\
1\ar@{<-}[r] 
& \fgt \fnom{F}_\alpha1 \ar@{<-}[r] 
& \fgt \fnom{F}_\alpha^21 \ar@{<-}[r] 
& \cdots  \ar@{<-}[r] 
& \lim \fgt \fnom{F}_\alpha^n1\ar@{<-}[r]_-{\iota_\alpha}
& \fgt \fcoalg_\alpha
}
\end{gathered}
\end{equation}
is surjective can be proved by going back to the syntax as in
\cite{kurzetal:cmcs2012}, just that this time, due to generalising
from $\lambda$-calculus to binding signatures, the notation becomes
even heavier and quite unpleasant. Therefore, we will give a semantic
proof in the next section, so that surjectivity becomes a consequence
of Theorem~\ref{thm:bindingsig-surj}. This, in turn, is a consequence of a more general
result Theorem~\ref{thm:gen-seq-pb-surj} about limits of sequences in nominal sets and will be proved in
the next section. \qed

We state explicitly the most important consequence of the theorem as
a corollary.

\begin{cor}[Nominal coalgebraic data types]
\label{cor:isomfinalcoalgcarriers}
The nominal set $(\TSigia ,\cdot)$ of  $\alpha$-equivalence
  classes of infinitary terms with finitely many free variables
  is the final coalgebra for the functor
  $\fnom{F}_\alpha:\Nom\to\Nom$ corresponding to a binding signature $(\Sigma,\mathsf{ar})$:
  \begin{equation}
    \label{eq:F-alpha-2}
\fnom{F}_\alpha\ \setone=\setvars+\coprod\limits_{\substack{\mathsf{op}\in\Sigma \\ \mathsf{ar}(\mathsf{op})=n_1,\ldots, n_k }}[\setvars]^{n_1}\setone\times\ldots\times[\setvars]^{n_k}\setone.    
  \end{equation}
\end{cor}

\begin{rem}
  Let us point out that Diagram~\eqref{eq:classes-sem} does not
  actually depend on the functors $\fnom{F},\fnom{F}_\alpha$ arising
  from a binding signature and makes sense for any pair of
  $\Nom$-endofunctors $\fnom{F},\fnom{F}_\alpha$ and any component-wise
  surjective natural transformation $\fnom{F}\to\fnom{F}_\alpha$
  subject to some natural conditions (which are satisfied by functors
  that do arise from binding signatures), namely that $\fnom{F}$
  preserves surjections and that both $\fnom{F},\fnom{F}_\alpha$ have initial algebras and final
  coalgebras. Finally, we want the surjectivity of
  $P\to\fcoalg_\alpha$ and we give a semantic analysis of it in
  the next subsection.
\end{rem}

\subsection{Presenting Limits in Nominal Sets}
\label{sec:limits-in-Nom}

The motivation of this section is to give a semantic proof of the fact
that the final $\fcoalg_\alpha$-coalgebra is the quotient by
$\alpha$-equivalence of $\tffv$, the infinitary terms with finitely
many free variables. We move this proof into a separate subsection
because the  semantic analysis depends on  certain facts on limits in nominal sets and leads 
to a novel notion of `bound variable relative to a
map' which may be of independent interest.

\subsubsection{Bound variables, safe maps, and safe squares}\ \\

\noindent
In nominal sets, the syntax dependent notion of free
variable is replaced by the semantic concept of minimal finite support. What about bound variables? Consider $[-]_{\alpha}: \Lambda \rightarrow \Lambda/\!\alphaconv$ and $x(\lambda y.y)\in\Lambda$. Then the bound variables of $x\lambda y.y$ can be computed as $\supp(x\lambda y.y)\setminus\supp([x\lambda y.y]_{\alpha})=\{x,y\}\setminus\{x\}=\{y\}$. Of course, this calculation depends on being able to assume that the bound variables and free variables of $x\lambda y.y$ do not overlap, or, in the terminology of \cite{kurzetal:cmcs2012}, that $x\lambda y.y$ is $\alpha$-safe. The next definition gives a semantic formulation of an element being safe with respect to a map, which now does not need to be a quotient by $\alpha$-equivalence.

\begin{defi}[Safe element]
  \label{def:f-safe}
Let $f:\setone\to\settwo$ be an equivariant function. 
We call $\elone\in\setone$ $f$-safe when
\begin{equation}
  \label{eq:f-safe}
  |\supp(\elone)|=\max\{\ |\supp(\eltwo)|\ |\ \eltwo\in f^{-1}(f(\elone)) \}.
\end{equation}
\end{defi}

The maximum in the right-hand side of~\ref{eq:f-safe} does not always exist, see Example~\ref{ex:safe-non-safe-map}.

\newcommand{\alphaequivmap}{[-]_{\alpha}}

\begin{exa}[$\alphaequivmap$-Safe Terms]
We consider the equivariant map $[-]_{\alpha}: \Lambda \rightarrow \Lambda/\!\alphaconv$.
Then  $M$ is an $\alphaequivmap$-safe term if it has a maximal number of variables among all the representatives of its $\alpha$-equivalence class.
The terms $x (\lambda y. y)$ and $\lambda x. x (\lambda y. y)$ are
$[-]_{\alpha}$-safe but
$ x (\lambda x.x)$ 
and $\lambda x. x (\lambda x.x)$ are not.
\end{exa}

\begin{rem}[$\alpha$-safe Term]
A term $M\in \Lambda $ is $\alpha$-safe in the sense of
 \cite[Definition~19]{kurzetal:cmcs2012} if and only if $M$ is
$[-]_{\alpha}$-safe in the sense of Definition \ref{def:f-safe}.
Intuitively, a $\lambda$-term $M$ is
$\alpha$-safe when $\bv(M) \cap \fv(M) = \emptyset$ and $M$ does not
have two different $\lambda$'s with the same binding variable, i.e.\ 
if $\lambda x$ and $\lambda y$ occur in two different positions of $M$
then $x \not = y$.  
\end{rem}

If a $\lambda$-term $M$ is $[-]_{\alpha}$-safe then the 
set $\bv(M)$ of bound variables
of $M$ is equal to $\var(M)\setminus\fv(M) = \supp(M)\setminus\supp([M]_{\alpha})$.
This motivates the following notation.

\begin{notation}
  If $f:\setone\to\settwo$ is an equivariant map 
   then we define
   \[\bv_{f}(\elone)= \supp(\elone)\setminus\supp(f(\elone)).\]
When no confusion may arise, we omit the subscript and write
$\bv(\elone)$ instead of $\bv_{f}(\elone)$.
\end{notation}

\begin{rem}
Let $f:\setone \to \settwo$ be equivariant and $\elone \in \setone$.
Then 
$\supp(f(\elone)) \subseteq \supp(\elone)$ and 
$|\bv_{f}(\elone)|= |\supp(\elone) | - |\supp(f(\elone))|$.
\end{rem}

\begin{lem} Let $f: \setone \to \settwo$ be an equivariant map.
Then $\elone$ is $f$-safe if and only if 
\begin{equation}
  |\bv_{f}(\elone)|=
  \max\{\ |\bv_{f}(\eltwo)|\ |\ \eltwo\in f^{-1}(f(\elone)) \}.
\end{equation}
\end{lem}

\begin{proof} 
Assume $\elone $ is $f$-safe and 
consider $\eltwo \in \settwo$ such that $f(\eltwo) =
f(\elone)$.
Then
\[|\bv_{f}(\elone)|= |\supp(\elone) | - |\supp(f(\elone))| \geq 
|\supp(\eltwo)| - |\supp(f(\eltwo))| = |\bv_{f} (\eltwo)|.\]
The converse is similar.
\end{proof}

\begin{defi}[Safe map]
  \label{def:safe-maps}
  Let $f:\setone\to\settwo$ be an equivariant map in $\Nom$. We call $f$ \emph{safe} when for all $\eltwo\in\settwo$ there exists an $f$-safe  $\elone\in\setone$ such that $f(\elone)=\eltwo$.
\end{defi}

\begin{exa}\label{ex:safe-non-safe-map}
The map $[-]_{\alpha}: \Lambda \rightarrow \Lambda/\!\alphaconv$ is safe
\cite[Lemma~20]{kurzetal:cmcs2012}.
But the map $!: 
\mathcal{P}_{\mathit{fin}}(\setvars)\rightarrow \{* \}$ is not.
\end{exa}

A diagram such as \eqref{eq:squ} is a weak pullback if for all identified $u$ and $v$, there is a $z$ witnessing this fact, that is, if $f(u)=q(v)$ then there exists $z\in Z$ such that $u=p(z)$ and $v=g(z)$. In the following we will need a similar  but weaker condition, which, intuitively, requires the existence of a witness only up to the renaming of bound variables. 

\begin{defi}[Safe square]
  \label{def:strange-pullback}
A square
\begin{equation}
  \label{eq:squ}
 \begin{gathered}
  \xymatrix{
\setone\ar@{<-}[r]^p\ar[d]_f & \setfour\ar[d]^g \\
\setthree\ar@{<-}[r]_q & \settwo
}
\end{gathered}
\end{equation}
is a \emph{safe square} when for all $f$-safe $\elone\in\setone$ and for all $\eltwo\in\settwo$ such that $f(\elone)=q(\eltwo)$ and 
$\bv_{f}(\elone) \# \eltwo$
there exists a $g$-safe $\elfour\in\setfour$  such that $p(\elfour)=\elone$ and $g(\elfour)=\eltwo$.
\end{defi}

\newcommand{\pbP}{P}
\newcommand{\nbb}{\mathbb{N}}
\newcommand{\limX}{\lim X_n}
\newcommand{\limY}{\lim Y_n}
\subsubsection{Representing limits in nominal sets}\ \\

\noindent Consider two $\omega^\mathit{op}$-chains in $\Nom$ and let $\limX$, respectively $\limY$ denote their limits in $\Nom$.
\begin{equation}
  \label{eq:fin-chain-gen}
\begin{gathered}
  \xymatrix
{
  \setone_0\ar@{<-}[r]^-{p_1}\ar@{->>}[d]^{f_0} & \setone_1 \ar@{<-}[r]^-{p_2}\ar@{->>}[d]^{f_1} & \setone_2 \ar@{<-}[r]\ar@{->>}[d]^{f_2} & \cdots  \ar@{<-}[r] & \limX\ar@{-->}[d]^f \\
\settwo_0\ar@{<-}[r]_-{q_1} & \settwo_1 \ar@{<-}[r]_-{q_2} & \settwo_1 \ar@{<-}[r] & \cdots  \ar@{<-}[r] & \limY
}
\end{gathered}
\end{equation}
By the universal property of the limits we obtain a map $f:\limX\to\limY$. In the category of sets, we have the theorem that if all the squares are weak pullbacks and all the $f_n$ are surjective, then $f$ is surjective. But in our main example where the $f_n$ quotient by $\alpha$-equivalence, the squares are not weak pullbacks.

Recalling that $\fgtu:\Nom\to\Set$ denotes the forgetful functor,
 consider the limits of the two chains in $\Set$

\begin{equation}
  \label{eq:fin-chain-gen-2}
 \begin{gathered}
  \xymatrix
{
  \fgtu\setone_0\ar@{<-}[r]\ar@{->>}[d]^{f_0} & \fgtu\setone_1 \ar@{<-}[r]\ar@{->>}[d]^{f_1} & \fgtu\setone_2 \ar@{<-}[r]\ar@{->>}[d]^{f_2} & \cdots  \ar@{<-}[r] & \lim\fgtu\setone_n\ar@{-->}[d]^g\ar@{<--}[r]^-a&\fgtu\limX\ar[d]^{\fgtu f}\ar@{}[dl]|{(*)} \\
\fgtu\settwo_0\ar@{<-}[r] & \fgtu\settwo_1 \ar@{<-}[r] & \fgtu\settwo_1 \ar@{<-}[r] & \cdots  \ar@{<-}[r] & \lim\fgtu\settwo_n\ar@{<--}[r]_-b&\fgtu\limY
}
\end{gathered}
\end{equation}
By the universal property of limits, there exist  unique  maps $a:\fgtu\limX\to\lim\fgtu\setone_n$ and $b:\fgtu\limY\to\lim\fgtu\settwo_n$ making the square $(*)$ commutative.
Again, the map $g:\lim\fgtu\setone_n\to\lim\fgtu\settwo_n$ induced in the limit  may not be surjective in general, see Example~\ref{ex:alpha-inf-lam}.
 
However we can prove the following general result:

\begin{thm}
\label{thm:gen-seq-pb-surj}
Assume that diagram~\eqref{eq:fin-chain-gen} is such that for all $n$ the square
\begin{equation}
  \label{eq:n-swp}
 \begin{gathered}
  \xymatrix{
\setone_n\ar@{<-}[r]^{p_n}\ar@{->>}[d]^{f_n} & \setone_{n+1}\ar@{->>}[d]^{f_{n+1}}\\
\settwo_n\ar@{<-}[r]_{q_n} & \settwo_{n+1}\\
}
\end{gathered}
\end{equation}
is safe and  $f_n$ is a safe map. Let $\pbP$ denote the pullback
\begin{equation}
  \label{eq:fin-pb}
 \begin{gathered}
  \xymatrix{
 \lim\fgtu\setone_n\ar[d]_-{g}\ar@{<-}[r]^-r &P\ar[d]^-{h}\SWpullbackcorner & \\
\lim\fgtu\settwo_n\ar@{<-}[r]_-b&\fgtu\limY \\
}
\end{gathered}
\end{equation}
Then $h:\pbP\to\fgtu\limY$ is a surjection.
\end{thm}

\proof
We start with $\eltwo\in\limY$. Consider the projections $\eltwo_n\in\settwo_n$ of $\eltwo$ obtained via the projections of the limiting cone.  Since $q_n(v_{n+1})=v_n$ and $q_n$ are equivariant maps we have the following inclusions of finite sets
$$\supp(\eltwo_0)\subseteq\supp(\eltwo_1)\subseteq\ldots\subseteq\supp(\eltwo).$$
Thus the sequence stabilises eventually, that is, there exists $n$ such that $\supp(\eltwo_n)=\supp(\eltwo_{n+k})$ for all natural numbers $k$. Moreover since $\eltwo=(\eltwo_0,\eltwo_1,\ldots)$ we have that for all $k\ge 0$  $$\supp(\eltwo_{n+k})=\supp(\eltwo).$$
We will now construct $\elone=(\elone_1,\elone_2,\ldots)$ in $\lim\fgtu\setone_n$ such that $g(\elone)=b(\eltwo)$. First let $\elone_n\in\setone_n$ be an $f_n$-safe element such that $f_n(\elone_n)=\eltwo_n$. For each $k\ge 1$, we define an $f_{n+k}$-safe $\elone_{n+k}\in\setone_{n+k}$ such that  $p_{n+k}(\elone_{n+k})=\elone_{n+k-1}$ 
and $f_{n+k}(\elone_{n+k})=\eltwo_{n+k}$. 

The proof is by induction on $k$. Since $\supp(\eltwo_{n+1})=\supp(\eltwo_n)$, we have that $(\supp(\elone_n)\setminus\supp(\eltwo_n))\cap\supp(\eltwo_{n+1})=\emptyset$. Since all the squares~\eqref{eq:n-swp} are safe squares, there exists an $f_{n+1}$-safe element $\elone_{n+1}\in\setone_{n+1}$ in the preimage of $\eltwo_{n+1}$ such that $p_{n+1}(\elone_{n+1})=\elone_n$.  This shows that the claim is true for $k=1$. 

For the inductive step $k\to k+1$, notice that $(\supp(\elone_{n+k})\setminus\supp(\eltwo_{n+k}))\cap\supp(\eltwo_{n+k+1})=\emptyset$ and that $\elone_{n+k}$ is $f_{n+k}$-safe. Using the safe square property, there exists $f_{n+k+1}$-safe $\elone_{n+k+1}\in\setone_{n+k+1}$ with the desired properties.

For $m<n$ define $\elone_m=p_mp_{m+1}\ldots p_n(\elone_n)$. Now observe that 
$(\elone_0,\elone_1,\ldots)$ is an element of $\lim\fgtu\setone_n$ whose image in $\lim\fgtu\settwo_n$ is $(\eltwo_0,\eltwo_1,\ldots)$. This means that 
$$g((\elone_0,\elone_1,\ldots))=b(\eltwo).$$
Since $P$ is the pullback of $g$ and $b$, there exists $\elthree\in P$ such that $r(\elthree)=(\elone_0,\elone_1,\ldots)$ and $h(\elthree)=\eltwo$. Thus $h:P\to\fgtu\limY$ is surjective.
\qed

\medskip Going back to Theorem~\ref{thm:comp-alpha-comm} and looking
at \eqref{eq:PTonto}, we find that so far we established the following
corollary of Theorem~\ref{thm:comp-alpha-comm}.

\begin{cor}\label{cor:isomfinalcoalgcarriers-2}
  Let $\fnom{F}, \fnom{F}_\alpha$ be endofunctors on $\Nom$ having
  final coalgebras, and let $\fnom{F}\to\fnom{F}_\alpha$ be a
  component-wise surjective natural transformation, and $P$ be a
  pullback as in \eqref{eq:classesbindsig-2}. Moreover assume that
  $\fnom{F}$ preserves surjections. If
\begin{itemize}
\item the induced maps  $\amap{n}:\fnom{F}^n1\to\fnom{F}_\alpha^n1$ are safe
\item and the squares 
\begin{equation}
\label{eq:ssq-alpha}
 \begin{gathered}
  \xymatrix{
\fnom{F}^n1\ar@{<-}[r]\ar@{->>}[d]_{\amap{n}} & \fnom{F}^{n+1}1\ar@{->>}[d]^{\amap{n+1}}\\
\fnom{F}_\alpha^n1\ar@{<-}[r] & \fnom{F}_\alpha^{n+1}1\\
}
\end{gathered}
\end{equation}
are safe squares,
\end{itemize}
then $P\to U\fcoalg_\alpha$ is onto.
\end{cor}
\medskip In the following we will prove the second part of Theorem~\ref{thm:comp-alpha-comm} using Corollary~\ref{cor:isomfinalcoalgcarriers-2}, that is, by establishing the two
bullet points above.
  Let us briefly outline the structure of that
argument. 
Recall that the maps $\amap{n}$ are defined inductively as follows:
$\amap{0}=\mathit{id}_1$ and
$\amap{n+1}=q_{\fnom{F}_\alpha^n1}\circ\fnom{F}(\amap{n})$.
Also recall that the natural transformation
$q:\fnom{F}\to\fnom{F}_\alpha$ is defined inductively depending on the
structure of $\fnom{F}$ and $\fnom{F}_\alpha$ using the
rules~\eqref{eq:8}.  Therefore, the argument will proceed by
induction on $n$ and on the structure of $\fnom{F}$. For this, 
we need some structural closure properties for safe maps. To this end, we also study a special case of safe maps, namely maps with orbit-finite fibres. This will help us to prove
 the first bullet point of  Corollary~\ref{cor:isomfinalcoalgcarriers-2}. For the second bullet we need a detailed study of safe squares. This is eventually done in Section~\ref{sec:prop-safe-sq}.

\subsubsection{Some properties of safe elements and safe maps}\

\begin{lem}
  \label{lem:pi-safe-is-safe}
If $\elone\in\setone$ is $f$-safe and $\pi$ is an arbitrary finite permutation then $\pi\cdot \elone$ is $f$-safe. 
\end{lem}
\proof

Suppose $\elone\in\setone$ is $f$-safe. Consider $\eltwo \in \setone$
such that 
$\eltwo\in f^{-1}(f(\pi\cdot\elone))$.
Then $\pi^{-1} \cdot \eltwo \in f^{-1}( f(\elone))$ because $f$ is
equivariant. Using Remark~\ref{rem:supp-equiv} we conclude 
\[
|\supp(\pi \cdot \elone) | = |\supp(\elone) | \geq |\supp (\pi^{-1} \cdot\eltwo)|
= |\supp(\eltwo)|.
\] 
\qed

\begin{lem}
\label{lem:choosing-safe-elem}
  Assume $f:\setone\to\settwo$ is safe. Consider $\eltwo\in\settwo$ and let $S$ be an arbitrary finite set of names. Then there exists an $f$-safe $\elone\in f^{-1}(\eltwo)$ such that $\bv(\elone)\cap S=\emptyset$.
\end{lem}
\proof
Since $f$ is safe, there exists an $f$-safe $\elone\in f^{-1}(\eltwo)$. Let $T=S\cap\bv(\elone)$. If $T$ is empty we are done. If not, let $T'$ be a finite set of names fresh for $\elone,\eltwo,S$ which has the same number of elements as $T$. Let  $\pi$ denote a finite permutation that swaps the elements of $T$ and $T'$ and fixes the remaining names in $\setvars$. 
 Then $\pi \cdot x = x$ for all
$x \in \supp(v)$,
thus $\pi\cdot \elone\in f^{-1}(\eltwo)$. 
By Lemma \ref{lem:pi-safe-is-safe}, $\pi \cdot \elone$ is $f$-safe.
Moreover, by Remark~\ref{rem:supp-equiv} we have $\supp(\pi\cdot \elone)=\pi\cdot\supp(\elone)$ and thus $\bv(\pi\cdot \elone)\cap S=\emptyset$.
\qed

\begin{lem}
\label{lem:safe-elem-and-prod}
Assume $f_1:\setone_1\to\settwo_1$  and $f_2:\setone_2\to\settwo_2$ are safe maps. Then $(\elone_1,\elone_2)$ is $(f_1\times f_2)$-safe if and only if the following hold
\begin{equation}\label{eq:3}
\begin{array}{l}
\elone_i\textrm{ is  }f_i\textrm{-safe for }i=1,2\\
\bv(\elone_1)\# \elone_2\\
\bv(\elone_2)\# \elone_1.
\end{array}  
\end{equation}
\end{lem}

\proof
Assume $(\elone_1,\elone_2)$ satisfy~\eqref{eq:3} and let $\eltwo_i$ denote $f_i(\elone_i)$  for $i=1,2$.
We will show that $(\elone_1,\elone_2)\in (f_1\times f_2)^{-1}(\eltwo_1,\eltwo_2)$ is $(f_1\times f_2)$-safe. 
First observe that 
\begin{equation}
  \label{eq:comp-bv-prod}
  \begin{array}{lcl}
\bv(\elone_1,\elone_2) &=& 
(\supp(\elone_1)\cup\supp(\elone_2))\setminus(\supp(\eltwo_1)\cup\supp(\eltwo_2))\\
& =& (\supp(\elone_1)\setminus\supp(\eltwo_1))\setminus\supp(\eltwo_2) \cup \\
&\cup &(\supp(\elone_2)\setminus\supp(\eltwo_2))\setminus\supp(\eltwo_1)\\
& =& \bv(\elone_1)\uplus\bv(\elone_2).
  \end{array}
\end{equation}
The last equality holds by~(\ref{eq:3}). Therefore
\begin{equation}
  \label{eq:bv-comp-3}
|\bv(\elone_1,\elone_2)|=|\bv(\elone_1)|+|\bv(\elone_2)|.  
\end{equation}
In order to show that $(\elone_1,\elone_2)$ is indeed $(f_1\times f_2)$-safe consider $(\elone_1',\elone_2')\in (f_1\times f_2)^{-1}(\eltwo_1,\eltwo_2)$. We have the inequalities
\begin{equation}
  \label{eq:comp-bv-prod-2}
  \begin{array}{lcl}
|\supp(\elone_1',\elone_2')\setminus\supp(\eltwo_1,\eltwo_2)| & =& 
|(\supp(\elone_1')\cup\supp(\elone_2'))\setminus(\supp(\eltwo_1)\cup\supp(\eltwo_2))|\\
&\le& |(\supp(\elone_1')\setminus(\supp(\eltwo_1)\cup\supp(\eltwo_2)))| + \\
&+ &|(\supp(\elone_2')\setminus(\supp(\eltwo_1)\cup\supp(\eltwo_2)))|\\
&\le& |(\supp(\elone_1')\setminus\supp(\eltwo_1)|+|(\supp(\elone_2')\setminus\supp(\eltwo_2)|\\
& \le & |\bv(\elone_1)|+|\bv(\elone_2)|.
  \end{array}
\end{equation}
The last inequality holds because $\elone_1$ and $\elone_2$ are $f_1$-safe, respectively, $f_2$-safe. Thus, using~\eqref{eq:bv-comp-3}, we can conclude that 
$$|\supp(\elone_1',\elone_2')\setminus\supp(\eltwo_1,\eltwo_2)|\le|\bv(\elone_1,\elone_2)|.$$

Conversely, assume that $(\elone_1',\elone_2')\in (f_1\times f_2)^{-1}(\eltwo_1,\eltwo_2)$ is $(f_1\times f_2)$-safe. By Lemma~\ref{lem:choosing-safe-elem}, there exists $\elone_1\in f_1^{-1}(\eltwo_1)$ an $f_1$-safe element such that
  \begin{equation}
    \label{eq:bv-cond-1}   
     \bv(\elone_1)\cap\supp(\elone_2)=\emptyset.
   \end{equation}
Similarly, there exists $\elone_2\in f_2^{-1}(\eltwo_2)$ an $f_2$-safe element such that 
\begin{equation}
  \label{eq:bv-cond-2}
   \bv(\elone_2)\cap\supp(\elone_1)=\emptyset. 
\end{equation}
From~\eqref{eq:bv-cond-1} and~\eqref{eq:bv-cond-2} we can derive~\eqref{eq:3}.
 By the first part of the proof we know that 
 $(\elone_1,\elone_2)$ is $(f_1\times f_2)$-safe and that the inequalities of~(\ref{eq:comp-bv-prod-2}) hold. 
 Since $(\elone_1',\elone_2')$ is also $(f_1\times f_2)$-safe 
 we know that all the inequalities of~(\ref{eq:comp-bv-prod-2}) are equalities. 
 For the last inequality of~(\ref{eq:comp-bv-prod-2}), this implies that each $\elone_i'$ is $f_i$-safe. For the second inequality of~(\ref{eq:comp-bv-prod-2}), this implies that $\bv(\elone_1')\# f_2(\elone_2)$ and $\bv(\elone_2)\# f_1(\elone_1)$. Hence, by the fact that also the first inequality of~\eqref{eq:comp-bv-prod-2} is actually an equality we have that $\bv(\elone_1')\#\bv(\elone_2')$. Since $\supp(\elone_1')=\bv(\elone_1')\cup\supp(  f_1(\elone_1'))$ we conclude that $\bv(\elone_2')\#\elone_1'$. Similarly $\bv(\elone_1')\#\elone_2'$.
\qed

\begin{lem}
  \label{lem:safe-composition-1}
Let $f:\setone\to\settwo$ and $g:\settwo\to\setthree$ be equivariant maps. If $\elone$ is $gf$-safe, then $\elone$ is $f$-safe.
\end{lem}
\proof

Consider $\eltwo\in f^{-1}(f(\elone))$. Then 
$gf(\eltwo) = gf(\elone)$.  Since $\elone$ is $gf$-safe,
we have that $|\supp(\eltwo)| \leq  |\supp(\elone)|$.
\qed

Unfortunately, safe maps are not closed under composition as the next example shows.

\begin{exa}
  Consider the set $\nbb$ of natural numbers as a nominal set equipped with the trivial action and let $\mathcal{P}_{\mathit{fin}}(\setvars)$ denote the nominal set of all finite subsets of $\setvars$. Let $f:\mathcal{P}_{\mathit{fin}}(\setvars)\to\nbb$ denote the map which sends any finite set of names to its 
  cardinal and let $g:\nbb\to\{*\}$ denote the unique map from $\nbb$ into the final nominal set $\{*\}$. Both $f$ and $g$ are equivariant and safe, but their composition $g\circ f$ is not safe.
\end{exa}

 Therefore we need a stronger notion of maps that still accommodates our examples but with better closure properties. This is the purpose of the next section.

\subsubsection{Maps with orbit-finite fibers}\ \\\vspace{-6 pt}

\noindent In this section, we introduce the notion of maps with orbit-finite fibers
and use it to prove that 
the maps $\amap{n}:\fnom{F}^n1\to\fnom{F}_\alpha^n1$  are safe.
\noindent

\begin{defi}[Orbite-finite fibers] 
  We say that a $\Nom$-morphism $f:\setone\to\settwo$ has orbit-finite fibers when for all $\eltwo\in\settwo$ we have that $f^{-1}(\eltwo)$ is included in the union of finitely many orbits of $\setone$.
\end{defi}

\begin{lem}
\label{lem:char-off-maps}
  The following are equivalent:
  \begin{enumerate}
    \item $f$ has orbit-finite fibers.
    \item For all $\eltwo\in\settwo$ there exists a finitely presentable nominal subset $\setone_\eltwo\subseteq\setone$ such that $f^{-1}(\eltwo)\subseteq\setone_\eltwo$.
    \item For all finitely presentable nominal subset $\settwo'\subseteq\settwo$ the nominal subset $f^{-1}(\settwo')$ of $\setone$ is  finitely presentable.
  \end{enumerate}
\end{lem}
\proof
$(1) \Longleftrightarrow (2)$ is immediate.

$(2) \Longrightarrow (3)$: Consider $\settwo'\subseteq\settwo$ a finitely presentable nominal subset. Then $\settwo'$ is a finite union of orbits $O_{\eltwo_1}\cup\ldots\cup O_{\eltwo_n}$. We have picked generators $\eltwo_1,\ldots,\eltwo_n$ for these orbits. For each $i$ there exists a finitely presentable nominal subset $\setone_i\subseteq\setone$ with $f^{-1}(\eltwo_i)\subseteq\setone_i$. Then for all $\pi$ we have that $f^{-1}(\pi\cdot\eltwo_i)\subseteq \pi\cdot\setone_i=\setone_i$. Therefore $f^{-1}(O_{\eltwo_i})\subseteq\setone_i$. Therefore, $f^{-1}(\settwo')\subseteq \setone_1\cup\ldots\cup\setone_n$. Since a finite union of finitely presentable nominal sets is a finitely presentable nominal set, we conclude that $f^{-1}(\settwo')$ is a nominal subset of a finitely presentable nominal set, thus it is  itself a finitely presentable nominal set.
$(3) \Longleftarrow (2)$: Consider $\eltwo\in\settwo$. Put $\settwo'$ to be the one-orbit nominal set generated by $\eltwo$. By $(3)$ we have that $f^{-1}(\settwo')$ is a finitely presentable nominal subset of $\setone$. We obviously have that $f^{-1}(\eltwo)\subseteq f^{-1}(\settwo')$.
\qed

\begin{lem}
\label{lem:has-off-imp-safe}
  If $f:\setone\to\settwo$ is surjective and has orbit-finite fibers, then $f$ is safe.
\end{lem}
\proof If two elements of $\setone$ are in the same orbit of $f^{-1}(\eltwo)$ their supports have the same number of variables. 
Indeed, if $\elone,\elone'\in f^{-1}(\eltwo)$ and 
$\elone'=\pi\cdot\elone$ for some permutation $\pi$, 
we have that 
$\supp(\elone')=\pi\cdot\supp(\elone)$, 
thus $\supp(\elone)$ and $\supp(\elone')$ 
have the same number of elements. 
By hypothesis $f^{-1}(\eltwo)$ is included in the union 
$O_1\cup\ldots\cup O_n$  of finitely many orbits of $\setone$. 
By the above observation, the set 
$$\{|\supp(\elone)| \ | \ \elone\in O_i\cap f^{-1}(\eltwo)\}$$
is a singleton  for each $1\le i\le n$. 
Since we have only finitely many orbits, it follows that 
$$\{|\supp(\elone)| \ | \ \elone\in f^{-1}(\eltwo)\}$$
has at most $n$ elements, and therefore has  a maximum. Thus $f$ is safe.
\qed

\begin{lem}
\label{lem:off-maps-closure-prop}
  Functions with orbit-finite fibers are closed under
(1) finite products, (2) coproducts, (3) abstraction, and (4) composition.
\end{lem}
\proof\hfill
\begin{enumerate}
\item[(1)] Consider $f_i:\setone_i\to\settwo_i$ for $i=1,2$ with orbit-finite fibers. Consider $(\eltwo_1,\eltwo_2)\in\settwo_1\times\settwo_2$. There exists finitely presentable nominal subsets $\setone_i'\subseteq\setone_i$ such that $f_i^{-1}(\eltwo_i)\subseteq\setone_i'$. Then $(f_1\times f_2)^{-1}(\eltwo_1,\eltwo_2)=f_1^{-1}(\eltwo_1)\times f_2^{-1}(\eltwo_2)\subseteq \setone_1'\times\setone_2'$. But, since orbit-finite nominal sets are closed under finite products, see for example~\cite[Lemma~2]{BojanczykBKL12}, we have that $\setone_1'\times\setone_2'$ is orbit-finite, so we are done.
\item[(2)] That's easy, component-wise.
\item[(3)] Consider $f:\setone\to\settwo$ with orbit-finite fibers. We want to prove that $[\setvars]f$ also has orbit-finite fibers. Notice that 
$$([\setvars]f)^{-1}(\abs{x}\eltwo)=\{\abs{x}\elone\ |\ \elone\in f^{-1}(\eltwo)\}.$$
There exist a finitely presentable nominal subset $\setone'\subseteq\setone$ such that $f^{-1}(\eltwo)\subseteq \setone'$. Therefore
 $([\setvars]f)^{-1}(\abs{x}\eltwo)\subseteq[\setvars]\setone'$. 
 By 
 Lemma~\ref{lem:abs-pres-fp} we know that $[\setvars]\setone'$ is finitely presentable.

\item[(4)] Assume $f:\setone\to\settwo$ and $g:\settwo\to\setthree$ have orbit-finite fibers. We show that $g\circ f$ also has orbit-finite fibers. Let $\elthree\in\setthree$. Then $g^{-1}(\elthree)\subseteq\settwo'$ for some finitely presentable nominal subset $\settwo'$ of $\settwo$. By Lemma~\ref{lem:char-off-maps} we have that $f^{-1}(\settwo')$ is finitely presentable. Since $(g\circ f)^{-1}(\elthree)= f^{-1}(g^{-1}(\elthree))\subseteq f^{-1}(\settwo')$ we are done.
\qed
\end{enumerate}

\begin{exa}
  The function $\theta_\setone:\setvars\times\setone\to[\setvars]\setone$ has orbit-finite fibers. Indeed, notice that $\theta^{-1}(\abs{x}\eltwo)\subseteq\setvars\times O_\eltwo$ where $O_\eltwo$ is the orbit spanned by $\eltwo$. But $\setvars\times O_\eltwo$ is finitely presentable, so we are done.
\end{exa}

\begin{lem}
\label{lem:q-has-off}
  The maps $q_\setone:\fnom{F}\setone\to\fnom{F}_\alpha\setone$ have orbit-finite fibers.
\end{lem}
\proof This is proved by induction on the structure of $\fnom{F}$, since the map $q_\setone$ is obtained from the identity map and $\theta_\setone$, via products, coproducts and composition. We can therefore apply Lemma~\ref{lem:off-maps-closure-prop}. \qed

\begin{prop}\label{prop:safe-maps}
The maps $\amap{n}:\fnom{F}^n1\to\fnom{F}_\alpha^n1$ have orbit-finite fibers and are safe.
\end{prop}

\proof We use induction on $n$. The base case is clear since $\amap{0}=\mathit{id}_1$. For the inductive step, notice that $\amap{n+1}=q_{\fnom{F}_\alpha^n1}\circ\fnom{F}(\amap{n})$.
 By items  (1) and (2) 
  of Lemma~\ref{lem:off-maps-closure-prop} 
  we have that $\fnom{F}(\amap{n})$ has orbit-finite fibers. 
  Then we can apply item (4) 
  of Lemma~\ref{lem:off-maps-closure-prop} and 
  Lemma~\ref{lem:q-has-off} to derive that $\amap{n+1}$ 
  also has orbit-finite fibers.
Since the $\amap{n}$ are surjective, it follows from Lemma~\ref{lem:has-off-imp-safe} that they are safe.
\qed

The maps with orbit-finite fibers have nice closure properties, but safe elements do not behave well with respect to composition, as shown in the next example.

\begin{exa}
  Let $\mathcal{P}_2(\setvars)$ be the nominal set of two-element sets of names and let $f:\setvars+\mathcal{P}_2(\setvars)\to\setvars+1$ denote the map $\mathit{id}_\setvars + !$ where $!$ denotes the unique map into the final nominal set. Let $g$ denote the unique map from $\setvars +1$ to $1$. Notice that $f$, $g$ (and therefore their composition) have orbit-finite fibers and are safe. Nevertheless
  \begin{itemize}
  \item $\elone\in\setvars$ is $f$-safe and $f(\elone)=\elone$ is $g$ safe, but $\elone$ is not $(g\circ f)$-safe.
\item $\{\elone,\eltwo\}$ is $(g\circ f)$-safe, but $f(\{\elone,\eltwo\})$ is not $g$-safe.
  \end{itemize}
\end{exa}
\newcommand{\fix}{\mathit{fix}}
Therefore, in the next section  we need to study the properties of safe squares.

\subsubsection{Properties of safe squares}\label{sec:prop-safe-sq} \ \\

\noindent
In the first part of this section we will show that safe squares are closed under products and coproducts. Then we will show that also the `vertical' composition of safe squares is safe provided that some additional properties are satisfied by the maps at issue. 
 This allows us to prove the second  bullet point of
 Corollary~\ref{cor:isomfinalcoalgcarriers-2}
 and  conclude our main result on the surjectivity of
 $P\to U\fcoalg_\alpha$ in Theorem \ref{thm:bindingsig-surj}.
 
\begin{lem}
\label{lem:choosing-good-safe-el}
Consider a safe square
\begin{equation}
\begin{gathered}
  \xymatrix{
\setone\ar@{<-}[r]^p\ar[d]_f & \setfour\ar[d]^g \\
\setthree\ar@{<-}[r]_q & \settwo
}
\end{gathered}
\end{equation}
If $\elone\in\setone$ is $f$-safe, $\eltwo\in\settwo$ and $S$ is a finite subset of $\setvars$ such that $f(\elone)=q(\eltwo)$,  $\bv_{f} (\elone)\#\eltwo$ and $\bv_{f} (\elone)\# S$
there exists a $g$-safe $\elfour\in\setfour$  such that $p(\elfour)=\elone$, $g(\elfour)=\eltwo$ and $\bv_{g}(\elfour)\# S$. 
\end{lem}
\proof By the definition of safe squares there exists a $g$-safe $\elfour$ such that $g(\elfour)=\eltwo$ and $p(\elfour)=\elone$. Let $T$ denote the intersection $S\cap\bv_{g}(\elfour)$. If $T$ is empty we are done. Otherwise consider a set $T'$ of names fresh for $\elone,\eltwo,\elfour,S$ having the same cardinality as $T$ and let $\pi$ denote a finite permutation that swaps the elements of $T$ with the elements of $T'$ and fixes all the other elements of $\setvars$. We will show that $\pi\cdot\elfour$ has all the required properties.  
By Lemma \ref{lem:pi-safe-is-safe}, 
 $\pi\cdot\elfour$ is $g$-safe.
 In order to prove that $g(\pi \cdot \elfour) = \eltwo$ 
 it is enough to check that $\pi \cdot x=x$ for all $x \in 
 \supp(\eltwo)$. 
 This is true because $T\subseteq\bv_g(\elfour)$ and $\bv_{g}(\elfour)\cap\supp(\eltwo)=\emptyset$, hence $T\cap\supp(\eltwo)=\emptyset$. 
 In order to check that $p(\pi\cdot\elfour)=\elone$,
  it is enough to show that $\pi \cdot x = x$ for all $x \in\supp(\elone)$. 
  But 
  $\supp(\elone)=\bv_{f}(\elone)\cup\supp(f(\elone))
  \subseteq\bv_{f}(\elone)\cup\supp(\eltwo)$. 
  We have established that $T\cap\supp(\eltwo)=\emptyset$. On the other hand, $T\subseteq S$ and $S\#\bv_{f}(\elone)$ imply that $T\cap\bv_{f}(\elone)=\emptyset$. Hence $T\# \elone$. Since $T'\#\elone$ we obtain that $\pi$ fixes all the names in $\supp(\elone)$. 
\qed

\begin{lem}
\label{lem:ssq-close-prod-coprod}
  Safe squares are closed under finite products and coproducts.
\end{lem}

\proof The case of  coproducts is easy. We show the closure of safe squares under finite products.
Assume
\begin{equation}
  \label{eq:squ-2}
\begin{gathered}
  \xymatrix{
\setone_i\ar[d]_-{f_i}\ar@{<-}[r]^-{p_i} & \setfour_i\ar[d]^-{q_i} \\
\setthree_i\ar@{<-}[r]_{g_i} & \settwo_i
}
\end{gathered}
\end{equation}
for $i\in\{1,2\}$ are safe squares. We will show that 
\begin{equation}
\begin{gathered}
  \xymatrix{
\setone_1\times\setone_2\ar[d]_-{f_1\times f_2}\ar@{<-}[rr]^-{p_1\times p_2} & & \setfour_1\times\setfour_2\ar[d]^-{q_1\times q_2} \\
 \setthree_1\times\setthree_2\ar@{<-}[rr]_-{g_1\times g_2}&  &\settwo_1\times\settwo_2
}
\end{gathered}
\end{equation}
is a safe square. 
Consider $f_1\times f_2$-safe $(\elone_1,\elone_2)$ and $(\eltwo_1,\eltwo_2)\in\settwo_1\times\settwo_2$ such that $\bv(\elone_1,\elone_2)\#(\eltwo_1,\eltwo_2)$ and $f_i(\elone_i)=g_i(\eltwo_i)$. By Lemma~\ref{lem:safe-elem-and-prod} we know that 
$\elone_i$ is $f_i$-safe and that $\bv(\elone_1)\#\bv(\elone_2)$. Moreover we can compute that
$$\bv(\elone_1,\elone_2)=\bv(\elone_1)\uplus\bv(\elone_2).$$
Since $\bv(\elone_1)\#\eltwo_2$ and  $\bv(\elone_1)\#\bv(\elone_2)$, by Lemma~\ref{lem:choosing-good-safe-el} (applied for $S=\bv(\elone_2)$) we can find  $\elfour_1\in\setfour_1$
 such that 
\begin{itemize}
  \item $\bv(\elfour_1)\# \eltwo_2$ and $\bv(\elfour_1)\#\bv(\elone_2)$
  \item $\elfour_1$ is $q_1$-safe.
  \item $q_1(\elfour_1)=\eltwo_1$ and $p_1(\elfour_1)=\elone_1$.
\end{itemize}
Since $\bv(\elone_2)\#\eltwo_1$ and $\bv(\elone_2)\#\bv(\elfour_1)$, by Lemma~\ref{lem:choosing-good-safe-el} (applied for $S=\bv(\elfour_1)$) we can find $\elfour_2\in\setfour_2$ such that
\begin{itemize}
  \item $\bv(\elfour_2)\# \eltwo_1$  and $\bv(\elfour_2)\#\bv(\elfour_1)$.
  \item $\elfour_2$ is $q_2$-safe.
  \item $q_2(\elfour_2)=\eltwo_2$ and $p_2(\elfour_2)=\elone_2$.
\end{itemize}
By construction we have that $\bv(\elfour_1)\# \eltwo_2$, $\bv(\elfour_2)\# \eltwo_1$  and $\bv(\elfour_2)\#\bv(\elfour_1)$. It follows that $\bv(\elfour_1)\#\elfour_2$ and $\bv(\elfour_2)\#\elfour_1$. By Lemma~\ref{lem:safe-elem-and-prod} we know that $(\elfour_1,\elfour_2)$ is $(q_1\times q_2)$-safe.
\qed

Next we will show that under some mild conditions safe squares are closed under vertical composition. Safe squares resemble weak pullbacks. It is straightforward to show that vertical composition of weak pullbacks gives a weak pullback. However, in the case of safe squares, some additional constraints are imposed on the elements, such as $f$-safety. Therefore, safe elements should behave well with respect to the vertical composition of the maps. We will find the following definitions handy.

\begin{defi}[Forward-safe]
  \label{def:fwd-safe-maps}
Let $f:\setone\to\settwo$ and $g:\settwo\to\setthree$ be safe maps. We say that the pair $(f,g)$ is \emph{forward-safe} if, for all $\elone\in\setone$ such that $\elone$ is $(g\circ f)$-safe, we have that $f(\elone)$ is $g$-safe.
\end{defi}

\begin{defi}[Backward-safe]
  \label{def:back-safe-maps}
Let $f:\setone\to\settwo$ and $g:\settwo\to\setthree$ be safe maps. We say that the pair $(f,g)$ is \emph{backward-safe} if, for all $\elone\in\setone$ such that $\elone$ is $f$-safe and $f(\elone)$ is $g$-safe, we have that $\elone$ is $(g\circ f)$-safe.
\end{defi}  

\begin{lem}
\label{lem:ssq-vert-comp} Consider the following diagram
\begin{equation}
\label{eq:ssq-vcomp}
\begin{gathered}
  \xymatrix{
\setone_1\ar@{<-}[r]^p\ar[d]_f & \setone_2\ar[d]^h\ar@{}[dl]|{(1)} \\
\settwo_1\ar@{<-}[r]_q\ar[d]_g & \settwo_2\ar[d]^k\ar@{}[dl]|{(2)}\\
\setthree_1\ar@{<-}[r]_s & \setthree_2.
}
\end{gathered}
\end{equation}
such that $(1)$ and $(2)$ are safe squares, the pair $(f,g)$ is forward-safe and the pair $(h,k)$ is backward-safe. Then the outer square in~\eqref{eq:ssq-vcomp} is a safe square.
\end{lem}

\proof Consider $\elone_1\in\setone_1$ and $\elthree_2\in\setthree_2$ such that $\elone_1$ is $(g\circ f)$-safe, $\bv_{gf}(\elone_1)\#\elthree_2$ and $gf(\elone_1)=s(\elthree_2)$.
Since $(f,g)$ is forward-safe we have that $f(\elone_1)$ is $g$-safe.
We also have that $\bv_g(f(\elone_1))\#\elthree_2$, because $\bv_g(f(\elone_1))\subseteq\bv_{gf}(\elone_1)$. Since $(2)$ is a safe square there exists a $k$-safe $\eltwo_2\in\settwo_2$ such that $k(\eltwo_2)=\elthree_2$ and $q(\eltwo_2)=f(\elone_1)$. 

Moreover by Lemma~\ref{lem:choosing-good-safe-el} we can assume that
\begin{equation}
  \label{eq:7}
\bv_k(\eltwo_2)\#\bv_f(\elone_1).
\end{equation}
 We can apply Lemma~\ref{lem:choosing-good-safe-el} in this case because $\bv_f(\elone_1)\#\bv_g(f(\elone_1))$.

By Lemma~\ref{lem:safe-composition-1} $\elone_1$ is $f$-safe. We also have that $\bv_f(\elone_1)\# \eltwo_2$. This holds because $\supp(\eltwo_2)=\bv_k(\eltwo_2)\uplus\supp(\elthree_2)$ and both $\bv_k(\eltwo_2)$ and $\supp(\elthree_2)$ are fresh for $\bv_f(\elone_1)$. The former is by~\eqref{eq:7} while the latter holds because $\bv_f(\elone_1)\subseteq\bv_{gf}(\elone_1)$ and $\bv_{gf}(\elone_1)\#\elthree_2$. Since $q(\eltwo_2)=f(\elone_1)$, we use that $(1)$ is a safe square to derive the existence of an $h$-safe $\elone_2\in\setone_2$ that satisfies $p(\elone_2)=\elone_1$ and $h(\elone_2)=\eltwo_2$.
Since $(h,k)$ is a backward-safe pair of maps, $\elone_2$ is $h$-safe and $h(\elone_2)=\eltwo_2$ is $k$-safe, we conclude that $\elone_2$ is also $(k\circ h)$-safe and thus satisfies all the requirements.
\qed

\begin{lem}
\label{lem:theta-special}
  If $\elone$ is $f$-safe and  $x\#\bv_{f}(\elone)$ 
   then $(x,\elone)$ is $(\theta\circ(\setvars\times f))$-safe.
\end{lem}

\proof
By Lemma \ref{lem:safe-elem-and-prod},
$(x, \elone)$ is $(\setvars\times f)$-safe.
Consider $(y, \eltwo)$ such that $\abs{x}f(\elone)=\abs{y}f(\eltwo)$.
We will show that
$|\supp(y,\eltwo)| \leq |\supp(x, \elone) |$.
We have that 
$  f(\elone) = (y \; x) \cdot f(\eltwo)$.
Then $f( \elone) = f( (y \; x) \cdot \eltwo)$
because $f$ is equivariant.
Since $(x, \elone)$ is $(\setvars\times f)$-safe,
$| \supp (x,(y \; x) \cdot \eltwo) | \leq 
| \supp (x,  \elone) |$.
We also have that
$|\supp(y,\eltwo)|  = | (y \; x) \cdot \supp(y, \eltwo) | =
| \supp (x, (y \; x) \cdot \eltwo) | \leq
|\supp(x, \elone) |
$\qed
\begin{lem}
  \label{lem:fwd-safe-closure-prop}
Back- and forward-safe pairs have the following closure properties:
\begin{enumerate}
\item\label{item1-lem-fwd} If $(f_1,g_1)$ and $(f_2,g_2)$ are forward-safe (backward-safe) pairs of maps then  $(f_1\times f_2,g_1\times g_2)$ is forward-safe (backward-safe).
\item\label{item2-lem-fwd} If $(f_i,g_i)$ are forward-safe (backward-safe) then $(\coprod f_i,\coprod g_i)$ is forward-safe (backward-safe).
\item\label{item3-lem-fwd} If $(f,g)$ is a forward-safe (backward-safe) pair of maps then $(\setvars\times f, \theta\circ(\setvars\times g))$ is forward-safe (backward-safe).
\end{enumerate}
\end{lem}
\proof\hfill
\begin{enumerate}
\item Let us show first that $(f_1\times f_2,g_1\times g_2)$ is forward-safe. Assume $(\elone_1,\elone_2)$ is $(g_1\times g_2)\circ (f_1\times f_2)$-safe. By Lemma~\ref{lem:safe-elem-and-prod} we know that each $\elone_i$ is $g_i\circ f_i$-safe, $\bv_{g_1f_1}(\elone_1)\#\elone_2$ and $\bv_{g_2f_2}(\elone_2)\#\elone_1$. Since each $(f_i,g_i)$ is forward-safe we have that $f_i(\elone_i)$ is $g_i$-safe. Moreover since $\bv_{g_i}(f_i(\elone_i))\subseteq\bv_{g_if_i}(\elone_i)$ and $\supp(f_i(\elone_i))\subseteq\supp(\elone_i)$ we conclude that $\bv_{g_1}(f_1(\elone_1))\# f_2(\elone_2)$ and $\bv_{g_2}(f_2(\elone_2))\# f_1(\elone_1)$. Therefore we can apply again Lemma~\ref{lem:safe-elem-and-prod} to conclude that $(f_1(\elone_1),f_2(\elone_2))$ is $g_1\times g_2$-safe.

Next we show that $(f_1\times f_2,g_1\times g_2)$ is backward-safe when each $(f_i,g_i)$ is. To this end assume $(\elone_1,\elone_2)$ is $f_1\times f_2$-safe and $(f_1(\elone_1),f_2(\elone_2))$ is $g_1\times g_2$-safe. We want to show that $(\elone_1,\elone_2)$ is $(g_1\times g_2)\circ (f_1\times f_2)$-safe. By Lemma~\ref{lem:safe-elem-and-prod} we have that 
\begin{itemize}
\item $\elone_i$ is $f_i$-safe and $f_i(\elone_i)$ is $g_i$-safe,
\item $\bv_{f_1}(\elone_1)\#\elone_2$ and $\bv_{f_2}(\elone_2)\#\elone_1$,
\item $\bv_{g_1}(f_1(\elone_1))\#f_2(\elone_2)$ and $\bv_{g_2}(f_2(\elone_2))\#f_1(\elone_1)$.
\end{itemize}
Since $(f_i,g_i)$ are backward-safe the first item above implies that each $\elone_i$ is $g_i\circ f_i$-safe. The next two items imply together that $\bv_{g_1f_1}(\elone_1)\#\elone_2$ and $\bv_{g_2f_2}(\elone_2)\#\elone_1$. By Lemma~\ref{lem:safe-elem-and-prod} we conclude that $(\elone_1,\elone_2)$ is $(g_1\times g_2)\circ (f_1\times f_2)$-safe.
\item As usual the case of coproducts seems trivial.
\item We first prove the lemma for forward-safe maps. Assume $(x,\elone)$ is $(\theta_\settwo\circ(\setvars\times g)\circ(\setvars\times f))$-safe. We will show that $(x,f(\elone))$ is $(\theta_\settwo\circ(\setvars\times g))$-safe. By Lemma~\ref{lem:safe-composition-1} we have that $(x,\elone)$ is $((\setvars\times g)\circ(\setvars\times f))$-safe. By Lemma~\ref{lem:safe-elem-and-prod} we have that $\elone$ is $(g\circ f)$-safe and $x\# \bv_{gf}(\elone)$. Since $(f,g)$ is forward-safe we have that $f(\elone)$ is $g$-safe. Moreover since $\bv_g(f(\elone))\subseteq\bv_{gf}(\elone)$ we have that $x\# \bv_g(f(\elone))$. Therefore, by Lemma~\ref{lem:safe-elem-and-prod} we have that $(x,f(\elone))$ is $(\setvars\times g)$-safe. By Lemma~\ref{lem:theta-special} we know that $(x,f(\elone))$ is $(\theta_\settwo\circ(\setvars\times g))$-safe.

  Now let us prove this closure property for backward-safe maps. Assume $(x,\elone)$ is $(\setvars\times f)$-safe and $(x,f(\elone))$ is $(\theta\circ(\setvars\times g))$-safe. We want to show that $(x,\elone)$ is $(\theta\circ(\setvars\times g)\circ(\setvars\times f))$-safe. By Lemma~\ref{lem:safe-elem-and-prod} we know that $\elone$ is $f$-safe and $x\#\bv_f(\elone)$. By Lemma~\ref{lem:safe-composition-1} we know that $(x,f(\elone))$ is $(\setvars\times g)$-safe, thus by Lemma~\ref{lem:safe-elem-and-prod} we know that $f(\elone)$ is $g$-safe and $x\# \bv_g(f(\elone))$. Since $(f,g)$ is backward-safe we have that $\elone$ is $(g\circ f)$-safe. We can also check that $x\#\bv_{gf}(\elone)$. Thus $(x,\elone)$ is $(\setvars\times(g\circ f))$-safe. Applying Lemma~\ref{lem:theta-special} we get that $(x,\elone)$ is $(\theta\circ(\setvars\times(g\circ f)))$-safe. 
\qed
\end{enumerate}

\begin{lem}
\label{lem:our-maps-fwd-back-safe}
  For every safe map $f:\setone\to\settwo$ the pair of maps $(\fnom{F}(f),q_{\settwo})$ is both backward-safe and forward-safe.
\end{lem}
\proof The proof is by induction on the grammar of 
 $\fnom{F}_{\alpha}$ 
 and Lemma~\ref{lem:fwd-safe-closure-prop}.
\qed

\begin{lem}
\label{lem:ssq-close-abs}
If 
\begin{equation}
  \label{eq:squ-abs1}
\begin{gathered}
  \xymatrix{
\setone \ar[d]_-{f}\ar@{<-}[r]^-{p} & \setfour \ar[d]^-{q} \\
\setthree \ar@{<-}[r]_{g} & \settwo
}
\end{gathered}
\end{equation}
is a safe square then
\begin{equation}
  \label{eq:squ-abs2}
\begin{gathered}
  \xymatrix{
\setvars \times \setone \ar[d]_-{ \theta \circ (\setvars \times f)}
\ar@{<-}[r]^-{\setvars \times p} & 
\setvars \times \setfour \ar[d]^-{ \theta \circ (\setvars \times q)} \\
[\setvars]\setthree \ar@{<-}[r]_{[\setvars]g} & [\setvars]\settwo
}
\end{gathered}
\end{equation}
is a safe square. 

\end{lem}

\proof
We apply Lemma \ref{lem:ssq-vert-comp}.
Note that the pairs of maps 
$(\setvars \times f, \theta)$ and $(\setvars \times q, \theta)$
 are forward and backward-safe by Lemma \ref{lem:fwd-safe-closure-prop} (3).
\qed

\begin{lem}
\label{lem:nat-q-sq-is-safe}
  For all equivariant $f:\setone\to\settwo$ the square
\begin{equation}
\begin{gathered}
  \xymatrix{
\fnom{F}\settwo\ar@{<-}[r]\ar@{->>}[d]_{q_{\settwo}} & \fnom{F}\setone\ar@{->>}[d]^{q_{\setone}}\\
\fnom{F}_\alpha\settwo\ar@{<-}[r] & \fnom{F}_\alpha\setone\\
}
\end{gathered}
\end{equation}
is a safe square.
\end{lem}

\proof
This is proved by induction on $F_\alpha$  using Lemmas \ref{lem:ssq-close-prod-coprod} and \ref{lem:ssq-close-abs}.
\qed

\begin{prop}\label{prop:safe-squares}
  For all $n$ the squares
\begin{equation}
\begin{gathered}
  \xymatrix{
\fnom{F}^n1\ar@{<-}[r]\ar@{->>}[d]_{\amap{n}} & \fnom{F}^{n+1}1\ar@{->>}[d]^{\amap{n+1}}\\
\fnom{F}_\alpha^n1\ar@{<-}[r] & \fnom{F}_\alpha^{n+1}1\\
}
\end{gathered}
\end{equation}
are safe squares.
\end{prop}
\proof
We use induction on $n$. 
For the inductive step observe that the $(n+1)^\textit{th}$ square is the composition
\begin{equation}
\label{eq:our-sq-comp}
\begin{gathered}
  \xymatrix{
\fnom{F}\fnom{F}^n1\ar@{<-}[r]\ar@{->>}[d]_{\fnom{F}\amap{n}}&\fnom{F}\fnom{F}^{n+1}1\ar@{->>}[d]^{\fnom{F}\amap{n+1}}\\
\fnom{F}\fnom{F}_\alpha^n1\ar@{<-}[r]\ar[d]_{q_{\fnom{F}_\alpha^n1}} & \fnom{F}\fnom{F}_\alpha^{n+1}1\ar[d]^{q_{\fnom{F}_\alpha^{n+1}1}}\\
\fnom{F}_\alpha\fnom{F}_\alpha^n1\ar@{<-}[r] & \fnom{F}_\alpha\fnom{F}_\alpha^{n+1}1\\
}
\end{gathered}
\end{equation}

By Lemma~\ref{lem:ssq-close-prod-coprod} we know that the upper square in~\eqref{eq:our-sq-comp} is a safe square. By Lemma~\ref{lem:nat-q-sq-is-safe} the lower square in~\eqref{eq:our-sq-comp} is also safe. Moreover, by Lemma~\ref{lem:our-maps-fwd-back-safe} we know that $(\fnom{F}\amap{n},q_{\fnom{F}_\alpha^n1})$ is forward-safe and $(\fnom{F}\amap{n+1},q_{\fnom{F}_\alpha^{n+1}1})$ is backward-safe. Thus, we can apply Lemma~\ref{lem:ssq-vert-comp} to conclude that the outer square in~\eqref{eq:our-sq-comp} is a safe square.
\qed

We established the two
bullet points of Corollary~\ref{cor:isomfinalcoalgcarriers-2} in
Propositions~\ref{prop:safe-maps} and \ref{prop:safe-squares}:

\begin{thm} \label{thm:bindingsig-surj}
  Let $\fnom{F}, \fnom{F}_\alpha$ be endofunctors on $\Nom$ obtained from a binding signature and
  $q:\fnom{F}\to\fnom{F}_\alpha$ be the natural
  transformation defined in~\ref{eq:alpha-rules-gen-sig}. Let $P$ be a pullback as in
  \eqref{eq:classesbindsig-2}. Then $P\to U\fcoalg_\alpha$ is onto.
\end{thm}

\section{Applications}

In this section,  we first give a general definition of substitution  on the
final coalgebra $\fcoalg_\alpha$ of a functor $\fnom{F}_\alpha$ coming
from a binding signature. 
We, then,  apply the general results given in the previous sections
to the infinitary $\lambda$-calculus
by defining substitution and the notions of
 B\"ohm, L\'evy-Longo and Berarducci trees
on $\alpha$-equivalence classes of $\lambda$-terms.

\subsection{Substitution on an Arbitrary Coalgebraic Data Type}
\label{sec:substitution}

The following lemma \cite[Lemma 2.1]{Moss_parametriccorecursion}
allows parameters in coinductive definitions. It dualises the way in
which primitive recursion strengthens induction.  In order to express
substitution, the set $\setone$ will be used for the term
$\eqclass{N}$ in $\eqclass{M [x:=N]}$ which is not subject to
recursion and the set $\settwo$ will be used for the recursion.

\begin{lem}\label{lemma:flattening}
  Let $\finalmap: \finalcarrier \to F(\finalcarrier)$ be a final
  coalgebra and $\mapone: \carrierone \to F(\carrierone)$ an arbitrary
  $F$-coalgebra.  Then, there is a unique map $\mapthree : \carriertwo
  \to \finalcarrier$ such that for any $\maptwo:\carriertwo \to
  F(\carrierone+\carriertwo)$, the following diagram commutes:
  \[
  \xymatrix@R=18pt@C=40pt{ \settwo \ar[d]_{\maptwo}
    \ar[rr]^{\mapthree} & &
    \finalcarrier  \ar[d]^{\finalmap}  \\
    F(\carrierone+\carriertwo) \ar[rr]^{F([\mapone^*,\mapthree])} & &
    F(\finalcarrier) }
  \]
  where $\mapone^*: \carrierone \to \finalcarrier$ is the unique
  homomorphism between $(\carrierone,\mapone)$ and
  $(\finalcarrier,\finalmap)$.
\end{lem}

We apply Lemma~\ref{lemma:flattening} to define substitution on the
final coalgebra $\fcoalg_\alpha$ of a functor $\fnom{F}_\alpha$ coming
from a binding signature. 

\begin{defi}[Substitution on $\alpha$-equivalence classes of 
infinitary terms coming from a binding signature]
\label{def:subsa}
Substitution on $\fcoalg_\alpha$ is defined as the unique map such that the diagram below commutes:
$$
\xymatrix@C=50pt{
 {\fcoalg_\alpha} \times \ \setvars \times \fcoalg_\alpha \ar[d]_{\mapsubsl} \ar@{-->}[rr]^<(0.4){\subsa} & & 
 {\fcoalg_\alpha}  \ar[d]^{\unfolda}  \\
 \fnom{F}_{\alpha} (\fcoalg_\alpha + \fcoalg_\alpha \times \ \setvars \times \fcoalg_\alpha  ) 
 \ar[rr]^<(0.3){\quad\quad\quad\fnom{F}_{\alpha}([\mathit{id},\subsa])} &  &
  \fnom{F}_{\alpha}(\fcoalg_\alpha) 
}
$$
where $\mapsubsl$ is defined in~\eqref{eq:h-case1},~\eqref{eq:h-case2} and~\eqref{eq:h-case3}.
\end{defi}

Intuitively $\subsa ([t]_{\alpha}, x, [s]_{\alpha})$  is given by $[t [x:= s]]_{\alpha}$. However, as explained in the introduction, substitution cannot be formally defined as a total function on raw terms: additional freshness side conditions are required. See also the explanation following~\eqref{eq:subs-alpha-sem}.

In order to define the equivariant map $\mapsubsl$ we can use the properties of the functor $\fnom{F}_\alpha$. As observed in Proposition~\ref{prop:generalizedgrammar} such
functors are $\Nom$-enriched, or equivalently (see~\cite{kock:strong-funct})
 \emph{strong}. That is, there exists a natural transformation
$$\tau_{X,Y}:\fnom{F}_\alpha X\times Y\to \fnom{F}_\alpha(X\times Y).$$
Each functor obtained from the grammar in~\eqref{eq:gramF} can be equipped with a strength. Most constructions are standard and if two functors are strong so is their
 composition, product or coproduct. The only interesting case is that of the abstraction functor.
We  define a strength using the concretion of Definition~\ref{def:concretion}. Explicitly $\tau_{\setone,\settwo}:[\setvars]\setone\times\settwo\to[\setvars](\setone\times\settwo)$ is defined by
$$(\abs{x}{\elone},\eltwo)\mapsto \abs{y}(\abs{x}{\elone}\conc y,\eltwo) $$
where $y$ is some/any fresh variable for $x,\elone,\eltwo$. By construction, $\tau$ is a well-defined and natural in both $\setone$ and $\settwo$.

The map $\mapsubsl$ is defined as follows. 

\begin{itemize}
  \item For tuples of the form $(x,x, N)$ we define $\mapsubsl(x,x, N)$ as the composite
    \begin{equation}
      \label{eq:h-case1}
\xymatrixcolsep{3pc}
\xymatrix
{
{\fcoalg_\alpha} \times \ \setvars \times \fcoalg_\alpha\ar[r]^-{\pi_3} & \fcoalg_\alpha\ar[r]^-{\unfolda} &
\fnom{F}_{\alpha}(\fcoalg_\alpha) \ar[r]^-{\fnom{F}_{\alpha}(\sf{inl})} & 
\fnom{F}_{\alpha}(\fcoalg_\alpha +\fcoalg_\alpha \times \ \setvars \times \fcoalg_\alpha)
}
    \end{equation}

  \item For tuples of the form $(y,x, N)$ with $x\neq y$ or for tuples of the form $(k,x,N)$ where $k$ is a constant, we define $\mapsubsl(y,x, N)$ as the composite

    \begin{equation}
      \label{eq:h-case2}
\xymatrixcolsep{3pc}
\xymatrix
{
{\fcoalg_\alpha} \times \ \setvars \times \fcoalg_\alpha\ar[r]^-{\pi_1} & \fcoalg_\alpha\ar[r]^-{\unfolda} &
\fnom{F}_{\alpha}(\fcoalg_\alpha) \ar[r]^-{\fnom{F}_{\alpha}(\sf{inl})} & 
\fnom{F}_{\alpha}(\fcoalg_\alpha +\fcoalg_\alpha \times \ \setvars \times \fcoalg_\alpha)
}
    \end{equation}

  \item For tuples $(M,x,N)$ such that $M$ is not a variable we define $\mapsubsl(M,x,N)$ as the composite:
    \begin{equation}
      \label{eq:h-case3}
\xymatrixcolsep{4pc}
\xymatrix
{
{\fcoalg_\alpha} \times \ \setvars \times \fcoalg_\alpha\ar[d]_-{\unfolda\times\setvars\times\fcoalg_\alpha}\ar@{-->}[r]^-{\mapsubsl}
&
\fnom{F}_{\alpha}(\fcoalg_\alpha +\fcoalg_\alpha \times \ \setvars \times \fcoalg_\alpha)
 \\ \fnom{F}_{\alpha}(\fcoalg_\alpha) \times \ \setvars \times \fcoalg_\alpha\ar[r]^-{\tau_{{\fcoalg_\alpha},\setvars \times \fcoalg_\alpha}}& 
 \fnom{F}_{\alpha}(\fcoalg_\alpha \times \ \setvars \times \fcoalg_\alpha)\ar[u]_-{\fnom{F}_{\alpha}(\sf{inr})} \\ 
}
    \end{equation}
 \end{itemize}
 
Since $\{(x,x,N)\ |\ x\in\setvars\}$, $\{(y,x,N)\ |\ x\neq y\in\setvars\}$, $\{(k,x,N)\ |\ k\textrm{ constant}\}$ and $\{(M,x,N)\ |\ M\not\in\setvars\}$ are nominal sets that form a partition of $\fcoalg_\alpha\times\setvars\times\fcoalg_\alpha$ the map $\mapsubsl$ is well-defined and equivariant, thus we can apply Lemma~\ref{lemma:flattening} to prove the existence of a unique substitution map $\subsa$.

\subsection{Substitution on $\alpha$-Equivalence Classes of Infinitary $\lambda$-Terms}

As an example we spell out the concrete calculations for substitution on the nominal set 
$\Lboti$ of finite and infinite $\lambda$-terms with $\bot$. 
The set $\Lboti$ is defined as the final coalgebra of the functor $\TLambdabot$ defined by:
\begin{equation}
  \TLambdabot \ \setone =\mathcal{V}+\{\bot\}+\mathcal{V}\times \setone+\setone \times \setone.
\end{equation}
Adding the extra constant $\bot$ is needed  in order to write  corecursive functions that compute the B\"ohm, L\'evy-Longo and Berarducci trees.
We also  consider the functor

\begin{equation}\label{eq:LNombot}
\LNombot \ \setone = \setvars+\{\bot\}+[\setvars]\setone+ \setone\times \setone.
\end{equation}

\begin{notation}\label{notation:lbifa}
    We write $\Lbifa$ for the final coalgebra of $\LNombot$, omitting the $\bot$ in the notation to improve readability.
    We continue to denote
    terms in  $\Lbotif$ by $M,N$, but will denote terms in
    $\Lbifa$ by $\eqclass{M}, \eqclass{N}$. By Corollary~\ref{cor:isomfinalcoalgcarriers} we have that $\Lbifa$ is isomorphic to $\Lbotif/{=_\alpha}$.

The injections for the coproduct $\setone +\settwo$
are denoted as 
 $\inl^{\setone, \settwo}: \setone \to \setone + \settwo$
and  $\inr^{\setone, \settwo}: \settwo \to \setone+ \settwo$.
But for the case of 
$\FLambda(\setone)$, we denote them as

\[
\begin{array}{lll}
\inbota^{\setone}: \{ \bot \} \to \LNombot(\setone) \\ 
\invara^{\setone}: \setvars \to \LNombot(\setone) \\ 
\inabsa^{\setone}: [\setvars]\setone \to \LNombot(\setone)  \\
\inappa^{\setone}: \setone \times \setone \to \LNombot(\setone).
\end{array}
\]
We drop the superscripts when 
they are clear from the context.
  
  Since $\unfolda$ is an isomorphism that partitions its domain
  $\Lbifa$ into four disjoint components, see (\ref{eq:FLambda}), we
  write typical elements of $\Lbifa$ as $\eqclass{x}$, $\bot$,
  $\eqclass{M_1 M_2}$, $\lambda{\eqclass{y}}.{\eqclass{M}}$ where
 \begin{equation}\label{eq:elements-Lbifa}
  \begin{array}{ll}
  \eqclass{x} & =\unfolda^{-1}(\invara \ \eqclass{x}) \\
  \eqclass{\bot} & =\unfolda^{-1}(\inbota \ \eqclass{\bot}) \\ 
  \eqclass{\lambda y.M} & =\unfolda^{-1}(\inabsa \ \abs{\eqclass{y}}{\eqclass{ M}}) \\
  \eqclass{M_1 M_2} & =\unfolda^{-1}(\inappa \ (\eqclass{M_1},\eqclass{M_2})). \\
\end{array}
\end{equation}
We use $\eqclass{x}$ to denote both an element in $\setvars$ and also
its copy in $\Lbifa$.
\end{notation}

\begin{exa}[Substitution on $\alpha$-equivalence classes of infinitary
 $\lambda$-terms with $\bot$]
\label{exa:applications}
By instantiating the definition of the map $\mapsubsl$ given in~\eqref{eq:h-case1},~\eqref{eq:h-case2} and~\eqref{eq:h-case3} for the functor $\LNombot$ we obtain
$\mapsubsl:  \Lbifa \times \ \setvars \ \times \ \Lbifa
\  \ \to  
 \LNombot(\Lbifa + \ \Lbifa \times \ \setvars  \  \times \Lbifa)
$
  given by
\[
\begin{array}{llll}
\mapsubsl(\eqclass{x},\eqclass{x},\eqclass{N}) & =
\unfolda(\eqclass{N})\\
\mapsubsl(\eqclass{y}, \eqclass{x},\eqclass{N}) & = 
\invara \  \eqclass{y} & \mbox{\textit{if} $\eqclass{y} \not = \eqclass{x}$} \\
\mapsubsl(\bot, \eqclass{x},\eqclass{N}) & = 
\inbota \  \bot & \\
\mapsubsl(\eqclass{M_1 M_2}, \eqclass{x},\eqclass{N}) & = 
\inappa \ ( (\eqclass{M_1}, \eqclass{x},\eqclass{N}), 
 (\eqclass{M_2}, \eqclass{x},\eqclass{N}))\\
\mapsubsl(\eqclass{\lambda y.M},\eqclass{x},\eqclass{N}) &= 
\inabsa \ \abs{\eqclass{z}}{  
( (\abs{\eqclass{y}}{\eqclass{M}}) \conc \eqclass{z} ,\eqclass{x} ,\eqclass{N}) }  & 
\mbox{\textit{if} }  \eqclass{z} \fresh (\eqclass{\lambda y.M,x,N}).
\end{array}
\]
To improve readability we omitted $\LNombot \inr$ or $\LNombot \inl$ in the definition of $\mapsubsl$. 
We obtain
the substitution function
$$\subsa: \Lbifa \times \ \setvars \times \Lbifa \ \to  \ \Lbifa $$ 
given by
  \begin{equation}
\label{eq:subs-alpha-sem}
\begin{array}{llll}
\subsa(\eqclass{x,x,N}) &= \eqclass{N}\\
\subsa(\eqclass{x,y,N}) &= \eqclass{x} & \textit{if } \eqclass{y}\not= \eqclass{x}\\
\subsa(\bot,\eqclass{y},\eqclass{N}) &= \bot\\
\subsa(\eqclass{M_1M_2,x,N}) &= \subsa(\eqclass{M_1,x,N)} \ \subsa(\eqclass{M_2,x,N})\\
\subsa(\eqclass{\lambda y.M,x,N)} &= \eqclass{\lambda y. \subsa(M ,x,N)}
& 
\mbox{\textit{if} }  \eqclass{y} \fresh (\eqclass{x,N}).
\end{array}
\end{equation}
It should be pointed out that $\subsa$ defined in~\eqref{eq:subs-alpha-sem} is the `semantic' version of the substitution map since it is defined on the (or any) final $\LNombot$-coalgebra. It is Corollary~\ref{cor:isomfinalcoalgcarriers} that allows us to identify the element $\eqclass{\lambda y.M}\in\Lbifa$ with the $\alpha$-equivalence class of an infinitary term $[\lambda y.M]_\alpha$ having finitely many free variables. Thus we obtain a `syntactic' version of the substitution map which looks indeed just like a notational variant of the $\Set$-based
(\ref{definition:informal:substitution}), but is now fully justified
as a coinductive definition on $\alpha$-equivalence classes of
$\lambda$-terms.
\end{exa}

We can now define $\beta$-reduction using $\subsa$.

 \begin{defi}[$\beta$-reduction on $\alpha$-equivalence classes]
 \label{def:betaonalpha}
  We define $\betaa$-reduction 
as the smallest  relation on $\Lbifa \times \Lbifa$
that satisfies
\[ \begin{array}{cc}
 \ProofRule{}{\eqclass{(\lambda x.P) Q} \onebetaa
 \subsa(\eqclass{P}, \eqclass{x}, \eqclass{Q})}{\betaa}\quad
 &
 \ProofRule{\eqclass{P} \onebetaa \eqclass{P'}}
 {\eqclass{\lambda x. P} \onebetaa \eqclass{\lambda x. P'}}{abs}
  \\ \\
 \ProofRule{\eqclass{P} \onebetaa \eqclass{P'} }
 {\eqclass{PQ} \onebetaa \eqclass{P' Q}}{app_L} 
 &  \ProofRule{
  \eqclass{Q} \onebetaa \eqclass{Q'}}
 {\eqclass{PQ} \onebetaa \eqclass{P Q'}}{app_R} 
\end{array}
\]
\end{defi}

We  define the notion of $\beta$-head reduction
which  contracts only the redex at the head position
and corresponds to the normalising leftmost strategy.
This reduction is used to define B\"ohm trees.

\begin{defi}[Head $\beta$-reduction on $\alpha$-equivalence classes]
 We define $\betaheada$-reduction 
as the smallest relation on $\Lbifa \times \Lbifa$ closed under
\[ \begin{array}{cc}
 \ProofRule{}{ \eqclass{(\lambda x.P) Q} 
 \onebetaheada \subsa(\eqclass{P},\eqclass{x}, \eqclass{Q})}{\betaheada} &
 \ProofRule{\eqclass{P} \onebetaheada \eqclass{P'}}
 {\eqclass{\lambda x. P} \onebetaheada \eqclass{\lambda x. P'}}{abs}
 \\\\
 \ProofRule{\eqclass{P} \onebetaheada \eqclass{P'} \ \ \ \ 
  \eqclass{P} \mbox{ is not an abstraction} }
 {\eqclass{PQ} \onebetaheada \eqclass{P' Q}}{app_L} 
 & 
\end{array}
\]
A term $\eqclass{M}$ is in head normal form (hnf) if it is of the form
 $\eqclass{\lambda x_1 \ldots x_n. y N_1 \ldots N_m}$. 
\end{defi}

We restrict  the $\beta$-head reduction 
by not contracting $\beta$-redexes
in the body of an abstraction and 
obtain the weak head $\beta$-reduction which is needed
to define the notion of L\'evy-Longo tree.

\begin{defi}[Weak head $\beta$-reduction on $\alpha$-equivalence classes]
 We define $\betaweakheada$-reduction 
as the smallest  relation on $\Lbifa \times \Lbifa$
closed under
\[ \begin{array}{cc}
 \ProofRule{}{\eqclass{(\lambda x.P)Q} 
 \onebetaweakheada \subsa(\eqclass{P},\eqclass{x}, \eqclass{Q})}
        {\betaweakheada}  \ \  \ \ &
 \ProofRule{\eqclass{P} \onebetaheada \eqclass{P'} }
 {\eqclass{PQ} \onebetaweakheada \eqclass{P' Q}}{app_L} 
 \end{array}
\]
A term $\eqclass{M}$ is in weak head normal form (whnf) if
it is either a head normal form or an abstraction. 
\end{defi}

We now define the  notion of top $\beta$-reduction
which only contracts $\beta$-weak head redexes
at depth $0$ and  
it will be used to define Berarducci trees.

\begin{defi}[Top $\beta$-reduction on $\alpha$-equivalence classes]
 We define $\betatopa$-reduction 
as the smallest  relation on $\Lbifa \times \Lbifa$ 
closed under
\[
\ProofRule{\eqclass{M} \finbetaweakheada
 \eqclass{(\lambda x.P)}}
{\eqclass{M Q} \onebetatopa
\subsa(\eqclass{P}, \eqclass{x}, \eqclass{Q})} {\betatopa}
\]
A term $\eqclass{M}$ is a top normal form (tnf) if it is
either a weak head normal form or an application
of the form $\eqclass{N P}$ where $\eqclass{N}$ cannot reduce to an
abstraction.
\end{defi}

The reflexive, transitive closures of 
$\onebetaheada$, 
$\onebetaweakheada$ and 
$\onebetatopa$
are denoted by 
$\finbetaheada$,
$\finbetaweakheada$ and
$\finbetatopa$,
respectively. The corresponding normal forms, 
head normal form (hnf),
weak head normal form (whnf) and 
top normal form (tnf),
should they exist are unique.

\subsection{Computing the Infinite Normal Form of $\alpha$-Equivalence Classes
of $\lambda$-Terms}
We now define the notions of B\"ohm tree, L\'evy-Longo tree, and Berarducci
tree using the finality of $\unfolda: \Lbifa \to \LNombot(\Lbifa)$.

\begin{defi}[B\"ohm tree on $\alpha$-equivalence classes]
\label{def:bohmtreeonalpha}
We define the B\"ohm tree of $\eqclass{M}$
as $\bohmtreea (\eqclass{M})$ where 
$\bohmtreea$
is the unique map such that
\[
\xymatrix{
\Lbifa \ar[d]_{\gbohmtreea} \ar[rr]^{\bohmtreea} & & 
\Lbifa  \ar[d]^{\unfolda}  \\
 \LNombot (\Lbifa ) \ar[rr]_{\LNombot(\bohmtreea)} &  &
  \LNombot(\Lbifa) 
}
\]
commutes, with $\gbohmtreea: \Lbifa \to \LNombot (\Lbifa)$
being defined as 
\[ \begin{array}{ll}
\gbohmtreea(\eqclass{M}) = &
 \left \{ \begin{array}{ll}  
   \unfolda (\eqclass{N})   &\quad\textit{if } \eqclass{M} \finbetaheada 
   \eqclass{N} \textit{ and $ \eqclass{N}$ is in hnf},\\    
 \inbota \
    \eqclass{\bot} &\quad\textit{otherwise}.
      \end{array} \right .
\end{array}
\]
\end{defi}

\begin{defi}[L\'evy-Longo tree on $\alpha$-equivalence classes]
\label{def:levylongoonalpha}
We define the L\'evy-Longo tree of $\eqclass{M}$
as $\lltreea(\eqclass{M})$ where 
$\lltreea$ is the unique map
such that 
\[
\xymatrix{
\Lbifa \ar[d]_{\glltreea} \ar[rr]^{\lltreea} & & 
\Lbifa  \ar[d]^{\unfolda}  \\
 \LNombot (\Lbifa ) \ar[rr]_{\LNombot(\lltreea)} &  &
  \LNombot(\Lbifa) 
}
\]
commutes, with $\glltreea: \Lbifa \to \LNombot (\Lbifa)$ being defined as
\[ \begin{array}{ll}
\glltreea(\eqclass{M}) = &
 \left \{ \begin{array}{ll} 
    \unfolda (\eqclass{N})    &\quad \textit{if } \eqclass{M} \finbetaweakheada \eqclass{N} \textit{ and $ \eqclass{N}$ is in whnf}, \\
    \inbota \ \bot  &\quad\textit{otherwise}. \\
 \end{array} \right .
\end{array}
\]
\end{defi}

\begin{defi}[Berarducci tree on $\alpha$-equivalence classes]
\label{def:berarduccionalpha}
  We define the Berarducci tree of $\eqclass{M}$ as $\bertreea
  (\eqclass{M})$ where $\bertreea$ is the unique
  map such that
\[
\xymatrix{
\Lbifa \ar[d]_{\gbertreea} \ar[rr]^{\bertreea} & & 
\Lbifa  \ar[d]^{\unfolda}  \\
 \LNombot (\Lbifa ) \ar[rr]_{\LNombot(\bertreea)} &  &
  \LNombot(\Lbifa) 
}
\]
commutes, with $\gbertreea: \Lbifa \to \LNombot (\Lbifa)$ being defined
as follows. 
\[ \begin{array}{ll}
\gbertreea(\eqclass{M}) = &
 \left \{ \begin{array}{ll}
  \unfolda (\eqclass{N})    &\quad \textit{if } \eqclass{M} \finbetatopa \eqclass{N} \textit{ and $ \eqclass{N}$ is in tnf}, \\
  \inbota \ \bot  &\quad  \textit{otherwise}. \\
 \end{array} \right .
\end{array}
\]

\end{defi}

Similar to the remark on the substitution map
mentioned at the end of
 Example \ref{exa:applications}, the maps 
 $\bohmtreea$, $ \lltreea$ and $ \bertreea$ 
 are the `semantic' counterparts 
  of   the maps that compute the B\"ohm, L\'evy-Longo
  and Berarducci trees.
 Corollary~\ref{cor:isomfinalcoalgcarriers} 
  allows us to
   obtain a `syntactic' version of these maps
    which look  like a notational variant of the $\Set$-based 
\eqref{definition:informal:bohmtree}, 
\eqref{definition:informal:levylongotree}
and \eqref{definition:informal:berarduccitree}, 
 but are now fully justified
as  coinductive definitions on $\alpha$-equivalence classes of
$\lambda$-terms.

\subsection{Nominal Sets of Infinite Normal Forms and Bisimulations} 
Recall the \emph{informal} corecursive definitions for the
 sets of B\"ohm, L\'evy-Longo and Berarducci Trees
 given in Definitions~\ref{informal:setbohmtrees},~\ref{informal:setoflltrees} and~\ref{informal:setofbertrees}.
 In this section, we formally define the first two  ones  on $\alpha$-equivalence classes
 as an application of  Corollary \ref{cor:isomfinalcoalgcarriers}.
 In order  to define the set of Berarducci trees, one needs to extend
  the notion of binding signature to include infinite products, which is beyond the scope of our paper.

\begin{exa}[B\"ohm trees up to $\alpha$-equivalence]
We define the functor $\funsetBTa$  by
\[
\begin{array}{lll}
\funsetBTa(X) & = \{ \bot \} + \setvars  + \coprod_{k} [\setvars]^{k}
         (\setvars \times {\sf List} (X)).  \\
\end{array}
\]
The final coalgebra of $\funsetBTa$, denoted by  $\setBTa$,
exists and it is isomorphic to 
the set 
$\setBTffv/\mathord{\alphaconv}$ of $\alpha$-equivalence classes of 
B\"ohm trees with finitely many free variables  by Corollary \ref{cor:isomfinalcoalgcarriers}
(see also \eqref{definition:informal:bohmtree}).

The nominal set $\Lbifa$ can be equipped with a $\funsetBTa$-coalgebra structure. Explicitly, consider $\xi_{\sf BT}:\Lbifa\to\funsetBTa(\Lbifa)$ defined by
\[ \begin{array}{ll}
\xi_{\sf BT}(\eqclass{M}) = &
 \left \{ \begin{array}{ll}
  \eqclass{x}  &\quad \textit{if } \eqclass{M} \finbetaheada \eqclass{x}, \\
  \abs{\eqclass{x_1}}\ldots\abs{\eqclass{x_n}}(\eqclass{x},\eqclass{M_1},\ldots,\eqclass{M_n})    &\quad \textit{if } \eqclass{M} \finbetaheada  \eqclass{\lambda x_1 \ldots\lambda x_n. x M_1 \ldots M_m},\\
    \bot  &\quad  \textit{otherwise}. \\
 \end{array} \right .
\end{array}
\]
Since $\setBTa$ is the final $\funsetBTa$-coalgebra, we have a unique morphism $\bohmtreea:\Lbifa\to\setBTa$ such that the diagram below commutes. Notice that this morphism is obtained by restricting the codomain of the map $\bohmtreea$ from  Definition~\ref{def:bohmtreeonalpha}.
\[
\xymatrix@C=40pt{
\Lbifa \ar[r]^-{\bohmtreea}\ar[d]_-{\xi_{\sf BT}} & \setBTa\ar[d]^-{\simeq}\\
\funsetBTa(\Lbifa)\ar[r]_-{\funsetBTa(\bohmtreea)} &\funsetBTa(\setBTa)
}
\]
We can now define the head bisimulation $\hbis$ on $\Lbifa$ as the kernel pair of the map $\bohmtreea$, that is, $\eqclass{M} \hbis \eqclass{N}$ 
if and only if $\eqclass{M}$ and $\eqclass{N}$ have the same  B\" ohm tree.
Explicitly, $\hbis$ is defined as the pullback
\[
\xymatrix@C=40pt{
\hbis\ar[r]^-{\pi_1}\ar[d]_-{\pi_2}\pullbackcorner & \Lbifa\ar[d]^-{\bohmtreea}\\
\Lbifa\ar[r]_-{\bohmtreea} & \setBTa 
}
\]
Since $[\setvars](-)$ preserves pullbacks and pullbacks commute with coproducts and limits in $\Nom$, we have that $\funsetBTa$ preserves pullbacks, thus the outer square of the diagram
\[
\xymatrix@C=40pt{
\funsetBTa(\hbis)\ar[rrr]^-{\funsetBTa(\pi_1)}\ar[ddd]_-{\funsetBTa(\pi_2)} & & & \funsetBTa(\Lbifa)\ar[ddd]^-{\funsetBTa(\bohmtreea)}\\
&\hbis\ar[r]^-{\pi_1}\ar[d]_-{\pi_2}\ar@{-->}[lu]\pullbackcorner & \Lbifa\ar[ru]_-{\xi_{\sf BT}}\ar[d]^-{\bohmtreea} &\\
& \Lbifa\ar[ld]^-{\xi_{\sf BT}}\ar[r]_-{\bohmtreea} & \setBTa\ar[rd] & \\
\funsetBTa(\Lbifa)\ar[rrr]_-{\funsetBTa(\bohmtreea)} & & & \funsetBTa(\bohmtreea)
}
\]
is also a pullback. Therefore we obtain an $\funsetBTa$-coalgebra structure on $\hbis$ such that the diagram
\begin{equation}
  \begin{gathered}
    \label{head-bis-2}
    \xymatrix@C=40pt{ \Lbifa \ar[d]_{\xi_{\sf BT}} & & \hbis \ar[d] 
      \ar[ll]_{\pi_1} \ar[rr]^{\pi_2} & &
      \ar[d]^{\xi_{\sf BT}}  \Lbifa \\
      \funsetBTa(\Lbifa) & &\funsetBTa(\hbis)
      \ar[ll]^{\funsetBTa(\pi_1)} \ar[rr]_{\funsetBTa(\pi_2)} & &
      \funsetBTa(\Lbifa) }
  \end{gathered}
\end{equation}
commutes. This shows that $\hbis$ is a bisimulation in the sense
of~\cite{AM89}. The commutativity of~\eqref{head-bis-2} means that
$\hbis$ is a binary relation on $\Lbifa$ such that for all
$\eqclass{M}$ and $\eqclass{N}$, if $\eqclass{M} \hbis \eqclass{N}$
and $\eqclass{M} \finbetaheada \eqclass{\lambda x_1 \ldots \lambda
  x_n. x M_1 \ldots M_m} $, then there are $\eqclass{N_1}, \ldots,
\eqclass{N_m}$ such that $\eqclass{N} \finbetaheada \eqclass{\lambda
  x_1 \ldots \lambda x_n. x N_1 \ldots N_m} $ and $\eqclass{M_i} \hbis
\eqclass{N_i}$ for all $1 \leq i \leq m$.

 \end{exa}

\begin{exa}[L\'evy-Longo trees up to $\alpha$-equivalence]
We define the functor  $\funsetLLTa$ by 
  \[ \funsetLLTa(X) = \{\bot \} + [\setvars] X  + \setvars\times{\sf List} (X). \]
The final coalgebra of $\funsetLLTa$, denoted by 
$\setLLTa$, exists and is isomorphic to the set 
$\setLLTffv/\mathord{\alphaconv}$ 
by Corollary \ref{cor:isomfinalcoalgcarriers} (see also \eqref{definition:informal:levylongotree}).

The nominal set $\Lbifa$ can be equipped with an $\funsetLLTa$-coalgebra structure $\xi_{\sf LLT}:\Lbifa\to\funsetLLTa(\Lbifa)$ as follows:
\[ \begin{array}{ll}
\xi_{\sf LLT}(\eqclass{M}) = &
 \left \{ \begin{array}{ll}
  \abs{\eqclass{x}}\eqclass{N}  &\quad \textit{if } \eqclass{M} \finbetaweakheada \eqclass{\lambda x.N}, \\
  (\eqclass{x},\eqclass{M_1},\ldots,\eqclass{M_n})    &\quad \textit{if } \eqclass{M} \finbetaweakheada  \eqclass{ x M_1 \ldots M_m},\\
    \bot  &\quad  \textit{otherwise}. \\
 \end{array} \right .
\end{array}
\]

The unique map from $\xi_{\sf LLT}$ into the final $\funsetLLTa$-coalgebra is given by the restriction of the map $\lltreea$ from Definition~\ref{def:levylongoonalpha} and maps the equivalence class of an infinitary $\lambda$-term to its L\'evy-Longo tree. 

The weak head bisimulation can be defined as the kernel pair of the map $$\lltreea:\Lbifa\to\setLLTa.$$ 
Similarly to the case of head simulation, we can show that $\whbis$ is a bisimulation in the sense of~\cite{AM89}, namely we have an $\funsetLLTa$-coalgebra structure on $\whbis$ such that
\begin{equation}
\begin{gathered}
\label{weak-head-bis-2}
\xymatrix@C=40pt{
\Lbifa  \ar[d]_{\xi_{\sf LLT}} & &\whbis \ar[d] \ar[ll]_{\pi_1} \ar[rr]^{\pi_2} & &
    \ar[d]^{\xi_{\sf LLT}}  \Lbifa \\
\funsetLLTa(\Lbifa)  & &\funsetLLTa(\whbis)   \ar[ll]^-{\funsetLLTa(\pi_1)} \ar[rr]_-{\funsetLLTa(\pi_2)} & &
\funsetLLTa(\Lbifa)
}
\end{gathered}
\end{equation}
The commutativity of the above diagram means that the weak head bisimulation  $\whbis$ is a binary relation on
 $\Lbifa$ such 
 that for all $\eqclass{M}$ and $\eqclass{N}$,
 if  $\eqclass{M} \whbis \eqclass{N}$, the following
 hold:
 \begin{enumerate}
 \item 
 $\eqclass{M} \finbetaweakheada
     \eqclass{\lambda x. M } $
             then
 $\eqclass{N} \finbetaweakheada
     \eqclass{\lambda x. N} $
   and $\eqclass{M} \whbis \eqclass{N}$.
                
\item  $\eqclass{M} \finbetaweakheada  \eqclass{x M_1 \ldots M_m} $
             then 
             $\eqclass{N} \finbetaheada
     \eqclass{ x N_1 \ldots N_m} $
   and 
$\eqclass{M_i} \whbis \eqclass{N_i}$ for all $1 \leq i \leq m$.

\end{enumerate} 
\end{exa}

\begin{exa}[Berarducci trees up to $\alpha$-equivalence]
\label{example:berarducci}

By Theorem \ref{thm:fin}, 
the final coalgebra of $\funsetBerTa$ exists 
where
\[ \funsetBerTa(X) = [\setvars]  X + \{\bot \} \times {\sf List} (X) + 
\setvars  \times {\sf List} (X) + {\sf Stream} (X) \]
However, we cannot apply 
Corollary \ref{cor:isomfinalcoalgcarriers}
in this case because our current definition of binding signatures does not include
infinite products, e.g. $\sf Stream$.

\end{exa}

\subsection{Infinitely Many Free Variables}
\label{section:infinitelymanyfreevariables}

\newcommand{\Tr}{{\sf Tr}}

In this section we follow a suggestion from Pitts to
treat  the case of terms with infinitely many free variables.
In Theorem \ref{theorem:metricompletionofLba}
we proved that  $(\Lbia,\dab)$ is  isomorphic to   $\clba$
provided that ${\mathcal V}$ is uncountable.
However, the isomorphism does not hold if 
${\mathcal V}$ is  countable (Example \ref{ex:alpha-inf-lam}).
Moreover, even if ${\mathcal V}$ is  countable
and  the set $\clba$ can be equipped with a permutation action, 
$\clba$ is not a nominal set, 
since the terms with infinitely many free variables 
are not finitely supported. 
One way to deal with terms with infinitely many free variables,
without leaving  the world of nominal sets,  is to extend 
the calculus with constants and to regard the free variables as constants.

Let $\Const$ be a countable set of names.
The sets $\Lc$  and $\Lbic$ are the sets $\L$ and $\Lib$ extended with 
a set of constants $\Const$.
We can equip the set $\Const$ with the trivial permutation action and consider the functor $\LNom+\Const:\Nom\to\Nom$. 
The following proposition follows from 
\cite[Theorem 5.12]{PittsMGS2011notes}, or  \cite[Theorem 5.1, Remark 5.3]{PittsAM:alpsri}.

\begin{prop}
  The initial algebra of $\LNom+\Const$ is the nominal set 
  $\Lbaonec$ of $\lambda$-terms extended with $\Const$-constants up 
  to $\alpha$-equivalence. 
\end{prop}

The following proposition follows from Corollary \ref{cor:isomfinalcoalgcarriers}
 since $\LNom+\Const$ is a functor obtained from a binding signature 
(see Definition \ref{definition:bindingsignature} and
 Proposition \ref{prop:generalizedgrammar}).

\begin{prop}\label{prop:cons-cor-bot}
  The final coalgebra of $\LNom+\Const$ is isomorphic to the nominal sets \\
   $\Lcai_{\sf fs}$ and  $\Lbifaonec$. 
\end{prop}

\begin{prop} Let   $\rho:\setvars\to \Const$ be a bijection. 
  We have an isomorphism $\Tr$ between $\Lbaone$ and the set of closed terms in $(\Lbaonec)^0$.
\end{prop}

\begin{proof}
The map $\Tr:\L \to \Lc^0$ is defined inductively as follows:
\[
\begin{array}{lcl}
\Tr(x)& =&   \rho(x)\\
\Tr(MN) & =& \Tr(M)\Tr(N)\\
\Tr(\lambda x. M) &= &\lambda x. \Tr(M) [x/\rho(x)]
\end{array}
\]
where $\Tr(M) [x/\rho(x)]$ is the result of replacing $\rho(x)$ by $x$
in $\Tr(M)$.

It is easy to check that
if $M \alphaconv N$ then $\Tr(M) \alphaconv \Tr (N)$.
Hence, $\Tr$ can be defined on equivalence classes and
we have $\Tr:\Lbaone\to(\Lbaonec)^0$. 
Note that the map $\Tr$ is not finitely supported. For example,
$\Tr( ( x y) \cdot x) \neq ( x y) \cdot \Tr(  x) $
if $\rho (x) \not = \rho(y)$.
Hence, we cannot apply Pitts' alpha structural induction principle.

The inverse of $\Tr$ is defined as follows. 
Given an $\alpha$-equivalence class 
$[M]_\alpha\in(\Lbaonec)^0$, we can find a representative 
$M\in\Lambda(\Const)^0$ such that  
the set of bound variables of $M$ is disjoint from
the image under $\rho^{-1}$ of the constants occurring in $M$.
We put $\Tr^{-1}(M)$ to be the equivalence class of the term obtained by replacing each constant $b$ occurring in $M$ by $\rho^{-1}(b)$.
\end{proof}

\begin{prop} \label{prop:infinitelymanyfreevariables}
  The set $\clba$ is isomorphic to  $(\Lbifaonec)^0$.
\end{prop}

\begin{proof}
By the universal property of the Cauchy completion, $\Tr$ is extended to an isomorphism between $\clba$ and the completion of  $(\Lbaonec)^0$. The proof is complete by observing that the
isomorphism between $\Lcai_{\sf fs}$ and  $\Lbifaonec$ (see Proposition~\ref{prop:cons-cor-bot}) cuts down to an isomorphism between 
 the completion of  $(\Lbaonec)^0$ and $(\Lbifaonec)^0$. 
Explicitly, given a Cauchy sequence in $\Lbaone$, the image under $\Tr$ is a \emph{finitely supported} Cauchy sequence in 
$(\Lbaonec)^0$, and thus also in $\Lbaonec$. Thus it converges to a unique element of $(\Lbifaonec)^0$. 
Conversely, given $M\in(\Lbifaonec)^0$, we consider the truncations $M^n$. Their translations $\Tr^{-1}(M^n)$ in $\Lbaone$ form a Cauchy sequence and we map $M$ to its limit.
\end{proof}

Substitution is defined on $\Lbifaonec$ by instantiating Definition \ref{def:subsa}. 
Now, $\beta$-reduction restricts to  $(\Lbifaonec)^0$ and  
 can be defined on $\clba$ via the  translation $\Tr$.
We illustrate how this works with an example.

\begin{exa}
Consider the term $(\lambda x_0 x_1. x_0 x_1) \Vterm$, 
where $\Vterm= x_0 (x_1 (x_2 (\ldots)))$  given in \eqref{eq:wrong-beta}.  
Suppose $\Const = \{ c_0, c_1, \ldots \}$ and $\rho(x_i) = c_i$
for all $i$.

Then,
\[\Tr( (\lambda x_0 x_1. x_0 x_1) \Vterm )=
 (\lambda x_0 x_1. x_0 x_1) {\sf allconst}  \]
 where ${\sf allconst} = c_0 (c_1 (c_2 (\ldots)))$.
 The translated term does not contain free variables, but bound variables and constants  only.
It is important to stress the fact that 
the constants represent the free variables of the original
term. 
 We can now safely perform the $\beta$-step
\begin{equation}
  \label{eq:right-beta}
  (\lambda x_0 x_1. x_0 x_1) {\sf allconst} \onebeta (\lambda x_1. x_0 x_1)
  [x_0 := {\sf allconst} ] = \lambda x_1. {\sf allconst} \ x_1.  
\end{equation}
This $\beta$-step is possible because
substitution is defined on the set $(\Lbifaonec)$.
 According to Proposition \ref{prop:infinitelymanyfreevariables}, 
 the equivalence class  $[\lambda x_1. {\sf allconst} \ x_1]_{\alpha} $
 translates back to an element in $\clba$ which is the limit of
 the Cauchy sequence $([\lambda x_n. x_0 (x_1 \ldots (x_{n-1}\ *))x_n]_{\alpha})_n $.
 Similar to Example \ref{ex:alpha-inf-lam}, it does not have any preimage 
 under $[-]_\alpha :\Lbi\rightarrow\clba$.
 
\end{exa}

\section{Related  and Future Work}

The problem of having insufficiently many fresh variables 
does not arise if we use
de Bruijn indices~\cite{debruijn,duppen}.
However, it is unclear whether using de Bruijn indices could lead to a coalgebraic treatment of 
the corecursion principle.

  It would also be interesting to investigate nominal
  coalgebraic data types with infinitely many free variables based on
  either Section~\ref{section:infinitelymanyfreevariables} or on
  variations of nominal sets allowing countable supports, see
  e.g. \cite{cheney:jsl,DowekG12}. This could have applications to the
  semantics of processes which are able to generate infinitely many
  fresh names.

Using the `same' endofunctor as \cite{gabb-pitt:lics99}, but on a
  different category, namely the category $\Set^\mathbb{F}$ of
  presheaves on finite sets, \cite{fiore:lics99} also exhibits
  finite $\lambda$-terms as an initial algebra. Roughly speaking, the difference between $\Nom$ and $\Set^\mathbb{F}$ is that the latter already comes equipped with a notion of substitution, see \cite[Section 7.3]{StatonPhD} for details. \cite{matt-uust:substitution} further develop substitution for algebraic and coalgebraic datatypes over presheaf-categories and describe the set of infinitary $\lambda$-terms as a final coalgebra. \cite{AdamekMV11} furthermore study the so-called rational fixed point, again over $\Set^\mathbb{F}\!,$ as a semantic universe for solutions of higher-order recursion schemes.

In \cite{DBLP:conf/cade/HendriksO03},
the authors define a calculus with 
an operator called {\em adbmal} to deal with
$\alpha$-conversion. This operator removes the
scope of a variable.
It will be interesting to extend 
this calculus for infinite terms. Using this operator, it  
would give an alternative approach to dealing with 
the  problem of having insufficient fresh variables.
 
Nominal Isabelle provides infrastructure for declaring nominal 
data types  and defining recursive functions over them
\cite{DBLP:journals/jar/Urban08}. 
Proposals for codata types in Isabelle are presented in
 \cite{lics2012Codatatypes}.
It will be nice to  include nominal codata types
in Isabelle 
in order to formalise the proofs of some theorems
on $\lambda$-calculus concerning B\"ohm trees and
infinitary $\lambda$-calculus.

Nominal extensions of typed $\lambda$-calculus have
been proposed in 
 \cite{DBLP:conf/popl/Pitts10,DBLP:journals/entcs/Cheney09}
 for system  T,
in \cite{DBLP:journals/corr/abs-1201-5240}
for LF and in \cite{DBLP:conf/lfmtp/WestbrookSA09}
for the 
Calculus of Inductive Constructions. 
Further research could include a study of typing
that combines nominal syntax with coinductive
data types.

The corecursion principle presented in this paper
cannot handle infinitary meta-terms 
as defined for 
Infinitary Combinatory Reduction Systems (iCRS)
\cite{DBLP:journals/iandc/KetemaS11}.
It will be interesting to prove
an $\alpha$-coinduction principle that includes 
meta-variables and meta-terms.
However, it is not straightforward to define  $\alpha$-equivalence
on them \cite{fern-gabb:nom-rewriting}.

One could also study how to extend the notion of binding signature
and  Corollary \ref{cor:isomfinalcoalgcarriers} to include
infinite products for representing the set of Berarducci trees 
up to $\alpha$ (see Example \ref{example:berarducci}). Another possible solution, suggested by a referee,  would be to use a version of Beki{\v c} lemma~\cite{LehmannS81,BackhouseBGW95, Freyd92} and to replace a nested final coalgebra by a non-nested  many-sorted one.

It will also be worthwhile to study $\alpha$-corecursion principles
for sets of infinitary terms obtained 
from alternative metrics such as the 001- and 101-metrics
\cite{KKSV97} or the metric that
captures  the infinite normal forms
of reactive programs \cite{severidevriesICFP2012}.
Also, meta-terms for Infinitary Combinatory Reduction Systems
that satisfy the finite property chain
can be defined using an alternative metric 
(\cite[page 20]{DBLP:journals/iandc/KetemaS11}).

\section*{Acknowledgements}

We are grateful to Andy Pitts for suggesting us to treat the
infinitely many free variables of a term as constants.  We would also
like to thank Christian Urban for helpful discussions. Finally, we acknowledge insightful  improvements suggested by the referees.

\bibliographystyle{alpha}

\end{document}